\documentclass{article}

\usepackage{arxiv}

\usepackage[utf8]{inputenc} 
\usepackage[T1]{fontenc}    
\usepackage{hyperref}       
\usepackage{url}            
\usepackage{booktabs}       
\usepackage{amsfonts}       
\usepackage{nicefrac}       
\usepackage{microtype}      
\usepackage{lipsum}

\usepackage{times}

\usepackage{soul}
\usepackage{url}
\usepackage[utf8]{inputenc}
\usepackage[small]{caption}
\usepackage{graphicx}
\usepackage{amsmath}
\usepackage{booktabs}
\usepackage{algorithm}
\usepackage{amssymb}
\usepackage{subcaption}
\usepackage{xspace}
\usepackage{xcolor}
\usepackage{footnote}

\newcommand{\csr}{committee selection rule\xspace}

\newcommand{\UCWD}{UCWD\xspace}
\newcommand{\DiReCF}{($\mu$, $\pi$)-DRCF\xspace}
\newcommand{\DiReCFp}[2]{(#1, #2)-DRCF\xspace}
\newcommand{\DiReCWD}{($\mu$, $\pi$, $\mathtt{f}$)-DRCWD\xspace}
\newcommand{\DiReCWDp}[2]{(#1, #2, $\mathtt{f}$)-DRCWD\xspace}

\newcommand{\group}[2]{A_{(#1, #2)}}
\newcommand{\groupsub}[2]{A_{(#1, #2_#1)}}
\newcommand{\population}[2]{P_{(#1, #2)}}
\newcommand{\populationsub}[2]{P_{(#1, #2_#1)}}
\newcommand{\nphard}{NP-hard\xspace}

\newcommand{\etal}{\emph{et~al.}\xspace}
\usepackage{amsthm}
\newtheorem{definition}{Definition}
\newtheorem{theorem}{Theorem}
\newtheorem{example}{Example}
\newtheorem{corollary}{Corollary}
\newtheorem{observation}{Observation}
\newtheorem{conjecture}{Conjecture}
\newtheorem{lemma}{Lemma}
\newtheorem{claim}{Claim}
\usepackage{multirow}
\usepackage{xfrac}
\usepackage{subcaption}
\usepackage{balance}
\usepackage[switch]{lineno} %
\usepackage[noend]{algorithmic}

\usepackage{imakeidx}

\DeclareMathOperator{\pos}{pos}

\title{DiRe Committee : Diversity and Representation Constraints in Multiwinner Elections}

\author{
  Kunal Relia\thanks{This work was supported in part by Julia Stoyanovich's  NSF Grants No. 1916647 and 1934464.} \\
  New York University, USA\\
  \texttt{krelia@nyu.edu} \\
}

\makeindex[columns=3, title=Alphabetical Index, intoc]
\begin{document}
\maketitle

\begin{abstract}
The study of fairness in multiwinner elections focuses on settings where candidates have attributes. However, voters may also be divided into predefined populations under one or more attributes (e.g., ``California'' and ``Illinois'' populations under the ``state'' attribute), which may be same or different from candidate attributes. The models that focus on candidate attributes alone may systematically under-represent smaller voter populations. Hence, we develop a model, DiRe Committee Winner Determination (DRCWD), which delineates candidate and voter attributes to select a committee by specifying diversity and representation constraints and a voting rule. We 
analyze its computational complexity, inapproximability, and parameterized complexity. We develop a heuristic-based algorithm, which finds the winning DiRe committee in under two minutes on 63\% of the instances of synthetic datasets and on 100\% of instances of real-world datasets. We present an empirical analysis of the running time, feasibility, and utility traded-off. 

Overall, DRCWD motivates that a study of multiwinner elections should consider both its actors, namely candidates and voters, as candidate-specific models can unknowingly harm voter populations, and vice versa. 
Additionally, even when the attributes of candidates and voters coincide, it is important to treat them separately as diversity 
does not imply representation and vice versa. This is to say that having a female candidate on the committee, for example, is different from having a candidate on the committee who is preferred by the female voters, and who themselves may or may not be female.
\end{abstract}


\clearpage
\tableofcontents
\clearpage

\section{Introduction}
\label{sec:intro}


The problem of selecting a committee from a given set of candidates arises in multiple domains; it ranges from political sciences (e.g., selecting the parliament of a country) to recommendation systems (e.g., selecting the movies to show on Netflix). Formally, given a set $C$ of $m$ candidates (politicians and movies, respectively), a set $V$ of $n$ voters (citizens and Netflix subscribers, respectively) give their ordered preferences over all candidates $c \in C$ to select a committee of size $k$. These preferences can be stated directly in case of parliamentary elections, or they can be derived based on 
input, such as when Netflix subscribers' viewing behavior is used to derive their preferences. In this paper, we focus on selecting a $k$-sized (fixed size) committee using direct, ordered, and complete preferences. 

Which committee is selected depends on the \csr, also called multiwinner voting rule. 
Examples of commonly used families of rules when a complete ballot of each voter is given are Condorcet principle-based rules \cite{faliszewski2016committee}, which select a committee that is at least as strong as every other committee in a pairwise majority comparison,   approval-based voting rules \cite{faliszewski2016committee,kilgour2010approval,sanchez2017proportional} where each voter submits an \emph{approval ballot} approving a subset of candidates, and ordinal preference ballot-based voting rules like k-Borda and $\beta$-Chamberlin-Courant ($\beta$-CC) \cite{elkind2017properties,faliszewski2017multiwinner} that are analogous to single-winner rules. We note that this version of CC rule is different from the Chamberlin–Courant \emph{approval} voting rule used in the context of approval elections \cite{aziz2017justified, lackner2021consistent}. We refer readers to Section 2.2 of \cite{faliszewski2017multiwinner} for further details on the commonly used families of multiwinner voting rules. In this paper, we focus on ordinal preference-based rules that are analogous to single-winner rules.


Recent 
work on fairness in multiwinner elections show that 
these rules can 
create or propagate biases by 
systematically harming 
candidates coming from historically 
disadvantaged
groups~\cite{bredereck2018multiwinner,celis2017multiwinner}. Hence, diversity constraints on candidate attributes were introduced to overcome this problem.  However, voters may be divided into predefined populations under one or more attributes, which may be different from candidate attributes. For example, voters in Figure~\ref{fig:example/voters} are divided into ``California'' and ``Illinois'' populations under the ``state'' attribute. The models that focus on candidate attributes alone may 
systematically under-represent smaller voter populations.  

\begin{figure}[t!]
\centering
\begin{subfigure}{.3\textwidth}
  \centering
  \includegraphics[width=0.975\linewidth]{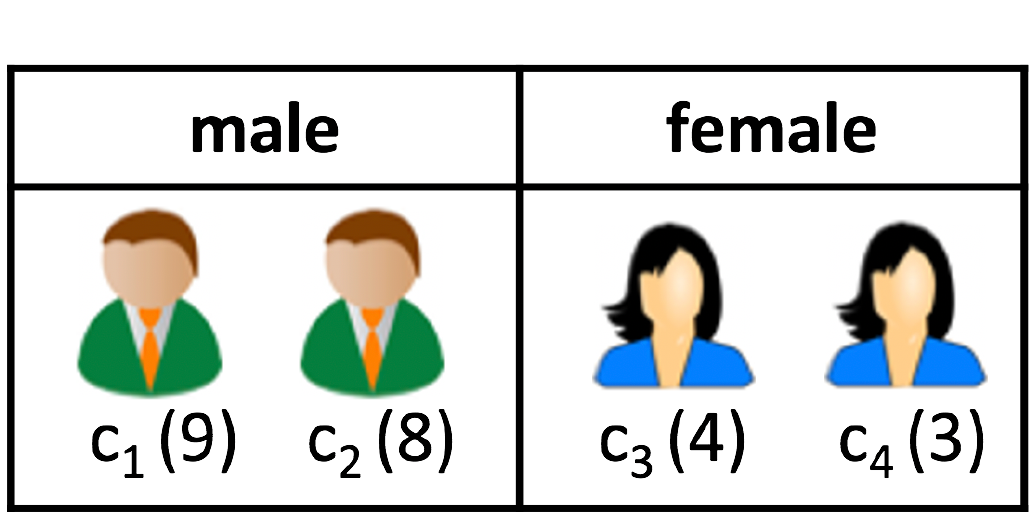}
  \caption{\textbf{candidates}}
  \label{fig:example/candidates}
\end{subfigure}%
\begin{subfigure}{.3\textwidth}
  \centering
  \includegraphics[width=0.975\linewidth]{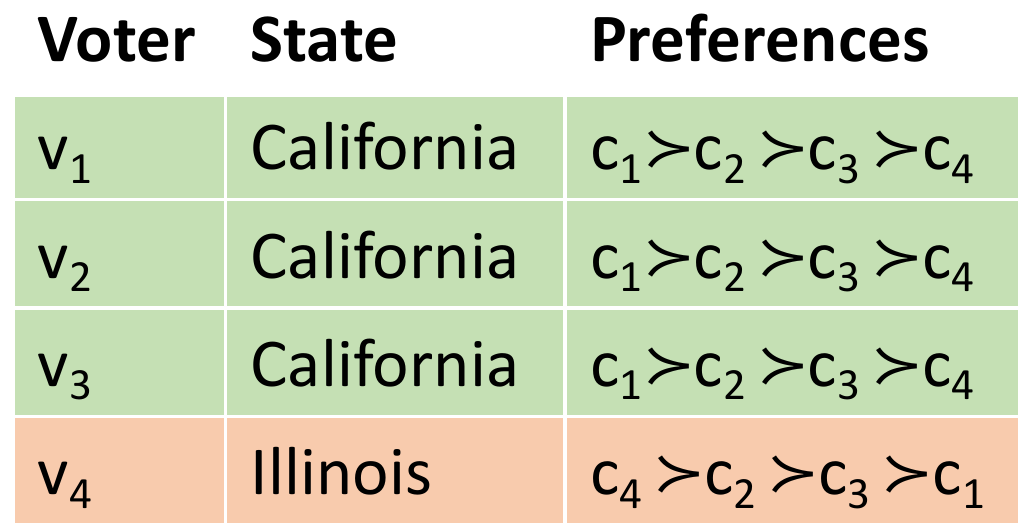}
  \caption{\textbf{voters}}
  \label{fig:example/voters}
\end{subfigure}
\caption{(a) Candidates with ``gender'' attribute (and their Borda scores) and (b) voters with ``state'' attribute. The winning committee (size $k$=2) for California and Illinois, states in the United States, is \{$c_1,c_2\}$ and \{$c_4,c_2\}$, respectively.}
\label{fig:example}
\end{figure}

\begin{example}
\label{eg:DiReMotivation}
Consider an election $E$ consisting of $m$ = 4 candidates (Figure~\ref{fig:example/candidates}) and $n$ = 4 voters giving ordered preference over $m$ candidates (Figure~\ref{fig:example/voters}) to select a committee of size $k$ = 2. Candidates and voters have one attribute each, namely gender and state, respectively. 
The $k$-Borda\footnote{The Borda rule associates score $m - i$ with the $i^{\text{th}}$ position, and $k$-Borda selects $k$ candidates with the highest Borda score.} winning committee computed for each voter population is $\{c_1$, $c_2\}$ for California and $\{c_4$, $c_2\}$ for Illinois.

Suppose that we impose a \emph{diversity constraint} that requires the committee 
to have at least one candidate of each gender, and a \emph{representation constraint} that requires the committee to have 
at least one candidate from the winning committee of each state.  Observe that the highest-scoring committee, which is also representative, consists of $\{c_1$, $c_2\}$ (score = 17), but this committee is not diverse, since both candidates are male. Further, the highest-scoring diverse committee consisting of $\{c_1$, $c_3\}$  (score = 13) is not representative because it does not include any winning candidates from Illinois, the smaller state. The highest-scoring diverse \textbf{and} representative committee is $\{c_2$, $c_3\}$ (score = 12).
\end{example}

This example illustrates the inevitable utility cost due to enforcing of additional constraints.

Note that, in contrast to prior work in computational social choice, we incorporate  voter attributes that are separate from candidate attributes. 
Also, our work is different from the notion of ``proportional representation'' \cite{sanchez2017proportional,brill2018multiwinner,monroe1995fully}, where the number of candidates selected in the committee from each group is proportional to the number of voters preferring that group, and from its variants such as ``fair'' representation \cite{koriyama2013optimal}. All these approaches dynamically divide the voters based on the cohesiveness of their preferences. 
Another related work, multi-attribute proportional representation \cite{lang2018multi}, couples candidate and voter attributes.  An important observation we make here is that, even if the attributes of the candidates and of the voters \emph{coincide}, it may still be important to treat them \emph{separately} in committee selection.  
This is because having a female candidate on the committee, for example, is different from having a candidate on the committee who is preferred by the female voters, and who themselves may or may not be female.



{\bf Contributions.} In this paper, we define a model that treats candidate and voter attributes separately during committee selection, and thus enables selection of the highest-scoring diverse \emph{and} representative committee. 
We show  NP-hardness of committee selection under our model for various settings, give results on inapproximability and parameterized complexity, and present a heuristic-based algorithm. Finally, we present an experimental evaluation using real and synthetic datasets, in which we show the efficiency of our algorithm, analyze the feasibility of committee selection and illustrate the utility trade-offs.

\section{Related Work}
\label{sec:RW}

Our work is at an intersection of multiple ideas, and hence, in this section, we briefly summarize the related work spread across different domains, some of which we already discussed in the previous section.

\paragraph{Fairness in Ranking and Set Selection.}
There is a growing understanding in the field of theoretical computer science about the possible presence of algorithmic bias in multiple domains \cite{baeza2016data, bellamy2018ai, celis2017ranking, danks2017algorithmic, hajian2016algorithmic, lambrecht2019algorithmic}, especially in variants of set selection problem \cite{stoyanovich2018online}. 
The study of fairness in ranking and set selection, closely related to the study of multiwinner elections, use constraints in algorithms to mitigate bias caused against historically disadvantaged groups. Stoyanovich \etal~\cite{stoyanovich2018online} use constraints in streaming set selection problem, and Yang and Stoyanovich~\cite{yang2017measuring} and Yang \etal~\cite{yang2019balanced} use constraints for ranked outputs. 
Kuhlman and Rundensteiner~\cite{Kuhlman2020rank} focus on fair rank aggregation and Bei~\etal \cite{bei2020candidate} use proportional fairness constraints. 
Our work adds to the research on the use of constraints to mitigate algorithmic bias.

\paragraph{Fairness in Participatory Budgeting.} Multiwinner elections are a special case of participatory budgeting, and fairness in the latter domain has also received particular attention. For example, projects (equivalent to candidates) are divided into groups and for fairness they consider lower and upper bounds on utility
achieved and the lower and upper bounds on cost of projects used in every group \cite{patel2020group}. Fluschnik \etal \cite{fluschnik2019fair} aim to achieve fairness among projects using their objective function. Next, Hershkowitz \etal \cite{hershkowitz2021district} have studied fairness from the utility received by the districts (equivalent to voters), Peters \etal \cite{peters2020proportional} define axioms for proportional representation of voters, and Lackner \etal \cite{lackner2021fairness} define fairness in long-term participatory budgeting from voters' perspective. However, we note that none of these work simultaneously consider fairness from the perspective of both, the projects and the districts. 

\paragraph{Two-sided Fairness.}
The need for fairness from the perspective of different stakeholders of a system is well-studied. 
For instance, Patro \etal \cite{patro2020fairrec}, Chakraborty \etal \cite{chakraborty2017fair}, and Suhr \etal \cite{suhr2019two} consider two-sided fairness in two-sided platforms\footnote{A two-sided platform is an intermediary economic platform having \emph{two distinct user groups} that provide each other with network benefits such that the decisions of each set of user group affects the outcomes of the other set \cite{rysman2009economics}. For example, credit cards market consists of  cardholders and merchants and health maintenance organizations consists of patients and doctors.} and Abdollahpouri\etal \cite{abdollahpouri2019multi} and Burke \etal \cite{burke2017multisided} shared desirable fairness properties for different categories of multi-sided platforms. However, this line of work focuses on multi-sided fairness in multi-sided platforms, which is technically different from an election. An election, roughly speaking, can be considered a ``one-sided platform'' consisting of more than one stakeholders as during an election, candidates do not make decisions that affect the voters. Hence, $\delta$-sided fairness in one-sided platform is also needed where $\delta$ is the number of distinct user-groups on the platform. More generally, $\delta$-sided fairness in $\eta$-sided platform warrants an analysis of $\delta \cdot \eta$ perspectives of fairness, i.e., the effect of fairness on each of the $\delta$ stakeholders for each of the $\eta$ fairness metrics being used. In elections, $\delta = 2$ (candidates and voters) and $\eta=1$ (voting). Additionally, Aziz \cite{aziz2021two} summarized a line of work related to diversity concerns in two-sided matching that focused on diversity with respect to one stakeholder only. 


\paragraph{Unconstrained Multiwinner Elections and Proportional Representation.}
The study of complexity of unconstrained multiwinner elections has received attention  \cite{faliszewski2017multiwinner}. Selecting a committee using Chamberlin-Courant (CC) \cite{chamberlin1983representative} rule is NP-hard \cite{procaccia2008complexity}, and approximation algorithms have resulted in the best known ratio of $1-\frac{1}{e}$
\cite{skowron2015achieving,lu2011budgeted}. Yang and Wang~\cite{yang2018parameterized} studied its parameterized complexity. 
Another commonly studied rule, Monroe \cite{monroe1995fully}, is also NP-hard \cite{betzler2013computation, elkind2017properties}. Sonar \etal \cite{sonar2020complexity} showed that even checking whether a given committee is optimal when using these two rules is hard. Finally, the hardness of problems involving restricted voter preferences and \csr have been studied   \cite{elkind2015structure,peters2017preferences} and so has the proportional representation in dynamic ranking \cite{israel2021dynamic}.

\paragraph{Constrained Multiwinner Elections.}
The study of complexity of using diversity constraints in elections 
and its complexity 
has also received particular attention. 
Goalbase score functions, which specify an arbitrary set of logic constraints and let the score capture the number of constraints satisfied, could be used to ensure diversity \cite{uckelman2009representing}. Using diversity constraints over multiple attributes in single-winner elections is NP-hard \cite{lang2018multi}. Also, using diversity constraints over multiple attributes in multiwinner elections is NP-hard, which has led to approximation algorithms and matching hardness of approximation results by Bredereck \etal \cite{bredereck2018multiwinner} and Celis \etal \cite{celis2017multiwinner}. Finally, due to the hardness of using diversity constraints over multiple attributes in approval-based multiwinner elections \cite{brams1990constrained}, these have been formalized as integer linear programs (ILP) \cite{potthoff1990use}. In contrast, Skowron \etal \cite{skowron2015achieving} showed that ILP-based algorithms fail in real world when using ranked voting-related proportional representation rules like Chamberlin-Courant and Monroe rules, even when there are no constraints.

Overall, the work by Bredereck \etal \cite{bredereck2018multiwinner}, Celis \etal \cite{celis2017multiwinner}, and  Lang and Skowron \cite{lang2018multi} is closest to ours. However, we differ as we: (i) consider elections with predefined voter populations under one or more attributes, (ii) delineate voter and candidate attributes even when they coincide, and (iii) consider representation \emph{and} diversity constraints. No previous work, to the best of our knowledge, has considered fairness from the perspective of voter attributes or has delineated candidate and voter attributes even when they coincide.



\section{Preliminaries and Notation}
\label{sec:prelim}


\paragraph{Multiwinner Elections.} Let $E = (C, V )$ be an election consisting of a candidate set $C = \{c_1,\dots,c_m\}$ and a voter set $V = \{v_1,\dots,v_n\}$, where each voter $v \in V$ has a preference list $\succ_{v}$ over $m$ candidates, ranking all of the candidates from the most to the least desired. $\pos_{v}(c)$ denotes the position of candidate $c \in C$ in the ranking of voter $v \in V$, where the most preferred candidate has position 1 and the least preferred has position $m$. 

Given an election $E = (C,V)$ and a positive integer $k \in [m]$ (for $k \in \mathbb{N}$, we write $[k] = \{1, \dots, k\}$), a multiwinner election 
selects a $k$-sized subset of candidates (or a committee) $W$ using a multiwinner voting rule $\mathtt{f}$ (discussed later) such that the score of the committee $\mathtt{f}(W)$ is the highest. Formally, given $E = (C,V)$ and $k$, $\mathtt{f}$ outputs the required committee $W$ of exactly $k$ candidates with the highest score. 
We assume ties are broken using a pre-decided priority order over all candidates.

\paragraph{Candidate Groups.} 
The candidates have $\mu$ attributes, $A_1 , . . . , A_\mu$, such that $\mu \in \mathbb{Z}$ and $\mu \geq 0$. Each attribute $A_i$, for all $i \in [\mu]$, partitions the candidates into $g_i \in [m]$ groups, $\group{i}{1} , . . . , \groupsub{i}{g} \subseteq C$. Formally, $\group{i}{j} \cap  \group{i}{j'} = \emptyset$, $\forall j,j' \in [g_i], j \ne j'$. For example, candidates in Figure~\ref{fig:example/candidates} have one attribute gender ($\mu$ = 1) with two disjoint groups, male and female ($g_1$ = 2). Overall, the set $\mathcal{G}$ of \emph{all} such arbitrary and potentially non-disjoint groups will be $\group{1}{1} , . . . , \groupsub{\mu}{g} \subseteq C$. Note that the number of groups a candidate belongs to is equal to the number of attributes $\mu$. 



\paragraph{Voter Populations.}
The voters have $\pi$ attributes, $A'_1 , . . . , A'_\pi$, such that $\pi \in \mathbb{Z}$ and $\pi \geq 0$. The voter attributes may be different from the candidate attributes. Each attribute $A'_i$, for all $i \in [\pi]$, partitions the voters into $p_i \in [n]$ populations, $\population{i}{1} , . . . , \populationsub{i}{p} \subseteq V$. Formally, $\population{i}{j} \cap  \population{i}{j'} = \emptyset$, $\forall j,j' \in [p_i], j \ne j'$. For example, voters in Figure~\ref{fig:example/voters} have one attribute state ($\pi$ = 1), which has two populations California and Illinois ($p_1$ = 2). Overall, the set $\mathcal{P}$ of \emph{all} such predefined and potentially non-disjoint populations will be $\population{1}{1} , . . . , \populationsub{\pi}{p} \subseteq V$. 

The number of populations a voter belongs to is equal to the number of attributes $\pi$. Additionally, we are given $W_P$, the winning committee of each population $P \in \mathcal{P}$. We note that a fine-grained accounting of representation is not possible in our model. This is because when a \csr such as Chamberlin-Courant rule is used to determine each population’s winning committee $W_P$, then a complete-ranking of each population’s collective preferences is not possible. Thus, we have design our model to only consider each population's winning committee $W_P$.


\paragraph{Multiwinner Voting Rules.} There are multiple types of multiwinner voting rules, also called committee selection rules. 
In this paper, we focus on committee selection rules $\mathtt{f}$ that are based on single-winner positional voting rules, and are
monotone and submodular ($\forall A \subseteq B, f(A) \leq f(B)$ and $f(B) \leq f(A) + f(B \setminus  A)$) \cite{bredereck2018multiwinner,celis2017multiwinner}. 


\begin{definition}
\label{def-CC} \textbf{Chamberlin–Courant (CC) rule:}
The CC rule \cite{chamberlin1983representative} associates each voter with a candidate in the committee who is their most preferred candidate in that committee. 
 The score of a committee is the sum of scores given by voters to their associated candidate. 
 Specifically, $\mathbf{\beta}$-CC uses Borda positional voting rule such that it assigns a score of $m-i$ to the $i^{\text{th}}$ ranked candidate who is their highest ranked candidate in the committee.

        
\end{definition}

\begin{definition}
\label{def-monroe} \textbf{Monroe rule:} The Monroe rule \cite{monroe1995fully} dynamically divides the $n$ voters into $\pi$ populations based on the cohesiveness of their preferences 
where $\pi$ = $k$ (assuming $k$ divides $n$).
Then, each subpopulation's most preferred candidate is selected into the $k$-sized committee. Formally, for each population, say $P \in \mathcal{P}$, select the candidate $c$  that has the highest score for that subpopulation: $\max_{c \in C}(\mathtt{f}_{P}(c))$. In other words, each candidate in the committee is represented by an equal number of voters.
\end{definition}

A special case of submodular functions are separable functions, which calculate the score of committee as follows: 
score of a committee $W$ is the sum of the scores of individual candidates in the committee. Formally, $\mathtt{f}$ is separable if it is submodular and $\mathtt{f}(W) = \sum_{c \in W}^{}\mathtt{f}(c)$ \cite{bredereck2018multiwinner}. 
Monotone and separable selection rules are natural and are considered good when the goal of an election is to shortlist a set of individually excellent candidates \cite{faliszewski2017multiwinner}:


\begin{definition}
\label{def-kborda} \textbf{$k$-Borda rule} The $k$-Borda rule
outputs committees of $k$ candidates with the highest Borda scores.
\end{definition}

\section{DiRe Committee Model}
\label{sec:DiReModel}

In this section, we formally define a model to select a diverse \emph{and} representative committee, namely $DiRe$ committee, and show its generality.

\begin{definition}
\label{def-UCWD}
\textbf{Unconstrained Committee Winner Determination (\UCWD):} We are given a set $C$ of $m$ candidates, a set $V$ of $n$ voters such that each voter $v \in V$ has a preference list $\succ_{v}$ over $m$ candidates, a \csr $\mathtt{f}$, and a committee size $k \in [m]$.
Let $\mathcal{W}$ denote the family of all size-$k$ committees. The goal of \UCWD is to select a committee $W \in \mathcal{W}$ that maximizes $\mathtt{f}(W)$. 
\end{definition}

We now discuss the diversity and representation constraints. The lowest possible value that these constraints can take is 1, which replicates real-world scenarios. For instance, the United Nations charter guarantees at least one representative to each member country in the United Nations General Assembly, independent of the country's population. Similarly, each state of the United States of America is guaranteed at least three representatives in the US House of Representatives. Hence, from fairness perspective, each candidate group and voter population deserves at least one candidate in the committee. Theoretically, all results in this paper hold even if the lowest possible value that the constraints can take is 0.
\paragraph{Diversity Constraints,} denoted by $l^D_G \in [1,$ $\min(k, |G|)]$ 
for each candidate group $G \in \mathcal{G}$, enforces at least $l^D_G$ candidates from the group $G$ to be in the committee $W$. Formally, for all $G \in \mathcal{G}$, $|G \cap W|\geq l^D_G$. 
We note that we do not propose to use the upper bounds as it induces quota system, which is not desirable from social choice perspective. 

\paragraph{Representation Constraints,} denoted by $l^R_P \in [1,k]$ for each voter population $P \in \mathcal{P}$, enforces at least $l^R_P$ candidates from the population $P$'s committee $W_P$ to be in the committee $W$. Formally, for all $P \in \mathcal{P}$, $|W_P \cap W|\geq l^R_P$. We again do not propose to use the upper bounds as it induces the undesirable quota system.  

\begin{definition}
\label{def-DiReCF}
\textbf{($\mu,\pi$)-DiRe Committee Feasibility (\DiReCF):} We are given an instance of election $E=(C,V)$,  
a committee size $k \in [m]$, a set of candidate groups $\mathcal{G}$ over $\mu$ attributes and their diversity constraints $l^D_G$ for all $G \in \mathcal{G}$, and a set of voter populations $\mathcal{P}$ over $\pi$ attributes and their representation constraints $l^R_P$ and the winning committees $W_P$ for all $P \in \mathcal{P}$.
Let $\mathcal{W}$ denote the family of all size-$k$ committees. 
The goal of \DiReCF is to select committees $W \in \mathcal{W}$ that satisfy the diversity and representation constraints such that $|G\cap W|\geq l^D_G$ for all $G\in \mathcal{G}$ and $|W_P\cap W|\geq l^R_P$ for all $P\in \mathcal{P}$. All such committees that satisfy the constraints are called DiRe committees.
\end{definition}

If a \csr $\mathtt{f}$ is also an input to the feasibility problem, we get the \DiReCWD problem:

\begin{definition}
\label{def-DiReCWD}
\textbf{($\mu,\pi$, $\mathtt{f}$)-DiRe Committee Winner Determination (\DiReCWD):} 
Given an instance of \DiReCF and a \csr $\mathtt{f}$, let $\mathcal{W}$ denote the family of all size-$k$ committees, then the goal of \DiReCWD is to select a committee $W \in \mathcal{W}$ that maximizes the $\mathtt{f}(W)$ among all $DiRe$ committees.
\end{definition}

We note that we denote the possible values that $\mu$ and $\pi$ can take using parenthesis. For example, `\DiReCWDp{$\leq2$}{0}' implies that we are specifying a setting $\forall \mu \in \mathbb{Z} : 0 \leq \mu \leq 2$. We use
the same notation for ‘$\geq$’ such that `\DiReCWDp{$\geq3$}{0}' implies that we are specifying a setting $\forall \mu \in \mathbb{Z} :  \mu \geq 3$ . We use the same notation for $\pi$.

\begin{observation}\label{obs:hardnessrelations}
\DiReCWD is a generalized version of \DiReCF and \UCWD. Hence, if \DiReCWD is polynomial time computable, then so are the corresponding \UCWD and \DiReCF problems. 
If either \UCWD is NP-hard or \DiReCF is NP-hard, then \DiReCWD is NP-hard.
\end{observation}



\subsection{\texorpdfstring{\DiReCWD}{Lg} and Related Models}
\label{sec:general}

Our model provides the flexibility to specify the diversity and representation constraints and to select the voting rule. Thus, 
in this section we define the diverse committee problem \cite{bredereck2018multiwinner,celis2017multiwinner}  and the apportionment problem \cite{brill2018multiwinner,hodge2018mathematics}  as special cases of \DiReCWD.

\paragraph{\DiReCWDp{$\mu$}{0} and Diverse Committee Problem.}
We define the diverse committee problem in our model \cite{bredereck2018multiwinner,celis2017multiwinner}: In the diverse committee problem, we are given an instance of \UCWD that consists of 
a set of candidate groups $\mathcal{G}$ and the corresponding diversity constraints, lower bound $l^D_G$ and upper bound $u^D_G$, for all $G \in \mathcal{G}$. Let $\mathcal{W}$ denote the family of all size-$k$ committees. The goal of the diverse committee problem is to select a committee $W \in \mathcal{W}$ that maximizes the $\mathtt{f}(W)$ among the committees that satisfy the constraints.

It is clear that \DiReCWDp{$\mu$}{0}, i.e., without the presence of any voter attributes, is equivalent to the diverse committee problem. As we do not use upper bounds, our model is generalizable when the upper bound $u^D_G$ in the diverse committee model is equal to the size of group $G$ for all $G \in \mathcal{G}$ and the minimum value that the lower bound can take is 1 for all $G \in \mathcal{G}$. This is in line with the approach used in Theorem 6 of Celis \etal \cite{celis2017multiwinner}. Formally, $u^D_G=|G|$ and $l^D_G\geq1$ for all $G \in \mathcal{G}$. 

\paragraph{\DiReCWDp{0}{1} and Apportionment Problem.}
We define the apportionment problem in our model \cite{brill2018multiwinner}: In the apportionment problem, we are given an instance of \UCWD that consists of
a set of disjoint voter populations $\mathcal{P}$ over one attribute and winning committees $W_P$ for all $P \in \mathcal{P}$. Let $\mathcal{W}$ denote the family of all size-$k$ committees. The goal of the apportionment problem is to select a committee $W \in \mathcal{W}$ that maximizes the $\mathtt{f}(W)$ among all the committees that satisfy the lower quota, i.e., $\forall P \in \mathcal{P}$, $|W_P \cap W| \geq \frac{|P|}{n}\cdot k$. 

It is easy to see that \DiReCWDp{0}{1}, which consists of zero candidate attributes and one voter attribute, is same as the apportionment problem if we set the representation constraint of each population to be equal to the lower quota of the apportionment problem. Formally, $\forall P \in \mathcal{P}$, $l^R_P$ = $\left\lfloor\frac{|P|}{n}\cdot k\right\rfloor$, realistically assuming  that $\forall P \in \mathcal{P}, \frac{|P|}{n} \geq \frac{1}{k}$.

Finally, we note that our model can be adopted to accept approval votes as an input and thus if each population is completely cohesive within itself, then the representation constraints can be set to formulate known representation methods like proportional justified representation \cite{sanchez2017proportional} and extended justified representation \cite{aziz2017justified} as \DiReCWD. Though we note that such reformulations may not be as straightforward as the discussed reformulations.

\section{Complexity Results}
\label{sec:compresults}



\begin{table}[t]
\centering

\begin{tabular}{|c||c|}
\hline
Instance of \DiReCWD & Computational Complexity\\

\hline\hline
($\leq2$, $0$, separable)-DRCWD  & P (Lem.~\ref{lemma:DiReCFinP})\\
\hline
($\geq3$, 0, separable)-DRCWD & \nphard (Thm.~\ref{thm:DiReCWDdiv30odd}, Thm.~\ref{thm:DiReCWDdiv30even})\\
\hline
($\geq0$, $\geq1$, separable)-DRCWD & \nphard (Thm.~\ref{lemma:DiReCWDrep01}, Cor.~\ref{lemma:DiReCWDrep}) \\
\hline
($\geq0$, $\geq0$, submodular)-DRCWD & \nphard (Thm.~\ref{lemma:DiReCWDsubmod}, Cor.~\ref{cor:DiReCWDsubmod})\\
\hline
\end{tabular}
\vspace{0.1cm}
\caption{A summary of complexity of \DiReCWD (Theorem~\ref{thm:DiReClass}, Corollary~\ref{cor:DiReClass}). The value in brackets for $\mu$ and $\pi$ denote that the results hold for all non-negative integers $\mu$ and all non-negative integers $\pi$ that satisfy the condition stated in the brackets. The results are under the assumption P $\neq$ NP. `Lem.' denotes Lemma. `Thm.' denotes Theorem. `Cor.' denotes Corollary.
}
\label{tab:compResults}
\end{table}

In this section, we give a classification of the computational complexity\footnote{The hardness, inapproximability, and parameterized complexity results throughout the paper are under the assumption P $\neq$ NP.} of the \DiReCWDp{$\mu$}{$\pi$} 
problem under different settings.
Finding a committee using a submodular scoring function like the utilitarian version of Chamberlin-Courant rule is known to be NP-hard \cite{procaccia2008complexity} 
and selecting a diverse committee when a candidate belongs to three groups is also known to be NP-hard \cite{bredereck2018multiwinner,celis2017multiwinner}. However, the proofs of these hardness results are fragmented over several papers and the proofs use reductions from several well-known NP-hard problems. For instance, the proof of hardness for the use of Chamberlin-Courant uses a reduction from exact 3-cover \cite{procaccia2008complexity} and the proof of hardness for computing a diverse committee uses a reduction from 3-dimensional matching \cite{bredereck2018multiwinner} and 3-hypergraph matching \cite{celis2017multiwinner}. Moreover, we are the first ones to introduce representation constraints and hardness due to its use is unknown. Hence, in this section, we provide a complete classification of the \DiReCWD problem by giving a reduction from a \emph{single} well-known NP-hard problem, namely, the vertex cover problem, inspired by the similar approach used in \cite{chakraborty2021classifying}.

Finally, we note that as the following classification holds for \emph{every} integer $\mu \geq 0$ (specifically, every whole number as $\mu$ can not be negative) and \emph{every} integer $\pi\geq0$, our reductions are designed for the same range of values.


\begin{theorem}\label{thm:DiReClass}
Let $\mu,$ $\pi \in \mathbb{Z}$ : $\mu,$ $\pi\geq 0$ and $\mathtt{f}$ be a \csr, then \DiReCWDp{$\mu$}{$\pi$} is NP-hard.
\end{theorem}

\begin{corollary}\label{cor:DiReClass}\textbf{Classification of Complexity of \DiReCWD.}
\begin{enumerate}
    \item If $\forall \mu \in \mathbb{Z} : \mu\geq0$, $\forall \pi \in \mathbb{Z} : \pi\geq0$, and $\mathtt{f}$ is a monotone, \textbf{submodular} function, then \DiReCWD is NP-hard. 
    \item If $\forall \mu \in [0$, $2]$, $\pi=0$, and $\mathtt{f}$ is a monotone, \textbf{separable} function, then \DiReCWD is in P. 
    \item If $\forall \mu \in \mathbb{Z} : \mu\geq3$, $\pi=0$, and $\mathtt{f}$ is a monotone, \textbf{separable} function, then \DiReCWD is NP-hard.
    \item If $\forall \mu \in \mathbb{Z} : \mu\geq0$, $\forall \pi \in \mathbb{Z} : \pi\geq1$, and $\mathtt{f}$ is a monotone, \textbf{separable} function, then \DiReCWD is NP-hard.
\end{enumerate}

\end{corollary}

\subsection{Tractable Case}

\begin{theorem}\label{them:DCFinP}[Theorem 21, Corollary 22 in full-version of Celis \etal \cite{celis2017multiwinner}]
The diverse committee feasibility problem can be solved in polynomial time when $\mu$ = 2.

\end{theorem}

Without loss of generality (W.l.o.g.), the above theorem holds when it is assumed that $\mu=2$. Hence, it holds for all $\mu \in \mathbb{Z}$ : $0 \leq \mu \leq$  2.
Therefore, based on the relationship between \DiReCWDp{$\mu$}{0} and Diverse Committee Problem (Section~\ref{sec:general}), we prove the following lemma, which in turn proves the statement in Corollary~\ref{cor:DiReClass}(2):
\begin{lemma}\label{lemma:DiReCFinP}
If $ \mu \in [0$, $2]$, $\pi=0$, and $\mathtt{f}$ is a monotone, \textbf{separable} function, then \DiReCWD is in P.
\end{lemma}

\begin{proof}
When $\pi$=0, there are no voter attributes or representation constraints, and hence, the \DiReCWDp{$\mu$}{$0$} problem is equivalent to the diverse committee problem. Moreover, when  $\mathtt{f}$ is a monotone, \textbf{separable} function, then the complexity of the \DiReCWDp{$\mu$}{$0$} is equivalent to the complexity of \DiReCFp{$\mu$}{$0$}. Thus, the polynomial time result of diverse committee \emph{feasibility} problem when the number of groups a candidate belongs to is equal to two, which in our model implies that the number of candidate attributes is  equal to two ($\mu= 2$), holds for our setting (Theorem 9 \cite{bredereck2018multiwinner},  Corollary 22 (full-version) \cite{celis2017multiwinner}). 

More specifically, when $\mu=2$, we use the algorithm given in the proof of Theorem 21 by Celis \etal \cite{celis2017multiwinner} and set the upper bound equal to the group size. Formally, $u_G^D=|G|$ for all $G \in \mathcal{G}$.

Next, when $\mu=1$, a straight-forward algorithm that selects the top $l_G^D$ scoring candidates for all $G \in \mathcal{G}$ results into a DiRe committee, which satisfies the diversity constraints $|G \cap W|\geq l_G^D$.
\end{proof}

\subsection{Hardness Results}
\paragraph{NP-hard problem used.} As discussed earlier, the NP-hardness of \DiReCWD when using representation constraints is unknown. Moreover, the known hardness results for using submodular but not separable scoring function and diverse committee selection problems were established via reductions from different NP-hard problems.  We will establish the NP-hardness of \DiReCWD for various settings of $\mu$, $\pi$, and $\mathtt{f}$ via reductions from a single well known NP-hard problem, namely, the vertex cover problem on 3-regular\footnote{A 3-regular graph stipulates that each vertex is connected to exactly three other vertices, each one with an edge, i.e., each vertex has a degree of 3. The VC problem on 3-regular graphs is NP-hard. We use 3-regular graphs to exploit its structure to prove the hardness of \DiReCWD with respect to (w.r.t.) diversity constraints (Theorem~\ref{thm:DiReCWDdiv30odd} and Theorem~\ref{thm:DiReCWDdiv30even}). We note that the reductions used in the proofs of Theorem~\ref{lemma:DiReCWDrep01} and Theorem~\ref{lemma:DiReCWDsubmod} do not need 3-regular graphs and hold for VC problem on arbitrary graphs as well.}, 2-uniform\footnote{The size of hyperedges has implications in the hardness of approximation and parameterized complexity results and hence, we mention it over here. For the complexity results, we use 2-uniform hypergraphs only.} hypergraphs \cite{garey1979computers, alimonti1997hardness}. 

\begin{definition}
\label{def-VC}
\textbf{Vertex Cover (VC) problem:} 
Given a graph $H$ consisting of a set of $m$ vertices $X$ = $\{x_1,x_2,\dots,x_{m}\}$ and a set of $n$ edges $E$ = $\{e_1,e_2,\dots,e_{n}\}$ where each $e \in E$ connects two vertices in $X$ such that $e$ corresponds to a 2-element subset of $X$, then a vertex cover of $H$ is a subset $S$ of vertices such that each $e$ contains at least one vertex from $S$ (i.e. $\forall$ $e \in E$, $e \cap S \neq \phi$). The vertex cover problem is to find the minimum vertex cover of $H$.
\end{definition}

\subsubsection{\texorpdfstring{\DiReCWD}{} w.r.t. diversity constraints}
When $\pi=0$, \DiReCWD is related to the diverse committee selection problem. However, the hardness of \DiReCWD when $\mu\geq3$ does not follow the hardness of the diverse committee selection problem when the number of groups that a candidate can belong to is \emph{greater than or equal} 3 \cite{bredereck2018multiwinner, celis2017multiwinner} as the reductions in these papers are specifically for the case when $\mu=3$. 

More specifically, Theorem 9 of Bredereck \etal \cite{bredereck2018multiwinner} uses a reduction from 3-Dimensional Matching that only holds for instances when the number of groups that a candidate can belong to is \emph{exactly} 3. Also, they set lower bound \emph{and} upper bound to 1, which is mathematically different from our setting where we only allow lower bounds. On the other hand, Theorem 6 (``NP-hardness of feasibility: $\Delta$ $\geq$ 3''\footnote{In Celis \etal \cite{celis2017multiwinner}, $\Delta$ denotes ``the maximum number of groups in which any candidate can be''.}) of Celis \etal \cite{celis2017multiwinner} uses two reductions: the first reduction from $\Delta$-hypergraph matching is indeed for the case when the number of groups that a candidate can belong to is \emph{greater than or equal} to 3 but is limited to instances when lower bound is set to 0 and upper bound to 1, which is a trivial case in our setting as we only use lower bounds and do not allow for upper bounds. Moreover, in-principle, the reduction from $\Delta$-hypergraph matching uses a different problem for each $\Delta$ as when $\Delta\neq\Delta'$, the $\Delta$-hypergraph matching and $\Delta'$-hypergraph matching are separate problems. The second reduction from 3-regular vertex cover is for instances when the number of groups that a candidate can belong to is \emph{exactly} 3. 


Hence, in this section, we give a reduction from a \emph{single} known NP-hard problem, namely the vertex cover problem, such that our result holds $\forall \mu \in \mathbb{Z} : \mu\geq3$ even when $\forall G \in \mathcal{G}$, $l_G^D=1$. Also, the reductions are designed to conform to the real-world stipulations: (i) each candidate attribute $A_i,  \forall i \in [\mu]$,  \emph{partitions} all the $m$ candidates into two or more groups and (ii) either no two attributes partition the candidates in the same way or if they do, the lower bounds across groups of the two attributes are not the same. For stipulation (ii), note that if two attributes partition the candidates in the same way and if the lower bounds across groups of the two attributes are also the same, then mathematically they are identical attributes that can be combined into one attribute.
The next two theorems help us prove the statement in Corollary~\ref{cor:DiReClass}(3).





\begin{theorem}\label{thm:DiReCWDdiv30odd}
If $\forall \mu \in \mathbb{Z} : \mu\geq 3$ and $\mu$ is an odd number, $\pi=0$, and   $\mathtt{f}$ is a monotone, separable function, then \DiReCWD 
is NP-hard, even when $\forall G \in \mathcal{G}$, $l_G^D=1$.
\end{theorem}

\begin{proof}

We reduce an instance of vertex cover (VC) problem to an instance of \DiReCWD. We have one candidate $c_i$ for each vertex $x_i \in X$, and $m \cdot (2\mu^2-7\mu+3)$ dummy candidates $d \in D$ where $m$ corresponds to the number of vertices in the graph $G$ and $\mu$ is a positive, odd integer (hint: the number of candidate attributes). Formally, we set $A$ = \{$c_1, \dots, c_m$\} and the dummy candidate set $D$ = \{$d_1, \dots, d_{m \cdot (2\mu^2-7\mu+3)}$\}. Hence, the candidate set $C$ = $A \cup D$ is of size $|C|=$ $m+(m \cdot (2\mu^2-7\mu+3))$ candidates. We set the target committee size to be $k + m \mu^2 - 3m\mu$.

Next, we have $\mu$ candidate attributes. Each edge $e \in E$ that connects vertices $x_i$ and $x_j$ correspond to a candidate group $G \in \mathcal{G}$ that contains two candidates $c_i$ and $c_j$. As our reduction proceeds from a 3-regular graph, each vertex is connected to three edges. This corresponds to each candidate $c \in A$ having three attributes and thus, belonging to three groups. Next, for each of the $m$ candidates $c \in A$, we have $\mu-3$ blocks of dummy candidates and each block contains $2\mu-1$ dummy candidates $d \in D$. Thus, we have a total of $m \cdot (\mu-3) \cdot (2\mu-1)$ dummy candidates, which equals to $m \cdot (2\mu^2-7\mu+3)$ dummy candidates. Next, each block of candidates contains 3 sets of candidates: \textbf{Set $T_1$} contains one candidate and \textbf{Sets $T_2$} and $T_3$ contain $\mu-1$ candidates each. Specifically, each of the $\mu-3$ blocks for each candidate $c \in A$ is constructed as follows: 
\begin{itemize}
    \item Set $T_1$ consists of single dummy candidate, $d_1^{T_1} \in T_1$.
    \item Set $T_2$ consists of $\mu-1$ dummy candidates, $d_i^{T_2} \in T_2$ for all $i \in [1,$ $\mu-1]$. 
    \item Set $T_3$ consists of $\mu-1$ dummy candidates, $d_j^{T_3} \in T_3$ for all $j \in [1,$ $\mu-1]$. 
\end{itemize}

Each candidate in the block has $\mu$ attributes and are grouped as follows:
\begin{itemize}
    \item The dummy candidate $d_1^{T_1} \in T_1$ is in the same group as candidate $c \in A$. It is also in $\mu-1$ groups, individually with each of $\mu-1$ dummy candidates, $d_i^{T_2} \in T_2$. Thus, the dummy candidate $d_1^{T_1} \in T_1$ has $\mu$ attributes and is part of $\mu$ groups.
    \item For each dummy candidate $d_i^{T_2} \in T_2$, it is in the same group as $d_1^{T_1}$ as described in the previous point. It is also in $\mu-1$ groups, individually with each of $\mu-1$ dummy candidates, $d_j^{T_3} \in T_3$. Thus, each dummy candidate $d_i^{T_2} \in T_2$ has $\mu$ attributes and is part of $\mu$ groups.
    \item For each dummy candidate $d_j^{T_3} \in T_3$, it is in $\mu-1$ groups, individually with each of $\mu-1$ dummy candidates, $d_i^{T_2} \in T_2$, as described in the previous point. Next, note that when $\mu$ is an odd number, $\mu-1$ is an even number, which means Set $T_3$ has an even number of candidates. We randomly divide $\mu-1$ candidates into two partitions. Then, we create $\frac{\mu-1}{2}$ groups over one attribute where each group contains two candidates from Set $T_3$ such that one candidate is selected from each of the two partitions without replacement. Thus, each pair of groups is mutually disjoint. Thus, each dummy candidate $d_j^{T_3} \in T_3$ is part of exactly one group that is shared with exactly one another dummy candidate $d_{j'}^{T_3} \in T_3$ where $j \neq j'$. Overall, this construction results in one attribute and one group for each dummy candidate $d_j^{T_3} \in T_3$. Hence, each dummy candidate $d_j^{T_3} \in T_3$ has $\mu$ attributes and is part of $\mu$ groups.
\end{itemize}
As a result of the above described grouping of candidates, each candidate $c \in A$ also has $\mu$ attributes and is part of $\mu$ groups. Note that each candidate $c \in A$ already had three attributes and was part of three groups due to our reduction from vertex cover problem on 3-regular graphs. Additionally, we added $\mu-3$ blocks of dummy candidates and grouped candidate $c \in A$ with candidate $d_1^{T_1} \in T_1$ from each of the $\mu-3$ blocks. Hence, each candidate $c \in A$ has $3+(\mu-3)$ attributes and is part of $\mu$ groups. We set $l^D_G=1$ for all $G \in \mathcal{G}$, which corresponds that each vertex in the vertex cover should be covered by some chosen edge.

Finally, we introduce $m+(m \cdot (2\mu^2-7\mu+3))$ voters. For simplicity, let $c'_i$ denote the $i^{\text{th}}$ candidate in set $C$. The first voter ranks the candidates based on their indices. 
\[c'_1 \succ c'_2 \succ c'_3 \succ \dots \succ c'_{2\mu^2m-7\mu m+4m }\]

The second voter improves the rank of each candidate by one position but places the top-ranked candidate to the last position. 
\[c'_2 \succ c'_3 \succ \dots \succ c'_{2\mu^2m-7\mu m+4m } \succ c'_1\]
Next, the third voter further improves the rank of each candidate by one position but places the top-ranked candidate to the last position. 
\[c'_3 \succ c'_4 \succ \dots \succ c'_{2\mu^2m-7\mu m+4m } \succ c'_1 \succ c'_2\]

Similarly, all the voters rank the candidates based on this method. Hence, the last voter will have the following ranking:

\[c'_{2\mu^2m-7\mu m+4m } \succ c'_1 \succ c'_2 \succ \dots \succ c'_{2\mu^2m-7\mu m+4m-1}\]

Finally, there are no voter attributes, and hence, $\pi=0$ and there are no representation constraints ($l^R_P=\phi)$. This completes our construction for the reduction, which is a polynomial time reduction in the size of $n$ and $m$. Note that we assume that the number of candidate attributes $\mu$ is always less than the number of candidates $|C|$. More specifically, our reduction holds when $3\leq\mu\leq |C|-2$, which is a realistic assumption as we ideally expect $\mu$ to be very small \cite{celis2017multiwinner}.

We first compute the score of the committee and then show the proof of correctness. When $\mathtt{f}$ is a monotone, separable scoring function, we know that 

\[\mathtt{f}(W) = \sum_{c \in W}^{}\mathtt{f}(c)\]

Next, given a scoring vector $\mathbf{s} = (s_1, s_2, \dots, s_{2\mu^2m-7\mu m+4m})$ where $s_1$ is the score associated with candidate $c$ in the ranking of voter $v$ whose $\pos_{v}(c)=1$ and so on, $s_1 \geq s_2 \geq \dots \geq s_{2\mu^2m-7\mu m+4m}$ and $s_1 > s_{2\mu^2m-7\mu m+4m}$, the score of each candidate $c \in C$ is

\[\mathtt{f}(c) = \sum_{v \in V}^{} s_{\pos_{v}(c)} \]
 
but as each candidate occupies each of the $m+(m \cdot (2\mu^2-7\mu+3))$ positions once, $\mathtt{f}(c)$ can be rewritten as

\[\mathtt{f}(c) = \sum_{i =1}^{|\mathbf{s}|} s_i \]

Hence, as all candidates $c \in C$ have the same score, the score of each $k + m \mu^2 - 3m\mu$-sized committee $W \in \mathcal{W}$ will be the highest such that $\mathtt{f}(W)$ is

\[\mathtt{f}(W)= \sum_{c \in W}^{}\mathtt{f}(c) = \sum_{c \in W}^{}\sum_{i =1}^{|\mathbf{s}|} s_i = (k + m \mu^2 - 3m\mu) \cdot \sum_{i =1}^{|\mathbf{s}|} s_i\]

Note that computing any highest scoring committee using a monotone, separable function takes time polynomial in the size of input. 

For clarity w.r.t. to the score of the committee, consider the following example: W.l.o.g., if we assume that $\mathtt{f}$ is $k$-Borda, then $\mathbf{s} = (m+(m \cdot (2\mu^2-7\mu+3))-1,\dots, 1, 0) $. Hence, all candidates $c \in C$ get the same Borda score $\mathtt{f}(c)$ of 

\[(m+(m \cdot (2\mu^2-7\mu+3))-1 + \dots + 1 + 0\]

\[= \frac{(m+(m \cdot (2\mu^2-7\mu+3))-1)\cdot(m+(m \cdot (2\mu^2-7\mu+3)))}{2}\]

\[=  \frac{4\mu^4m^2 - 28\mu^3m^2 + 65\mu^2m^2 - 56\mu m^2 + 16m^2 - 2\mu^2m + 7\mu m -4m}{2}\]

which is the sum of first $m+(m \cdot (2\mu^2-7\mu+3))-1$ natural numbers, all the scores in the scoring vector of Borda rule. 
Therefore, each $k + m \mu^2 - 3m\mu$-sized committee will be the highest scoring committee $W \in \mathcal{W}$ with a $\mathtt{f}(W)$ of 

\[(k + m \mu^2 - 3m\mu) \cdot \frac{4\mu^4m^2 - 28\mu^3m^2 + 65\mu^2m^2 - 56\mu m^2 + 16m^2 - 2\mu^2m + 7\mu m -4m}{2}\]

Hence, the NP-hardness of the problem is due to finding a feasible committee that satisfies for all $G \in \mathcal{G}$, $|G \cap W|\geq l^D_G$ where $l^D_G=1$. Therefore, for the proof of correctness, we show the following:
\begin{claim}

We have a vertex cover $S$ of size at most $k$ that satisfies $e \cap S\neq\phi$ for all $e \in E$ if and only if we have a committee $W$ of size at most $k + m \mu^2 - 3m\mu$ that satisfies all the diversity constraints, which means that for all $G \in \mathcal{G}$, $|G \cap W|\geq l^D_G$ which equals $|G \cap W|\geq 1$ as $l^D_G=1$ for all $G \in \mathcal{G}$. 
\end{claim}

($\Rightarrow$) If the instance of the VC problem is a yes instance, then the corresponding instance of \DiReCWD is a yes instance as each and every candidate group will have at least one of their members in the winning committee $W$, i.e., $|G \cap W|\geq 1$ for all $G \in \mathcal{G}$. Note that we have set $l^D_G=1$ for all $G \in \mathcal{G}$. 

More specifically, for each block of candidates, we select one dummy candidate from Set $T_1$ and all $\mu-1$ dummy candidates from Set $T_3$. This helps to satisfy the condition $|G \cap W|\geq 1$ for all candidate groups that contain at least one dummy candidate $d \in D$. Overall, we select $\mu$ candidates from $\mu-3$ blocks for each of the $m$ candidates that correspond to vertices in the vertex cover. This results in $(\mu \cdot (\mu-3) \cdot m) = m \mu^2 - 3m\mu$ candidates in the committee. Next, for groups that do not contain any dummy candidates, select $k$ candidates $c \in A$ that correspond to $k$ vertices $x \in X$ that form the vertex cover. These candidates satisfy the constraints. Specifically, these $k$ candidates satisfy $|G \cap W|\geq 1$ for all the candidate groups that do not contain any dummy candidates. Hence, we have a committee of size $k + m \mu^2 - 3m\mu$.

($\Leftarrow$) 
The instance of the \DiReCWD is a yes instance when we have $k + m \mu^2 - 3m\mu$ candidates in the committee. This means that each and every group will have at least one of their members in the winning committee $W$, i.e., $|G \cap W|\geq 1$ for all $G \in \mathcal{G}$. Then the corresponding instance of the VC problem is a yes instance as well. This is because the $k$ vertices $x \in X$ that form the vertex cover correspond to the $k$ candidates $c \in A$ that satisfy $|G \cap W|\geq 1$ for all the candidate groups that do not contain any dummy candidates. This completes the proof.
\end{proof}


\begin{theorem}\label{thm:DiReCWDdiv30even}
If $\forall \mu \in \mathbb{Z} : \mu\geq 3$ and $\mu$ is an even number, $\pi=0$, and   $\mathtt{f}$  is a monotone, separable function, then \DiReCWD 
is NP-hard, even when $\forall G \in \mathcal{G}$, $l_G^D=1$.
\end{theorem}

\begin{proof}
We reduce an instance of vertex cover (VC) problem to an instance of \DiReCWD. We have two candidate $c_i$ and $c_{m+i}$ for each vertex $x_i \in X$, and $2m \cdot (2\mu^2-7\mu+3)$ dummy candidates $d \in D$ where $m$ corresponds to the number of vertices in the graph $G$ and $\mu$ is a positive, even integer (hint: the number of candidate attributes). Formally, we set $A$ = \{$c_1, \dots, c_m$\} $\cup$ \{$c_{m+1}, \dots, c_{2m}$\}  and the dummy candidate set $D$ = \{$d_1, \dots, d_{2m \cdot (2\mu^2-7\mu+3)}$\}. Hence, the candidate set $C$ = $A \cup D$ is of size $|C|=$ $2m+(2m \cdot (2\mu^2-7\mu+3))$ candidates. We set the target committee size to be $2k + 2m \mu^2 - 6m\mu$.

Next, we have $\mu$ candidate attributes. Each edge $e \in E$ that connects vertices $x_i$ and $x_j$ correspond to two candidate groups $G, G' \in \mathcal{G}$ such that group $G$ contains two candidates $c_i$ and $c_j$ that correspond to vertices $x_i$ and $x_j$ and the group $G'$ contains two candidates $c_{m+i}$ and $c_{m+j}$ that also correspond to vertices $x_i$ and $x_j$. Note that by having $2m$ candidates in $A$, we are in fact duplicating the graph $H$. As our reduction proceeds from a 3-regular graph, each vertex is connected to three edges. This corresponds to each candidate $c \in A$ having three attributes and thus, belonging to three groups. Next, for each candidate $c \in A$, we have $\mu-3$ blocks of dummy candidates, each block containing $2\mu-1$ dummy candidates $d \in D$. Thus, we have a total of $2m \cdot (\mu-3) \cdot (2\mu-1)$ dummy candidates, which equals to $2m \cdot (2\mu^2-7\mu+3)$ dummy candidates. Next, each block of candidates contains 3 sets of candidates: \textbf{Set $T_1$} contains one candidate and \textbf{Sets $T_2$} and $T_3$ contain $\mu-1$ candidates each. Specifically, each of the $\mu-3$ blocks for each candidate $c \in A$ is constructed as follows in line with the construction in the proof for Theorem~\ref{thm:DiReCWDdiv30odd}: 
\begin{itemize}
    \item Set $T_1$ consists of single dummy candidate, $d_1^{T_1} \in T_1$.
   \item Set $T_2$ consists of $\mu-1$ dummy candidates, $d_i^{T_2} \in T_2$ for all $i \in [1,$ $\mu-1]$. 
    \item Set $T_3$ consists of $\mu-1$ dummy candidates, $d_j^{T_3} \in T_3$ for all $j \in [1,$ $\mu-1]$. 
\end{itemize}

Each candidate in the block has $\mu$ attributes and are grouped as follows:
\begin{itemize}
    \item The dummy candidate $d_1^{T_1} \in T_1$ is in the same group as candidate $c \in A$. It is also in $\mu-1$ groups, individually with each of $\mu-1$ dummy candidates, $d_i^{T_2} \in T_2$. Thus, the dummy candidate $d_1^{T_1} \in T_1$ has $\mu$ attributes and is part of $\mu$ groups.
    \item For each dummy candidate $d_i^{T_2} \in T_2$, it is in the same group as $d_1^{T_1}$ as described in the previous point. It is also in $\mu-1$ groups, individually with each of $\mu-1$ dummy candidates, $d_j^{T_3} \in T_3$. Thus, each dummy candidate $d_i^{T_2} \in T_2$ has $\mu$ attributes and is part of $\mu$ groups.
\end{itemize}

Note that the grouping of the candidates in Set $T_3$ differs significantly from the construction in the proof for Theorem~\ref{thm:DiReCWDdiv30odd}: 
\begin{itemize}
    \item For each dummy candidate $d_j^{T_3} \in T_3$, it is in $\mu-1$ groups, individually with each of $\mu-1$ dummy candidates, $d_i^{T_2} \in T_2$, as described in the previous point. Next, note that when $\mu$ is an even number, $\mu-1$ is an odd number, which means Set $T_3$ has an \emph{odd} number of candidates.  
    We randomly divide $\mu-2$ candidates into two partitions. Then, we create $\frac{\mu-2}{2}$ groups over one attribute where each group contains two candidates from Set $T_3$ such that one candidate is selected from each of the two partitions without replacement. Thus, each pair of groups is mutually disjoint. Hence, each dummy candidate $d_j^{T_3} \in T_3$ is part of exactly one group that is shared with exactly one another dummy candidate $d_{j'}^{T_3} \in T_3$ where $j \neq j'$. Overall, this construction results in one attribute and one group for all but one dummy candidate $d_j^{T_3} \in T_3$, which results into a total of $\mu$ attributes and $\mu$ groups for these $\mu-2$ candidates.
    This is because $\frac{\mu-2}{2}$ groups can hold $\mu-2$ candidates. Hence, one candidate still has $\mu-1$ attributes and is part of $\mu-1$ groups. If this block of dummy candidates is for candidate $c_i \in A$, then another corresponding block of dummy candidates for candidate $c_{m+i} \in A$ will also have one candidate $d_{z}^{T'_3} \in T'_3$ who will have $\mu-1$ attributes and is part of $\mu-1$ groups. We group these two candidates from separate blocks. Hence, now that one remaining candidate also has $\mu$ attributes and is part of $\mu$ groups. As there is always an even number of candidates in set $A$ ($|A|=2m$), such cross-block grouping of candidates among a total of $(\mu-3) \cdot 2m$ blocks, also an even number, is always possible.
\end{itemize}
As a result of the above described grouping of candidates, each candidate $c \in A$ also has $\mu$ attributes and is part of $\mu$ groups. Note that each candidate $c \in A$ already had three attributes and was part of three groups due to our reduction from vertex cover problem on 3-regular graphs. Additionally, we added $\mu-3$ blocks of dummy candidates and grouped candidate $c \in A$ with candidate $d_1^{T_1} \in T_1$ from each of the $\mu-3$ blocks. Hence, each candidate $c \in A$ has $3+(\mu-3)$ attributes and is part of $\mu$ groups. We set $l^D_G=1$ for all $G \in \mathcal{G}$, which corresponds that each vertex in the vertex cover should be covered by some chosen edge.

Finally, we introduce $2m+(2m \cdot (2\mu^2-7\mu+3))$ voters, in line with our reduction in proof of Theorem~\ref{thm:DiReCWDdiv30odd}. For simplicity, let $c'_i$ denote the $i^{\text{th}}$ candidate in set $C$. The first voter ranks the candidates based on their indices. 
\[c'_1 \succ c'_2 \succ c'_3 \succ \dots \succ c'_{4\mu^2m-14\mu m+8m }\]

The second voter improves the rank of each candidate by one position but places the top-ranked candidate to the last position. 
\[c'_2 \succ c'_3 \succ \dots \succ c'_{4\mu^2m-14\mu m+8m } \succ c'_1\]

Similarly, all the voters rank the candidates based on this method. Hence, the last voter will have the following ranking:

\[c'_{4\mu^2m-14\mu m+8m} \succ c'_1 \succ c'_2 \succ \dots \succ c'_{4\mu^2m-14\mu m+8m-1}\]

Finally, there are no voter attributes, and hence, $\pi=0$ and there are no representation constraints ($l^R_P=\phi)$. This completes our construction for the reduction, which is a polynomial time reduction in the size of $n$ and $m$. Note that we assume that the number of candidate attributes $\mu$ is always less than the number of candidates $|C|$. More specifically, our reduction holds when $3\leq\mu\leq |C|-2$, which is a realistic assumption as we ideally expect $\mu$ to be very small \cite{celis2017multiwinner}.

We first compute the score of the committee and then show the proof of correctness. When $\mathtt{f}$ is a monotone, separable scoring function, we know that 

\[\mathtt{f}(W) = \sum_{c \in W}^{}\mathtt{f}(c)\]

Next, given a scoring vector $\mathbf{s} = (s_1, s_2, \dots, s_{4\mu^2m-14\mu m+8m})$ where $s_1$ is the score associated with candidate $c$ in the ranking of voter $v$ whose $\pos_{v}(c)=1$ and so on, $s_1 \geq s_2 \geq \dots \geq s_{4\mu^2m-14\mu m+8m}$ and $s_1 > s_{4\mu^2m-14\mu m+8m}$, the score of each candidate $c \in C$ is

\[\mathtt{f}(c) = \sum_{v \in V}^{} s_{\pos_{v}(c)} \]
 
but as each candidate occupies each of the $2m+(2m \cdot (2\mu^2-7\mu+3))$ positions once, $\mathtt{f}(c)$ can be rewritten as

\[\mathtt{f}(c) = \sum_{i =1}^{|\mathbf{s}|} s_i \]

Hence, as all candidates $c \in C$ have the same score, the score of each $2k + 2m \mu^2 - 6m\mu$-sized committee $W \in \mathcal{W}$ will be the highest such that $\mathtt{f}(W)$ is

\[\mathtt{f}(W)= \sum_{c \in W}^{}\mathtt{f}(c) = \sum_{c \in W}^{}\sum_{i =1}^{|\mathbf{s}|} s_i = (2k + 2m \mu^2 - 6m\mu) \cdot \sum_{i =1}^{|\mathbf{s}|} s_i\]

Note that computing any highest scoring committee using a monotone, separable function takes time polynomial in the size of input. 

For clarity w.r.t. to the score of the committee, consider the following example: W.l.o.g., if we assume that $\mathtt{f}$ is $k$-Borda, then $\mathbf{s} = (2m+(2m \cdot (2\mu^2-7\mu+3))-1,\dots, 1, 0) $. Hence, all candidates $c \in C$ get the same Borda score $\mathtt{f}(c)$ of 

\[(2m+(2m \cdot (2\mu^2-7\mu+3))-1 + \dots + 1 + 0\]

\[=\frac{(2m+(2m \cdot (2\mu^2-7\mu+3))-1)\cdot(2m+(2m \cdot (2\mu^2-7\mu+3)))}{2}\]

\[=  8\mu^4m^2 - 56\mu^3m^2 + 130\mu^2m^2 - 112\mu m^2 + 32m^2 - 2\mu^2m + 7\mu m -4m\]

which is the sum of first $2m+(2m \cdot (2\mu^2-7\mu+3))-1$ natural numbers, all the scores in the scoring vector of Borda rule. 
Therefore, each $2k + 2m \mu^2 - 6m\mu$-sized committee will be the highest scoring committee $W \in \mathcal{W}$ with a $\mathtt{f}(W)$ of 

\[(2k + 2m \mu^2 - 6m\mu) \cdot (8\mu^4m^2 - 56\mu^3m^2 + 130\mu^2m^2 - 112\mu m^2 + 32m^2 - 2\mu^2m + 7\mu m -4m)\]

Hence, the NP-hardness of the problem is due to finding a feasible committee that satisfies for all $G \in \mathcal{G}$, $|G \cap W|\geq l^D_G$ where $l^D_G=1$. Therefore, for the proof of correctness, we show the following:

\begin{claim}
We have a vertex cover $S$ of size at most $k$ that satisfies $e \cap S\neq\phi$ for all $e \in E$ if and only if we have a committee $W$ of size at most $2k + 2m \mu^2 - 6m\mu$ that satisfies all the diversity constraints, which means that for all $G \in \mathcal{G}$, $|G \cap W|\geq l^D_G$ which equals $|G \cap W|\geq 1$ as $l^D_G=1$ for all $G \in \mathcal{G}$. 
\end{claim}

($\Rightarrow$) If the instance of the VC problem is a yes instance, then the corresponding instance of \DiReCWD is a yes instance as each and every candidate group will have at least one of their members in the winning committee $W$, i.e., $|G \cap W|\geq 1$ for all $G \in \mathcal{G}$. Note that we have set $l^D_G=1$ for all $G \in \mathcal{G}$. 

More specifically, for each block of candidates, we select one dummy candidate from Set $T_1$ and all $\mu-1$ dummy candidates from Set $T_3$. This helps to satisfy the condition $|G \cap W|\geq 1$ for all candidate groups that contain at least one dummy candidate $d \in D$. Overall, we select $\mu$ candidates from $\mu-3$ blocks for each of the $2m$ candidates that correspond to vertices in the vertex cover. This results in $(\mu \cdot (\mu-3) \cdot 2m) = 2m \mu^2 - 6m\mu$ candidates in the committee. Next, for groups that do not contain any dummy candidates, select $2k$ candidates $c \in A$ that correspond to $k$ vertices $x \in X$ that form the vertex cover. These candidates satisfy the constraints. Specifically, these $2k$ candidates satisfy $|G \cap W|\geq 1$ for all the candidate groups that do not contain any dummy candidates. Hence, we have a committee of size $2k + 2m \mu^2 - 6m\mu$.

($\Leftarrow$) 
The instance of the \DiReCWD is a yes instance when we have $2k + 2m \mu^2 - 6m\mu$ candidates in the committee. This means that each and every group will have at least one of their members in the winning committee $W$, i.e., $|G \cap W|\geq 1$ for all $G \in \mathcal{G}$. Then the corresponding instance of the VC problem is a yes instance as well. This is because the $k$ vertices $x \in X$ that form the vertex cover correspond to the $2k$ candidates $c \in A$ that satisfy $|G \cap W|\geq 1$ for all the candidate groups that do not contain any dummy candidates. We remind that we had constructed $2m$ candidates in the instance of \DiReCWD problem that correspond to $m$ vertices in the VC problem, which means that we need $2k$ candidates instead of $k$ candidates to satisfy diversity constraints for candidate groups that do not contain any dummy candidates. This completes the proof.
\end{proof}

\subsubsection{\texorpdfstring{\DiReCWD}{} w.r.t. representation constraints}
We now study the computational complexity of \DiReCWD due to the presence of voter attributes. Note that the reduction is designed to conform to the real-world stipulations that are analogous to the stipulations for the candidate attributes. 
The following theorem helps us prove the statement in Corollary~\ref{cor:DiReClass}(4).
\begin{theorem}\label{lemma:DiReCWDrep01}
If $\mu=0$, $\forall \pi \in \mathbb{Z} : \pi\geq 1$, and   $\mathtt{f}$ is a monotone, separable function, then \DiReCWD 
is NP-hard, even when $\forall P \in \mathcal{P}$, $l_P^R=1$.

\end{theorem}


\begin{proof}
We reduce an instance of vertex cover (VC) problem to an instance of \DiReCWD. We have one candidate $c_i$ for each vertex $x_i \in X$, and $n\cdot m$ dummy candidates $d \in D$ where $n$ corresponds to the number of edges and $m$ corresponds to the number of vertices in the graph $G$. Formally, we set $A$ = \{$c_1, \dots, c_m$\} and the dummy candidate set $D$ = \{$d_1, \dots, d_{n m}$\}. Hence, the candidate set $C$ = $A \cup D$ consists of $m+(n\cdot m)$ candidates. We set the target committee size to be $k$.

We now introduce $n^2$ voters, $n$ voters for each edge $e \in E$. More specifically, an edge $e \in E$ connects vertices $x_i$ and $x_j$. Then, the corresponding $n$ voters $v \in V$ rank the candidates in the following collection of sets $\mathcal{T} = (T_1$, $T_2$, $T_3$, $T_4)$ such that $T_1\succ T_2\succ T_3 \succ T_4$:

\begin{itemize}
    \item \textbf{Set $T_1$}: candidates $c_i$ and $c_j$ that correspond to vertices $x_i$ and $x_j$ are ranked at the top two positions, ordered based on their indices. For $a^{\text{th}}$ voter where $a \in [n]$, we denote the candidates $c_i$ and $c_j$ as $c_{i_a}$ and $c_{j_a}$.
    \item \textbf{Set $T_2$}: $m$ out of ($n \cdot m$) dummy candidates are ranked in the next $m$ positions, again ordered based on their indices. For each voter, these $m$ candidates are distinct as shown below. Hence, for all pairs of voters $v, v' \in V : v\neq v'$, we know that $T_2^v \cap T_2^{v'}=\phi$.
    \item \textbf{Set $T_3$}: the next $m-2$ positions are occupied by the remaining $m-2$ candidates in $A \setminus \text{Set }T_1$ that correspond to the vertices in graph $G$, ordered based on their indices.
    \item \textbf{Set $T_4$}: the last $(n-1)\cdot m$ positions are occupied by the remaining $(n-1)\cdot m$ dummy candidates in $D \setminus \text{Set }T_2$, ordered based on their indices.
\end{itemize}

More specifically, the voters rank the candidates as shown below:
\begin{table}[H]
\centering
\resizebox{\textwidth}{!}{\begin{tabular}{c|lllllllllllll}
\textbf{Voters}&\multicolumn{1}{c}{\textbf{Set $T_1$}} & \multicolumn{1}{c}{$\succ$} & \multicolumn{7}{c}{\textbf{Set $T_2$}} & \multicolumn{1}{c}{$\succ$} & \multicolumn{1}{c}{\textbf{Set $T_3$}}& \multicolumn{1}{c}{$\succ$} & \multicolumn{1}{c}{\textbf{Set $T_4$}}\\
\hline\\
$v_1^{1}$, \dots, $v_1^{n}$&$ c_{i_1} \succ c_{j_1} $ & $\succ$ & $ d_1 $&$\succ$&$ d_2 $&$\succ $&$\dots$&$ \succ$&$ d_{m} $ & $\succ$ & $A \setminus \{c_{i_1}, c_{j_1}\} $ &  $\succ$ & $ D \setminus \{d_1, \dots, d_m\}$  \\[5pt]

$v_2^{1}$, \dots, $v_2^{n}$&$ c_{i_2} \succ c_{j_2}$ & $\succ$ & $ d_{m+1} $&$\succ$&$ d_{m+2} $&$\succ$&$ \dots $&$\succ $&$d_{2m} $ & $\succ$ & $A \setminus \{c_{i_2}, c_{j_2}\} $&  $\succ$ & $D \setminus \{d_{m+1}, \dots, d_{2m}\}$ \\[5pt]

$v_3^{1}$, \dots, $v_3^{n}$&$ c_{i_3} \succ c_{j_3}$ & $\succ$ & $ d_{2m+1} $&$\succ$&$ d_{2m+2} $&$\succ$&$ \dots $&$\succ$&$ d_{3m} $ & $\succ$ & $A \setminus \{c_{i_3}, c_{j_3}\} $ &  $\succ$ & $D \setminus \{d_{2m+1}, \dots, d_{3m}\}$  \\[5pt]
\vdots&&&&&&&&&&&&\\[5pt]

$v_n^{1}$, \dots, $v_n^{n}$&$ c_{i_n} \succ c_{j_n}$ & $\succ$ & $ d_{(n-1)m+1} $&$\succ$&$ d_{(n-1)m+2} $&$\succ$&$ \dots $&$\succ$&$ d_{nm} $ & $\succ$ & $A \setminus \{c_{i_n}, c_{j_n}\} $ &  $\succ$ & $D \setminus \{d_{(n-1)m+1}, \dots, d_{nm}\}$  \\[5pt]
\end{tabular}}
\label{tab:voterRankingsRepCons}
\end{table}

Next, there are no candidate attributes, and hence, $\mu=0$ and there are no diversity constraints ($l^D_G=\phi)$. The voters are divided into disjoint population over one or more attributes when $\forall \pi \in \mathbb{Z}, \pi\geq1$. Specifically, the voters are divided into populations as follows: $\forall x \in [\pi]$, $\forall y \in [n]$, $\forall z \in [n]$, voter $v_y^z \in V$ is part of a population $P\in\mathcal{P}$ such that $P$ contains all voters with the same $z \mod x$ and $y$. 
Each voter is part of $\pi$ populations.
We set the representation constraint to 1. Hence, $l^R_P=1$ for all $P \in \mathcal{P}$. The winning committee $W_P$ for each population $P \in \mathcal{P}$ will always consist of the top $k$-ranked candidates in the ranking of the voters in population $P$, which means that  $W_P$, $\forall P \in \mathcal{P}$, can not contain candidates from Set $T_3$ and Set $T_4$. This is because, by construction, (a) the ranking of all voters within a population $v \in P$, for all $P \in \mathcal{P}$, is the same and (b) the first $k$ candidates of each population will only get selected because either (i) they will indeed be the highest scoring candidates for the population or (ii) in case of a tie, they get precedence because we break ties based on the indices of candidates such that $c_i$ gets precedence over $c_j$ for all $i<j$.



This completes our reduction, which is a polynomial time reduction in the size of $n$ and $m$. For the proof of correctness, we show the following:
\begin{claim}
We have a vertex cover $S$ of size at most $k$ that satisfies $e \cap S\neq\phi$ for all $e \in E$ if and only if we have at least one committee $W$ of size at most $k$ that satisfies all the representation constraints, which means that for all $P \in \mathcal{P}$, $|W_P \cap W|\geq l^R_P$ which equals $|W_P \cap W|\geq 1$ as $l^R_P=1$ for all $P \in \mathcal{P}$. 
\end{claim}

($\Rightarrow$) If the instance of the VC problem is a yes instance, then the corresponding instance of \DiReCWD is a yes instance as each and every population's winning committee, $W_P$ for all $P \in \mathcal{P}$, will have at least one of their members in the winning committee $W$, i.e., $|W_P \cap W|\geq 1$ for all $P \in \mathcal{P}$. Indeed, even had the winning committee of each population been of size 2 instead of $k$, the instance of \DiReCWD will be a yes instance as the vertex cover corresponds to the winning committee representing each and every population as $|W_P \cap W| \geq 1$ for all $P \in \mathcal{P}$. 

($\Leftarrow$) The instance of the \DiReCWD is a yes instance when each and every population's winning committee, $W_P$ for all $P \in \mathcal{P}$, will have at least one of their members in the winning committee $W$, i.e., $|W_P \cap W|\geq 1$ for all $P \in \mathcal{P}$. Then the corresponding instance of the VC problem is a yes instance as well. More specifically, there are two cases when the instance of the \DiReCWD can be a yes instance:
\begin{itemize}
    \item \textbf{Case 1 - When only the candidates from Set $T_1$ are in the committee $W$:} An instance of \DiReCWD when $\mu=0$ and $\pi=1$ is a yes instance when each and every population has at least one representative in the committee, i.e., $|W_P \cap W| \geq 1$ for all $P \in \mathcal{P}$. We note that for all $P \in \mathcal{P}$, each population's winning committee $W_P$ consists of two candidates from Set $T_1$ and top $k-2$ candidates from Set $T_2$. Hence, when the winning committee $W$ consists of only the candidates from Set $T_1$ of the ranking of each and every voter $v \in V$, it implies that it will be a yes instance, which in turn, implies that there is a vertex cover of size at most $k$ that covers all the edges $e \in E$ because the vertices in vertex cover $x \in S$ correspond to the candidates in the winning committee $c\in W$.
    \item \textbf{Case 2 - When candidates from Set $T_1$ and Set $T_2$ are in the committee $W$:} In Case 1, we showed that if a candidate $c$ in the winning committee $W$ is from Set $T_1$, then it corresponds to a vertex in the vertex cover. Additionally, as the population's winning committee $W_P$ for all $P \in \mathcal{P}$ is of size $k$, an instance of \DiReCWD can be a yes instance even if a dummy candidate from Set $T_2$ is in the winning committee $W$. More specifically, there are two sub-cases:
    \begin{itemize}
        \item \textbf{for some population $P\in\mathcal{P}$, dummy candidate $d$ from Set $T_2$ AND candidate $c$ from from Set $T_1$ are in the committee $W$:} if a population's candidate $c$ from Set $T_1$, who is also in $W_P$, is in $W$, then this sub-case is equivalent to Case 1, and hence, a corresponding vertex in the vertex cover $v \in S$ exists. We note that this sub-case does not allow for any of population to have a representative from $W_P$ in $W$ \emph{only} from Set $T_2$, which is our next sub-case.
        \item \textbf{for some population $P\in\mathcal{P}$, only dummy candidate $d$ from Set $T_2$ is in the committee $W$:} if for a given population $P \in \mathcal{P}$, a committee $W$ represents the population via only the dummy candidate $d$ who is in a population's winning committee $d \in W_P$, then the representation constraint $l^R_P=|W_P \cap W|= 1$ is satisfied as $W_P \cap W = \{d\}$. However, for all pairs of voters $v, v' \in V : v\neq v'$, we know that $T_2^v \cap T_2^{v'}=\phi$. 
        Hence, we can replace any such dummy candidate $d \in W_P$ with a candidate $c \in W_P$ as that candidate $d$ can not be representing any other population $P' \in \mathcal{P}\setminus P$. Formally, a winning committee $W$ is always tied\footnote{W.l.o.g., we make a subtle assumption that all $m + (n\cdot m)$ candidates bring the same utility to the committee $W$. The aim to make this assumption is to show that even under this assumption, the problem remains hard, which is to say that even finding a \emph{feasible} committee that simply satisfies the constraints is NP-hard even when we have $\pi=1$. The assumption does not change the composition of each population's winning committee $W_P$ for all $P\in \mathcal{P}$.} to another winning committee $W'$ where $W'$=$(W\setminus \{d\}) \cup \{c\}$ where $\{c,d\} \in W_P$ for some $P \in \mathcal{P}$. This is equivalent to saying that we are replacing candidate $d$ from Set $T_2$ with a candidate $c$ from Set $T_1$ of the population $P$. Thus, a yes instance of \DiReCWD due to $W$, or due to the equivalent committee $W'$, in this sub-case corresponds to a vertex cover $S$ that covers all the edges $e \in E$.
        
    \end{itemize}
\end{itemize}
These cases complete the other direction of the proof of correctness.
    
Finally, we note that for this reduction and the proof of correctness, we assume the ties are broken using a predecided order of candidates. We also note that as we are using a separable \csr, computing scores of candidates takes polynomial time. This completes the overall proof.
\end{proof}

\begin{corollary}\label{lemma:DiReCWDrep}
If $\forall \mu \in \mathbb{Z} : \mu\geq0$, $\forall \pi \in \mathbb{Z} : \pi\geq1$, and   $\mathtt{f}$ is a monotone, separable function, then \DiReCWD 
is NP-hard, even when $\forall G \in \mathcal{G}$, $l_G^D=1$ and $\forall P \in \mathcal{P}$, $l_P^R=1$.
\end{corollary}

The reduction in the proof of Theorem~\ref{lemma:DiReCWDrep01} holds $\forall \pi \in \mathbb{Z} : \pi\geq1$ as each voter in the reduction can belong to more than one population. Next, as focus of this section was to understand the computational complexity with respect to representation constraints, we ease the stipulation that required each candidate attribute to partition all candidates into \emph{more than two} groups. Hence, for each candidate attribute $A_i$, $\forall i \in [\mu]$, we simply create one group that consists of \emph{all} the candidates and set $l_G^D=1$ for all $G\in\mathcal{G}$ and the problem still remains NP-hard. 

\subsubsection{\texorpdfstring{\DiReCWD}{} w.r.t. submodular scoring function}

Chamberlin-Courant (CC) rule is a well-known monotone, submodular scoring function \cite{celis2017multiwinner}, which we use for our proof. The novelty of our reduction is that it
holds for determining the winning committee using CC rule that uses \emph{any} positional scoring rule with scoring vector $\mathbf{s} = \{s_1,\dots,s_m\}$ such that $s_1 = s_2$, $s_m\geq0$, and $ \forall i \in [3, m-1], s_i \in \mathbb{Z}$ : $s_i\geq s_{i+1}$ and $s_2 > s_i$.


The following theorem and corollary proves the statement in Corollary~\ref{cor:DiReClass}(1). 
\begin{theorem}\label{lemma:DiReCWDsubmod}
If   $\mathtt{f}$ is a monotone, \textbf{submodular} function, then \DiReCWD 
is NP-hard even when $\mu=0$ and $\pi=0$.
\end{theorem}

\begin{proof}
We reduce an instance of vertex cover (VC) problem to an instance of \DiReCWD. Each candidate $c_i \in C$ corresponds to a vertex $x_i \in X$. For each edge $e \in E$, we have a voter $v \in V$ whose complete linear order is as follows: the top two most preferred candidates correspond to the two vertices connected by an edge $e$. These two candidates are ranked based on their indices. The remaining $m-2$ candidates are ranked in the bottom $m-2$ positions, again based on their indices. We set the committee size to $k$. This is a polynomial time reduction in the size of $n$ and $m$.

For the proof of correctness, we note that there are no candidate and voter attributes, and thus, no diversity and representation constraints. Hence, we show the following:
\begin{claim}
We have a vertex cover $S$ of size at most $k$ that satisfies $e \cap S\neq\phi$ for all $e \in E$ if and only if we have a committee $W$ of size at most $k$ with total misrepresentation of zero, which means that at least one of the top 2 ranked candidates of each voter is in the committee $W$. 
\end{claim}

($\Rightarrow)$ If the instance of the VC problem is a yes instance, then the corresponding instance of \DiReCWD is a yes instance as each and every voter will have at least one of their top two candidates in the committee and this will result in a misrepresentation score of zero as $s_1 = s_2$ and $\forall i \in [3,m], s_i < s_2$. 

($\Leftarrow$) If the instance of \DiReCWD is a yes instance, then the VC is also a yes instance. When a committee $W$ does not represent a voter's one of the top-2 candidates, it implies that the dissatisfaction is greater than zero. Hence, for each voter $v \in V$ that is not represented, there exists an edge $e \in E$ that is not covered. Hence, we can say that \DiReCWD is NP-hard with respect to $\mu=0$, $\pi=0$ and $\mathtt{f}=$ submodular function.
\end{proof}

\begin{corollary}\label{cor:DiReCWDsubmod}
If $\forall \mu \in \mathbb{Z} : \mu\geq0$, $\forall \pi \in \mathbb{Z} : \pi\geq0$, and   $\mathtt{f}$ is a monotone, \textbf{submodular} function, then \DiReCWD 
is NP-hard, even when $\forall G \in \mathcal{G}$, $l_G^D=1$ and $\forall P \in \mathcal{P}$, $l_P^R=1$.
\end{corollary}

The proof of Theorem~\ref{lemma:DiReCWDsubmod} shows that when we use a submodular but not separable \csr $\mathtt{f}$, \DiReCWD is NP-hard even when $\mu=0$ and $\pi=0$. Next, as focus of this section was to understand the computational complexity with respect to monotone submodular scoring rule, we ease the stipulation that required each candidate attribute to partition all candidates into \emph{more than two} groups and required each voter attribute to partition all voters into \emph{more than two} population. The problem remains hard even when we have candidate attributes and the diversity constraints are set to one and have voter population and the representation constraints set to one. Specifically, for each candidate attribute, create one group that contains \emph{all} the candidates and for each voter attribute, create one population that contains \emph{all} the voters. This is analogous to \emph{not} having any candidate or voter attributes. Hence, even when $l^D_G=1$ for all $G\in \mathcal{G}$ and $l^R_P=1$ for all $P\in\mathcal{P}$, \DiReCWD 
is NP-hard if  $\mathtt{f}$ is submodular \csr.

\section{Inapproximability and Parameterized Complexity}
\label{sec:apxTract}

\begin{savenotes}
\begin{table*}[t!]
\centering

\resizebox{\textwidth}{!}{\begin{tabular}{|l|l||c|c|c|}
\hline
Result & Parameter & ($\geq3$,0) & (0,$\geq1$) & ($\geq1$,$\geq1$)\\
\hline\hline
\multirow{1}{*}{inapproximability} & - & $(1-\varepsilon)\cdot$ $(\ln \mu - \mathcal{O}(\ln \ln \mu))$ (Thm.~\ref{thm:apxtract/30apx}) & $k-\varepsilon$ (Thm.~\ref{thm:apxtract/01apx})\footnote{For Theorem~\ref{thm:apxtract/01apx}, we assume that the Unique Games Conjecture (UGC) \cite{khot2002power} holds, specifically as the result that showed pseudorandom sets in the Grassmann graph have near-perfect expansion completed the proof of 2-to-2 Games Conjecture \cite{subhash2018pseudorandom}, which is considered to be a significant evidence towards proving the UGC. Moreover, GapUG($\frac{1}{2}$, $\varepsilon$) is found to be NP-hard, i.e., a weaker version of the UGC holds with completeness $\frac{1}{2}$ (See \cite{dinur2018towards} and ``Evidence towards the Unique Games Conjecture'' in \cite{subhash2018pseudorandom} for more details). Without the assumption on UGC, the result for our problem when $\mu=0$ and $\pi\geq1$ will change and for arbitrarily small constant $\varepsilon > 0$, the problem is inapproximable within a factor of $k-1-\varepsilon$ for every integer $k\geq3$ \cite{dinur2005new} and within a factor of $\sqrt{2}-\varepsilon$ when $k=2$ \cite{subhash2018pseudorandom,khot2017independent}.} & $(1-\varepsilon)\cdot$ $\ln$ ($|\mathcal{G}|+|\mathcal{P}|$) (Thm.~\ref{thm:apxtract/11apx})\\
\hline\hline
\multirow{1}{*}{parameterized} & $k$ is constant &\multicolumn{3}{|c|}{$\mathcal{O}(m^k\cdot (|\mathcal{G}|+|\mathcal{P}|))$ (Obs.~\ref{thm:apxtract/kconstant})}\\
\cline{2-5}
complexity & $k \ll  m$  & W[2]-hard (Cor.~\ref{cor:DiReCWDFPT30}) &$\mathcal{O}(c^k+m)$ (Thm.~\ref{thm:RCWDFPT})& W[2]-hard (Cor.~\ref{cor:DiReCWDFPT11})\\
\hline
\end{tabular}}
\caption{A summary of inapproximability and parameterized complexity of \DiReCF. The value in brackets of the header row represent the values of $\mu$ and $\pi$, respectively, such that results hold for all $\mu \in \mathbb{Z}$ and all $\pi \in \mathbb{Z}$ that satisfy the condition stated in the brackets. The results are under the assumption P $\neq$ NP. `Thm.' denotes Theorem. `Obs.' denotes Observation. `Cor.' denotes Corollary. $\varepsilon$ denotes an arbitrarily small constant such that $\varepsilon>0$ and the results are meant to hold for every such $\varepsilon>0$.}
\label{tab:apxTractResults}
\end{table*}
\end{savenotes}

The hardness of \DiReCWD is mainly due to the hardness of \DiReCF, which is to say that satisfying the diversity and representation constraints is computationally hard, even when all constraints are set to 1. Formally, the hardness remains even when $l^D_G=1$ for all $G \in\mathcal{G}$ and $l^R_P=1$ for all $P\in\mathcal{P}$. Hence, in this section, we focus on the hardness of approximation to understand the limits of how well we can approximate \DiReCF and focus on parameterized complexity of \DiReCF. 

It is natural to try to reformulate representation constraints as diversity constraints. However, in our model, it is not possible to do so as each candidate attribute partitions \emph{all} $m$ candidates into groups and the lower bound is set such that $l_G^D \in [1, \min(k, |G|)]$ for all $G \in \mathcal{G}$. However, for representation constraints, $W_P$, for all $P \in \mathcal{P}$, contains only $k$ candidates and the remainder $m-k$ candidates consisting of $C \setminus W_P$, for all $P \in \mathcal{P}$, may never be selected. Hence, representation constraints can not be easily reformulated to diversity constraints. Moreover, even if we relax the lower bound of the diversity constraint to $l_G^D \in [0, \min(k, |G|)]$ instead of $l_G^D \in [1, \min(k, |G|)]$, for all $G \in \mathcal{G}$, to allow for such a reformulation, the following settings of \DiReCF and \DiReCWD are technically different and we \emph{may not} carry out any reformulations amongst each other:

\begin{itemize}
    \item Using only diversity constraints 
    \item Using only representation constraints
    \item Using both, diversity and representation, constraints 
\end{itemize}

The above listed settings are technically different from each other as the sizes of candidate groups and the size of the winning committees of populations have implications on our approach to solve a problem. For instance, using both, diversity and representation, constraints and using only representation constraints are mathematically as different as the vertex cover problem on hypergraphs and the vertex cover problem on $k$-uniform hypergraphs, respectively. The differences between the hardness of approximation for the latter two problems is well-known. Overall, while reformulations such as converting representation constraints to diversity constraints do not impact the computational complexity of the problem, it affects the approximation and parameterized complexity results. Hence, we study the hardness of approximation and the parameterized complexity of the above listed settings of \DiReCF in detail without carrying out any reformulations between the different settings of the constraints. 

\begin{observation}
\label{obs:direcfequi}$\forall \mu \in \mathbb{Z}$ and $\forall \pi \in \mathbb{Z}$, the following settings of the \DiReCF problem are not equivalent: (i) $\mu$=0 and $\pi\geq1$, (ii) $\mu\geq3$ and $\pi=0$, and (iii) $\mu\geq1$ and $\pi\geq1$.
\end{observation}

\subsection{Inapproximability} 
\label{sec:apxTract/apx}

In this subsection, we focus on allowing size violation as deciding on which constraints to violate is not straightforward, especially as constraints are linked to human groups. Hence, we define the size optimization version of \DiReCF and study its inapproximability:

\vspace{0.1cm}
\begin{definition}\textbf{\DiReCF-size-optimization:}
In the \DiReCF-size-optimization problem, 
given a set $C$ of $m$ candidates, a set $V$ of $n$ voters such that each voter $v_i$ has a preference list $\succ_{v_i}$ over $m$ candidates, a committee size $k \in [m]$, a set of candidate groups $\mathcal{G}$ and the corresponding diversity constraints $l^D_G$ for all $G \in \mathcal{G}$, and a set of voter populations $\mathcal{P}$ and the corresponding representation constraints $l^R_P$ and the winning committees $W_P$ for all $P \in \mathcal{P}$,
find a minimum-size committee $W \subseteq C$ such that $W$ satisfies all the diversity and representation constraints, i.e., $|G\cap W|\geq l^D_G$ for all $G\in \mathcal{G}$ and $|W_P\cap W|\geq l^R_P$ for all $P\in \mathcal{P}$, respectively.
\end{definition}
\vspace{0.1cm}
\begin{theorem}\label{thm:apxtract/30apx}
For   $\varepsilon > 0$, $\forall \mu \in \mathbb{Z}$ : $\mu\geq3$, and  $\pi=0$, \DiReCF-size-optimization
problem 
is 
inapproximable within $(1-\varepsilon)\cdot$ $(\ln \mu - \mathcal{O}(\ln \ln \mu))$, even when $l^D_G$ = 1  $\forall$ $G \in \mathcal{G}$.
\end{theorem}

\begin{proof}
We reduce from the set multi-cover problem with sets of bounded size, a known NP-hard problem \cite{garey1979computers}, to \DiReCFp{$\mu$}{$0$}-size-optimization
problem.  

More specifically, given a set $X$ = $\{v_1, ..., v_{|\mathcal{G}|}\}$, and a collection of $m$ sets $S_i$ $\subseteq$ $X$ such that $|S_i|\leq\mu$, the goal is to choose some sets of minimum cardinality covering each element $v_i$.


Then, we construct a \DiReCFp{$\mu$}{$0$}-size-optimization instance. To do so, we have a corresponding candidate $c_i$ for each set $S_i$, and a corresponding group $G \in \mathcal{G}$ which is equal to $\{c_j : v_i \in S_j\}$ for each element $v_i$. Hence there are $m$ candidates and $|\mathcal{G}|$ candidate groups such that each candidate belongs to at most $\mu$ groups. The diversity constraints $l^D_G$ are set to be equal to 1, which corresponds to the requirement that each element is covered. 

This is an approximation-preserving reduction for all $\mu\geq3$ and $\pi=0$. Hence, the minimum cardinality of the constrained set cover problem is at most $k$ if and only if an at most $k$-sized feasible committee exists. Given that set multi-cover problem is inapproximable within $(1-\varepsilon)\cdot$ $(\ln \mu - \mathcal{O}(\ln \ln \mu))$ \cite{trevisan2001non}, so is our \DiReCFp{$3$}{$0$}-size-optimization  problem. We note that this result holds for \DiReCF-size-optimization problem for all $\mu\in\mathbb{Z}:\mu\geq3$ and $\pi=0$. 
\end{proof}

While the above proof is similar in flavor to the one given in Theorem 7 (Hardness of feasibility with committee violations) of Celis \etal \cite{celis2017multiwinner}, we note that our inapproximability ratio differs from their inapproximability ratio of $(1-\varepsilon)\cdot$ \emph{$\ln$} ($|\mathcal{G}|$). This is because our ratio exploits the candidate structure where each candidate is bounded by the number of attributes $\mu$, which bounds the number of groups they can be a part of. Hence, our reduction is from set cover problem where each set is of bounded size.

\vspace{0.1cm}

For our next result, we first give a reduction from regular hitting set (HS) to \DiReCFp{$1$}{$1$}. Next, as the regular HS problem is equivalent to the minimum set cover problem \cite{niedermeier2003efficient}, 
the latter's inapproximability \cite{dinur2014analytical} holds for our problem.

\begin{theorem}\label{thm:apxtract/11apx}
For   $\varepsilon > 0$,  $\forall \mu \in \mathbb{Z}$ : $\mu\geq 1$, and $\forall \pi \in \mathbb{Z}$ : $\pi\geq1$, \DiReCF-size-optimization problem 
is 
inapproximable within a factor of $(1-\varepsilon)\cdot$ \emph{$\ln$} ($|\mathcal{G}|+|\mathcal{P}|$), even when $l^D_G$ = 1  $\forall$ $G \in \mathcal{G}$ and $l^R_P$ = 1  $\forall$ $P \in \mathcal{P}$.
\end{theorem}


\begin{proof}
We reduce from regular hitting set (HS), a known NP-hard problem \cite{garey1979computers}, to \DiReCFp{1}{1}-size-optimization problem. 

An instance of HS consists of a universe $U$ = $\{x_1,x_2,\dots,x_{m}\}$ and a collection $\mathcal{Z}$ of subsets of $U$, each of size $\in$ $[1,m]$. The objective is to find a subset $S \subseteq U$ of size at most $k$ that ensures for all $T \in \mathcal{Z}$, $|S \cap T|\geq1$. 

We construct the \DiReCFp{$1$}{$1$}-size-optimization instance as follows. For each element $x$ in the universe $U$, we have the candidate $c$ in the candidate set $C$. For each subset $T$ in collection $\mathcal{Z}$, we either have candidate group $G \in \mathcal{G}$ or  winning committee $W_P$ of population $P\in \mathcal{P}$. Note that we have $|\mathcal{G}|+|\mathcal{P}|=|\mathcal{Z}|$. We set $l^D_G=1$ for all $G \in \mathcal{G}$ and $l^R_P=1$ for all $P \in \mathcal{P}$, which means $|W \cap G|\geq$1 and $|W \cap W_P|\geq$1, respectively. This corresponds to the requirement that $|S \cap T|\geq1$. 

Hence, we have a subset $S$ of size at most $k$ that satisfies $|S \cap T|\geq1$ if and only if we have a committee $W$ of size at most $k$ that satisfies $|W \cap G|\geq$1 for all $G \in \mathcal{G}$ and $|W \cap W_P|\geq$1 for all $P \in \mathcal{P}$. 

We note that this is also an \emph{approximation-preserving} reduction for all $\mu\geq 1$ and $\pi\geq 1$. Given that minimum set cover problem, which is equivalent to the hitting set problem, is inapproximable within $(1-\varepsilon)\cdot$ \emph{$\ln$} ($|\mathcal{G}|+|\mathcal{P}|$) \cite{dinur2014analytical}, so is our \DiReCFp{$1$}{$1$}-size-optimization problem. We note that this result holds for \DiReCF-size-optimization problem  $\forall \mu\in\mathbb{Z}:\mu\geq1$ and $\forall \pi\in\mathbb{Z}:\pi\geq1$. 
\end{proof}

Assuming the Unique Games Conjecture \cite{khot2002power}, Bansal and Khot~\cite{bansal2010inapproximability} showed that vertex cover problem on $k$-uniform hypergraphs, for any integer $k \geq 2$, is inapproximable within $k - \varepsilon$, even when the $k$-uniform hypergraph is almost $k$-partite. We use this result for our next theorem.

\begin{theorem}\label{thm:apxtract/01apx}
For   $\varepsilon > 0$, $\mu=0$, and $\forall \pi \in \mathbb{Z}$ : $\pi\geq 1$, \DiReCF-size-optimization problem, 
assuming the Unique Games Conjecture \cite{khot2002power}, is inapproximable within $k-\varepsilon$, even when $l^R_P=1$  $\forall$ $P \in \mathcal{P}$.
\end{theorem}

\begin{proof}
We give a reduction from vertex cover problem on $k$-uniform hypergraphs to \DiReCFp{0}{$1$}. 

An instance of vertex cover problem on $k$-uniform hypergraphs consists of a set of vertices $X$ = $\{x_1,x_2,\dots,x_{m}\}$ and a set of $n$ hyperedges $S$, each connecting exactly $k$ vertices from $X$. A vertex cover $X' \subseteq X$ is a subset of vertices such that each edge contains at least one vertex from $X$ (i.e. $s\cap X'\neq\phi$ for each edge $s\in S$). The vertex cover problem on $k$-uniform hypergraphs is to find a vertex cover $X'$ of size at most $d$. 

We construct the \DiReCFp{$0$}{$1$} instance as follows. For each vertex $x \in X$, we have the candidate $c \in C$. For each edge $s \in S$, we have a population's winning committee $W_P$ of size $k$ for all $P \in \mathcal{P}$. Note that we have $|\mathcal{P}|=|S|$. We set $l^R_P=1$ for all $P \in \mathcal{P}$, which means $|W \cap W_P|\geq$1. This corresponds to the requirement that $s\cap X'\neq\phi$. 

Hence, we have a vertex cover $X'$ of size at most $d$  if and only if we have a committee $W$ of size at most $d$ that satisfies $|W \cap W_P|\geq$1 for all $P \in \mathcal{P}$. 

This is an \emph{approximation-preserving} reduction for $\mu=0$ and for all $\pi\geq1$. Given that the vertex cover problem on $k$-uniform hypergraphs is inapproximable within $k-\varepsilon$ \cite{bansal2010inapproximability} assuming the Unique Games Conjecture, so is our \DiReCFp{$0$}{$1$}-size-optimization problem. We note that this result holds for \DiReCF-size-optimization problem for all $\mu=0$ and $\pi\in\mathbb{Z}:\pi\geq1$. 
\end{proof}

In addition to this general inapproximability result, we informally conjecture that improve upon the ratio of $k-\varepsilon$.

\begin{conjecture}
\label{con:DiReCF}
[Informal] If $\mu=0$ and $\forall \pi \in \mathbb{Z}, \pi \geq 1$, then \DiReCF-size-optimization problem can be approximated to at most $k-(1-o(1))\frac{k(k-1)\ln\ln g(\phi)}{\ln g(\phi)}$ using a polynomial time algorithm.
\end{conjecture}

A proof of above conjecture implies that there exists a polynomial time approximation algorithm for the \DiReCF-size-optimization problem ($\mu=0$ and $\pi \geq 1$) with approximation ratio at most 
$k-(1-o(1))\frac{k(k-1)\ln\ln g(\phi)}{\ln g(\phi)}$ where $g(\phi)$ is a function that maps the cohesiveness of the preferences $\phi$ to the maximum number of winning committees $W_P$ that a candidate can belong to. Specifically, if such a $g(\phi)$ exists and if $\pi=1$, then the stated approximation ratio exists directly due to Halperin \cite{halperin2002improved}.

\subsection{Parameterized Complexity} 
\label{sec:apxTract/FPT}
In most real-world elections, the committee size $k$ is constant. Hence, our first result here is inspired by the parameterized complexity results in this field \cite{procaccia2008complexity,yang2018parameterized}.

\begin{observation}\label{thm:apxtract/kconstant}
The \DiReCF problem can be solved in $\mathcal{O}(m^k\cdot (|\mathcal{G}|+|\mathcal{P}|))$. If $k$ is a constant, then it is a polynomial time algorithm.
\end{observation}
We select a set of committees $\mathcal{W}$, each of size $k$, and then check for the satisfiability of the constraints for each committee $W \in \mathcal{W}$. It is easy to see that $\mathcal{W}$ has $  \binom{m}{k}$ committees, that is, $|\mathcal{W}|\leq m^k$. 
Checking whether a committee $W \in \mathcal{W}$ satisfies all the constraints takes $\mathcal{O}(\mathcal{|G|+|P|})$, which is the total number of constraints to be checked. Hence, we can solve \DiReCF in time polynomial in $m$ and $n$, given $k$ is constant.

Next, when the committee size ($k$) is not a constant, the rate of growth of the number of candidates to be elected may be much slower than the number of candidates ($k\ll m$).

\begin{theorem}\label{thm:HSW2hard}\cite{niedermeier2003efficient, downey2012parameterized}
The regular hitting set problem with unbounded subset size is W[2]-hard w.r.t. $k$.
\end{theorem}

\begin{corollary}\label{cor:DiReCWDFPT30}
If  $\forall \mu\in\mathbb{Z}$ : $\mu\geq3$ and $\pi=0$, then \DiReCF problem 
is W[2]-hard w.r.t. $k$ and the hardness holds even when $l^D_G$ = 1  $\forall$ $G \in \mathcal{G}$.
\end{corollary}
\begin{proof}
In the proof for Theorem~\ref{thm:apxtract/30apx}, we gave a reduction from minimum set cover problem to \DiReCF problem w.r.t. $\mu$ and $\pi$ for all $\mu\in\mathbb{Z}$ : $\mu\geq3$ and $\pi=0$. Additionally, we know that the minimum set cover problem has a well-known one to one relationship with the hitting set problem with no restriction on the subset size \cite{niedermeier2003efficient, ausiello1980structure, crescenzi1995compendium}. Hence, as regular HS with unbounded size of subsets is W[2]-hard \cite{niedermeier2003efficient}, our results here follow due to the one to one relationship between the regular HS problem and the minimum set cover problem.
\end{proof}

\begin{corollary}\label{cor:DiReCWDFPT11}
If $\forall \mu \in\mathbb{Z}$ : $\mu\geq1$ and $\forall \pi \in\mathbb{Z}$ : $\pi\geq1$, then \DiReCF problem 
is W[2]-hard w.r.t. $k$ and the hardness holds even when $l^D_G$ = 1  $\forall$ $G \in \mathcal{G}$ and $l^R_P$ = 1  $\forall$ $P \in \mathcal{P}$.
\end{corollary}

When $\mu\geq1, \pi\geq1$, our problem is equivalent to regular HS (Theorem~\ref{thm:apxtract/11apx}).


\begin{theorem}\label{thm:RCWDFPT}
If $\mu=0$, $\forall \pi \in \mathbb{Z}$: $\pi\geq1$, and $l^R_P=1, \forall P \in \mathcal{P}$, then \DiReCF problem can be solved using an $\mathcal{O}(c^k+m)$ time algorithm where $c=d-1+\mathcal{O}(d^{-1})$ and $d = k$. If $k\ll m$,  then it is a polynomial time algorithm.
\end{theorem}
\begin{proof}
Proof of Theorem~\ref{thm:apxtract/01apx} shows that our problem is equivalent to $k$-HS when $\mu=0$ and $ \pi\geq1$. Hence, our algorithm here is motivated from bounded tree search algorithm in Section 6 of \cite{niedermeier2003efficient} where they showed that when $k$ is small, a $d$-hitting set problem, which upper bounds the cardinality of every element in the subsets to be hit to $d$, can be solved using an $\mathcal{O}(c^k+m)$ time algorithm with $c=d-1+\mathcal{O}(d^{-1})$. In our case, $d=k$. We have  modified our algorithm from \cite{niedermeier2003efficient} to return \emph{all} committees that satisfy the representation constraints.

\begin{algorithm}[H]
\caption{Parameterized Polynomial-time Algorithm for Representation Constraints}
\label{alg:algorithmFPT}
\textbf{Input}: $C$, $k$, and $W_P$ and $l^R_P$ for all $P\in\mathcal{P}$ \\
\textbf{Output}: $W$ : $|W\cap W_P|\geq1$ for all $P\in\mathcal{P}$ \\

\begin{algorithmic}[1] 
\STATE $C=C\setminus\{y\}$ : $\forall x,y \in C$, if $y\in W_P$, then $x \in W_P$, for all $ P \in \mathcal{P}$\\
\FOR{\textbf{each} $W_P$}
\STATE Create $k$ branches, one each for each $c_i \in W_P$ 
\STATE Choose $c_1$ for the hitting set and choose that $c_1$ is not in the hitting set, but $c_i$ is for all $i \in [2,k]$
\ENDFOR




\end{algorithmic}
\end{algorithm}

In the above algorithm, steps 3 and 4 creates $k$ branches in total. Hence, if the number of leaves in a branching tree is $b_k$, then the first branch has at most $b_{k-1}$ leaves. Next, let $b'_k$ be the number of leaves in a branching tree where there is at least one set of size $k-1$ or smaller. For each $i \in [2,k]$, there is some committee $W_P$ in the given collection such that $c_1 \in W_P$, but $c_i \notin W_P$. Therefore, the size of $W_P$ is at most $k-1$ after excluding $c_1$ from and including $c_i$ in the committee $W$. Altogether we get $b_k\leq b_{k-1}+(k-1)b'_{k-1}$.

If there is already a set with at most $k-1$ elements, we can repeat the above steps and get $b'_k\leq b_{k-1}+(k-2)b'_{k-1}$. The branching number of this recursion is $c$ from above, and note that it is always smaller than $k-1+\mathcal{O}(k^{-1})$.
\end{proof}

As a conclusion of our theoretical analyses, we make an interesting observation: When $\pi=0$, \DiReCF becomes NP-hard when $\mu = 3$. On the other hand, when $\mu=0$, \DiReCF becomes NP-hard even when $\pi = 1$. This means that introducing representation constraints makes the problem hard ``faster'' than introducing diversity constraints. In contrast, with respect to the parameter $k$, the former case is W[2]-hard and the latter is fixed parameter tractable for all $\pi \in \mathbb{Z} : \pi \geq 1$. This reinforces our claim that even if it may seem natural to try and reformulate representation constraints as diversity constraints, we should not do so as the size of candidate groups and the size of winning committee of voter populations has implications on how one may try to solve the problem efficiently.

\section{Heuristic Algorithm}
\label{sec:HeurAlgo}

In the previous sections, we saw that our model, which is useful from the social choice theory perspective to have more ``fairer'' elections, is computationally hard and it is hard even when we parameterize the problem on the size of the committee. Hence, we take a pragmatic approach to evaluate if our model is efficient in practice. We do so by developing a two-stage heuristic-based algorithm,  
in part motivated from the literature on distributed constraint satisfaction \cite{russell2002artificial}, which allow us to efficiently compute \emph{DiRe} committees in practice. 

We develop a heuristic-based algorithm as the use of integer linear program formulation in multiwinner elections is not efficient \cite{skowron2015achieving}, especially when using the Monroe rule. Moreover, in addition to the known temporal efficiency of using a heuristic approach as compared to a linear programming approach, our empirical evaluation shows that the algorithm returns an optimal solution (discussed later in Section~\ref{subsec:effOfHA}), thus overcoming one of the biggest disadvantages of using a heuristic approach.

\subsection{DiReGraphs}
We represent an instance of the \DiReCWD problem from Figure~\ref{fig:example} as a DiReGraph (Figure~\ref{fig:DiReGraph}). The constraints are represented by quadrilaterals and candidates by ellipses. More specifically, there are candidates (Level B) and the DiRe committee (Level D). Next, there is a global committee size constraint (Level A) and unary constraints that lower bound the number of candidates required from each candidate group or voter population (Level C). Edges connecting candidates (Level B) to unary constraints (Level C) depends on the candidate's membership in a candidate group or a population's winning committee. The idea behind DiReGraph is to have a ``network flow'' from A to D such that \emph{all} nodes on level C are visited. More specifically, the aim is to select $k$ candidates (Level A) from $m$ candidates (Level B) such that the in-flow at the unary constraint nodes (Level C) is equal to the specified diversity or representation constraint. A node is said to have an  in-flow of $\tau$ when $\tau$ candidates in the committee $W$ are part of the group/winning population. Formally, $\tau=|W\cap G|$ for each candidate group $G \in \mathcal{G}$ and $\tau=|W\cap W_P|$ for each population $P \in \mathcal{P}$. When the last condition is fulfilled, there will be a \emph{DiRe} committee (Level D). 

\begin{example}\textbf{Creating DiReGraph:}
Consider the election setup shown in Figure~\ref{fig:example}. The candidate $c_2$ (Figure~\ref{fig:example}) is a male who is in winning committees of both the states, namely California and Illinois. Hence, $c_2$ in DiReGraph (Figure~\ref{fig:DiReGraph}) is connected with the three sets of constraints, one each for male and the two states, namely CA (California) and IL (Illinois).
\end{example}

\begin{figure}[t!]
\centering
\begin{subfigure}{.4\textwidth}
  \centering
  \includegraphics[width=0.975\linewidth]{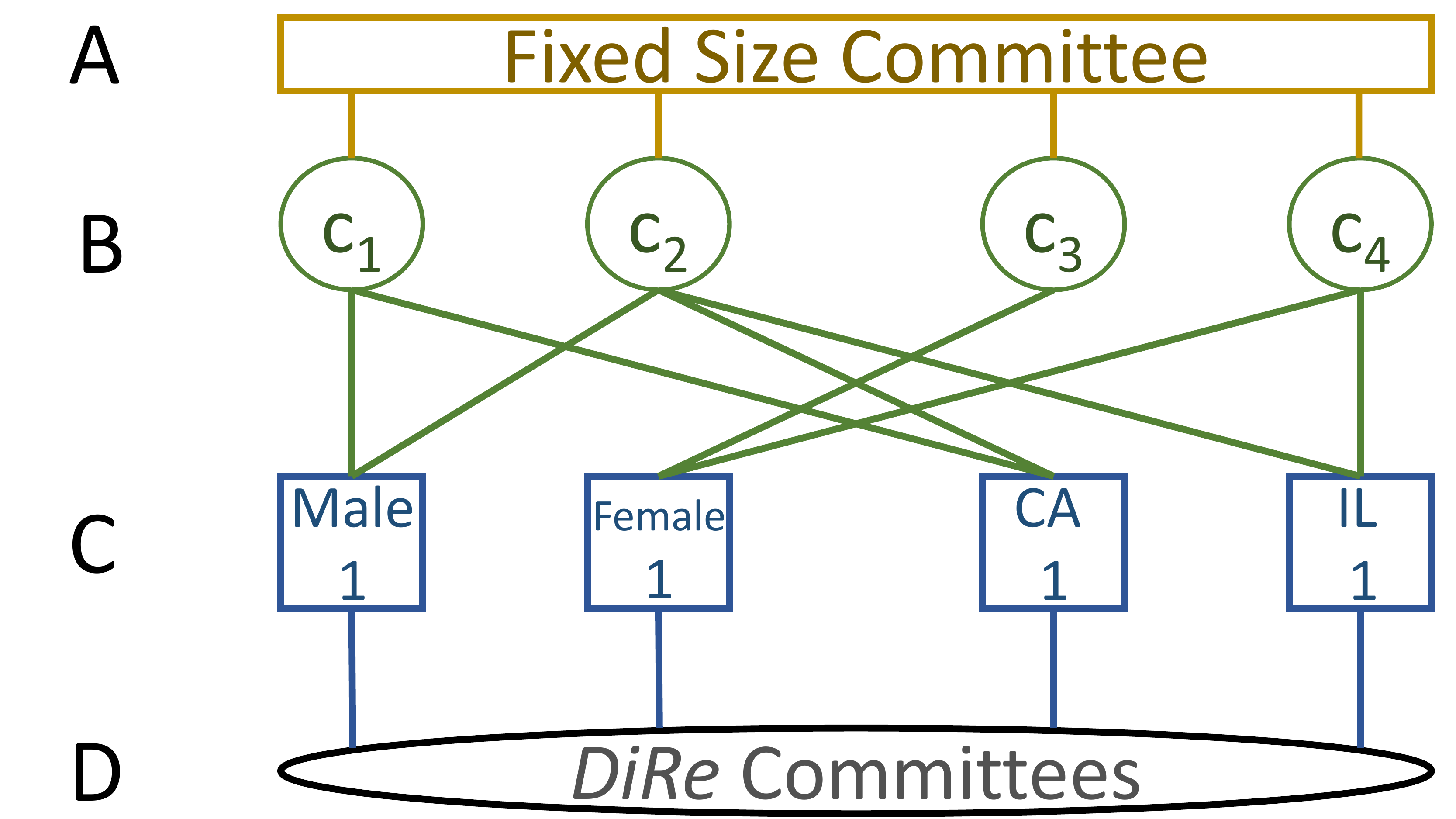}
\end{subfigure}%
\caption{\DiReCF as DiReGraph. (A) Global committee size constraint and (C) the diversity/representation constraints connected by edges with (B) the candidates and (D) the DiRe committee.}
\label{fig:DiReGraph}
\end{figure}

\begin{algorithm}[t!]
\caption{DiRe Committee Feasibility Algorithm}
\label{alg:algorithmDiReCF}
\textbf{Input}:\\
variables $X$ = \{$X_1,\dots, X_{\mathcal{|G|}+\mathcal{|P|}}$\}\\
domain $D$ = ($D_1,\dots, D_{\mathcal{|G|}+\mathcal{|P|}}$) : each $D_i$ is  $G \in$  $\mathcal{G}$ or $W_P : P \in$  $\mathcal{P}$  \\
unary constraints $S$ = \{$S_1,\dots, S_{\mathcal{|G|}+\mathcal{|P|}}$\} : each $S_i$ is $l^D_G$ for each $G \in \mathcal{G}$ or $l^R_P$ for each $P \in\mathcal{P}$ \\
\textbf{Output}:\\ set $\mathcal{W}$ of committees : $\forall W \in \mathcal{W}$, $|W\cap D_i|\geq S_i$\\

\begin{algorithmic}[1] 
\STATE Create DiReGraph $DiReG$
\STATE $SG$ = subgraph of nodes on levels  B \& C of $DiReG$
\STATE $SCC$ = strongly connected components of $SG$

\FOR{\textbf{each} $comp_i$, $comp_j$ $\in$ $SCC$}
\FOR{\textbf{each} $X_{u}$ = \{$X_i \cup$ $X_j$\} : $X_i \in$  $comp_i$ and $X_j \in$ $comp_j$}


\STATE \textbf{if} !$\mathtt{pairwise\_feasible}$($X_{u}$, $D$, $S$) \textbf{return} \emph{false}
\ENDFOR
\ENDFOR

\FOR{\textbf{each} $comp$ in $SCC$}
\STATE $X_{SCC}$ = list of $X_i$ for each $S_i$ at level C of $comp$
\STATE \textbf{if} !$\mathtt{pairwise\_feasible}$($X_{SCC}$, $D$, $S$) \textbf{return} \emph{false}
\ENDFOR
\STATE Recreate DiReGraph $DiReG$ using reduced domain
\STATE \textbf{return} $\mathtt{heuristic\_backtrack}$(\{\}, $DiReG$, $X$, $D$, $S$)
\end{algorithmic}
\end{algorithm}

\subsection{DiRe Committee Feasibility Algorithm}
Algorithm~\ref{alg:algorithmDiReCF} has two stages: (i) \emph{preprocessing} reduces the search space used to satisfy the constraints and efficiently finds infeasible instances, and (ii) \emph{heuristic-based search} of candidates decreases the number of steps needed either to find a feasible committee, or return infeasibility. 

\subsubsection{Create DiReGraph}
The first step of the algorithm is to create the DiReGraph based on the variables that are given as the input. We have the following input: variables $X$ = \{$X_1,\dots, X_{\mathcal{|G|}+\mathcal{|P|}}$\} are represented by the nodes on Level C. The domain $D$ = ($D_1,\dots, D_{\mathcal{|G|}+\mathcal{|P|}}$) of these variables are represented by edges that connect the node on Level C to the nodes (candidates) on Level B. Formally, for each $D_i \in D$ where $D_i$ is  $G \in$  $\mathcal{G}$ or $W_P : P \in$  $\mathcal{P}$, we have an edge $e$ that connect node on Level B with node on Level C.  The constraints $S$ = \{$S_1,\dots, S_{\mathcal{|G|}+\mathcal{|P|}}$\} correspond to the diversity and representation constraints. Formally, for each $S_i \in S$, $S_i$ is $l^D_G$ for each $G \in \mathcal{G}$ or $l^R_P$ for each $P \in\mathcal{P}$.

\begin{example}\textbf{Input Variables:}
$X$ = \{Male, Female, CA, IL\}

$D$ = (\{$c_1, c_2$\}, \{$c_3 c_4$\}, \{$c_1, c_2$\}, \{$c_2, c_4$\})

$S$ = \{1, 1, 1, 1\}
\end{example}

\begin{algorithm}[t!]
\caption{Pairwise Feasibility Algorithm}
\label{alg:PFA}

\textbf{function} $\mathtt{pairwise\_feasible}$($X$, $D$, $S$) \textbf{returns} false if an inconsistency is found, or true
\begin{algorithmic}[1] 
\STATE $queue$ = ($X_i$, $X_j$) : $X_i$, $X_j$ in $X$ and $X_i \neq X_j$
\WHILE{$queue$ is not empty}
\STATE ($X_i$, $X_j$) = $\mathtt{remove\_first}$($queue$)
\STATE \textbf{if} $|D_i\cap D_j|$ $< S_i + S_j - k$ 
\textbf{return} \emph{false}
\IF{$\mathtt{domain\_reduce}$($X$, $D$, $C$, $X_i$, $X_j$)}
\STATE \textbf{if} $|D_i|=0$ \textbf{return} \emph{false}
\FOR{ \textbf{each} $X_x \in$ $X$ $\setminus$ $\{X_i, X_j$\}}
\STATE add ($X_x$, $X_i$) to $queue$
\ENDFOR
\ENDIF
\ENDWHILE
\STATE \textbf{return} \emph{true}
\end{algorithmic}

\textbf{function} $\mathtt{domain\_reduce}$($X$, $D$, $S$, $X_i$, $X_j$) \textbf{returns} true iff the domain $D_i$ of $X_i$ is reduced

\begin{algorithmic}[1] 
\STATE $domain\_reduced$ = $false$
\FOR{ \textbf{each} $d \in D_i$}
\IF{all $S_i$ sized combinations from $D_i$ containing $d$ does not satisfy pairwise constraints with all $S_j$ sized combinations from $D_j$} 
\STATE $D_i$ = $D_i \setminus \{d\}$
\STATE $domain\_reduced$ = $true$
\ENDIF
\ENDFOR
\STATE \textbf{return} \emph{$domain\_reduced$}
\end{algorithmic}
\end{algorithm}

\begin{algorithm}[t!]
\caption{Heuristic Backtracking Algorithm}
\label{alg:HBA}
\textbf{function} $\mathtt{heuristic\_backtrack}$($solution$, $DiReG$, $X$, $D$, $S$) \textbf{returns} a solution or infeasibility
\begin{algorithmic}[1] 
\STATE $X_i$.inFlow = 0 
\STATE \textbf{if} $|solution| \leq k$ \textbf{AND} each $X_i$.inFlow $\geq$ $S_i$ \textbf{return} $solution$
\STATE $local\_X$ = 
$\mathtt{select\_unsatisfied\_variable}$($DiReG$, $X$, $D$, $S$)
\FOR{\textbf{each} $local\_cand$ in $\mathtt{sort\_candidates}$($local\_X$, $solution$, $DiReG$, $X$, $D$, $S$)}
\IF{$local\_cand$ is consistent with $solution$}
\STATE $tuple$ = $\{local\_X$ = $local\_X$.append($local\_cand$)\}
\STATE $solution$ = $solution$ $\cup$ $tuple$
\STATE $local\_X$.inFlow = $local\_X$.inFlow + 1
\STATE $result$ = $\mathtt{heuristic\_backtrack}$($solution$, $DiReG$, $X$, $D$, $S$)
\STATE \textbf{if} {$result$ $\neq$ \emph{infeasibility}} 
\textbf{return} $result$
\ENDIF
\STATE $tuple$ = \{$local\_X$ = $local\_X$.remove($local\_cand$)\}
\STATE $solution$ = $solution$ $\setminus$ $tuple$
\STATE $local\_X$.inFlow = $local\_X$.inFlow - 1
\ENDFOR
\STATE \textbf{return} \emph{infeasibility}
\end{algorithmic}

\textbf{function} $\mathtt{select\_unsatisfied\_variable}$($DiReG$, $X$, $D$, $S$) 
\begin{algorithmic}[1] 
\STATE \textbf{return} $X_i$ with the lowest $\sfrac{|D_i|}{\max((S_i - X_i\text{.inFlow}), 1)}$ ratio
\end{algorithmic}

\textbf{function} $\mathtt{sort\_candidates}$($local\_X$, $solution$, $DiReG$, $X$, $D$, $S$) 
\begin{algorithmic}[1] 
\STATE \textbf{return} candidates $D_i$ sorted in decreasing order of their out degree in $DiReG$
\end{algorithmic}

\end{algorithm}

\subsubsection{Preprocessing} Find the strongly connected components $SCC$ of a graph in time linear in the size of $m$ and $|\mathcal{G}|+|\mathcal{P}|$, equivalent to $m$ and $n$ in real-world settings. The next step is find inter- and intra-component pairwise feasibility. We note that we only do a pairwise feasibility test as previous work has shown that doing a three-way, a four-way or greater feasibility tests increase the computational time significantly without improving the scope of finding a group of variables whose combination guarantees an infeasible instance \cite{russell2002artificial}.

\paragraph{{Inter-component pairwise feasibility:}} 

Select two variables $X_i$, $X_j$ corresponding to constraints $S_i$, $S_j$ on level C of DiReGraph, one each from different components of $SCC$. Do a pairwise feasibility check for each pair and return infeasibility if any one pair of variables can not return a valid committee. The correctness and completeness of this step is easy. If there are more constraints than the available candidates, it is impossible to find a feasible solution. Also, if a pair of constraints are pairwise infeasible, then it is clear that they will remain infeasible overall.

\vspace{0.1cm}
\paragraph{{Intra-component pairwise feasibility:}} Repeat the above procedure but now, within a component. This step also helps in returning infeasibililty efficiently.

\vspace{0.1cm}
\paragraph{{Reducing domain:}} 

Based on empirical evidence of the previous work that used a setting similar to ours, pairwise infeasibility causes a majority of overall infeasible instances \cite{russell2002artificial}.  Hence, if a committee did exist, the domain of each variable is reduced by removing candidates that explicitly do not help to find feasible committees.

Now do a restricted version of intra-component pairwise feasibility. If algorithm reaches this stage, we know that all of the constraints are pairwise feasible due to presence of at least one solution. Hence, reduce the domain by removing a candidate who, when included in the solution, always returns pairwise infeasible solution with another constraint. Specifically, fix a candidate $c$ from the domain of a variable $X_i$ and do a pairwise feasibility check with other domain $X_j$ across all possible solutions that contains the candidate $c$. If all solutions that contain $c$ result in infeasibility, then remove candidate $c$ from the domain of $X_i$.

\subsubsection{Heuristic Backtracking.}

Use depth-first search for backtracking. Specifically, choose one variable $X_i$ at a time, and backtrack when $X_i$ has no legal values left to satisfy the constraint. This technique repeatedly chooses an unassigned variable, and then tries all values in its domain, trying to find a solution. If an infeasibility is returned, traverse back by one step and move forward by trying another value.

 \paragraph{Select unsatisfied variable:} 
Use the ``minimum-remaining-values (MRV)'' heuristic to choose the variable having the fewest legal values. This heuristic picks a variable that is most likely to cause a failure
soon, thereby pruning the search tree. For example, if some variable $X_i$ has no legal values left, the MRV heuristic will select $X_i$ and infeasibility will be returned, in turn, avoiding additional searches.

\paragraph{{Sort most favorite candidates:}} 

Use the ``most-favorite-candidates (MFC)'' heuristic to sort the candidates in domain $D_i$ such that a candidate on level B of DiReGraph who is most connected to level C (out-degree) is ranked the highest. This heuristic tries to reduce the branching factor on future choices by selecting the candidate that is involved in the largest number of constraints. 

Overall, the aim is to select the most favorite candidates into the committee as they help satisfy the highest proportion of constraints. For completeness and to get multiple DiRe committees, after sorting step, use a ``shift-left'' approach where the second candidate becomes the first, the first becomes the last, and so on. This allows us to get multiple DiRe committees.

\begin{example}
\textbf{Sorting candidates:}
In Figure~\ref{fig:DiReGraph}, the ordering of candidates will be $c_2$, $c_1$, $c_4$, and $c_3$ as $c_2$ has out-degree of 3, $c_1$ and $C_4$ has 2, and $c_3$ has 1. Ties are broken randomly. 
\end{example}

We now give an example to explain the entire algorithm.

\begin{example}\textbf{Implementation of Algorithm:}
Consider the election setup shown in Figure~\ref{fig:example}, which consists of four candidates and four voters, each having one attribute. 

The input to the algorithm is (i) a set of variables (candidate group names and voter population names) = \{Male, Female, CA, IL\}. (ii) a collection of sets of domain for each variable (candidates part of the candidate groups and winning committee of each population) = (\{$c_1, c_2$\}, \{$c_3, c_4$\}, \{$c_1, c_2$\}, \{$c_2, c_4$\}). (iii) a collection of constraints for each variable = (1, 1, 1, 1). 

The first step of the algorithm is to create a DiReGraph as shown in Figure~\ref{fig:DiReGraph}. Level A is set to 2 as $k=2$. Level B consists of four nodes, each representing one candidate. Level C consists of four nodes, which is equal to $|\mathcal{G}|+|\mathcal{P}|$, the number of candidate groups and voter populations in the election. Level D consists of the final output. Each node on Level B is connected with Level A and each node on Level C is connected to Level D. The candidate $c_1$ is a male who is in winning committees of California. Hence, $c_1$ in DiReGraph is connected with the two sets of constraints, one each for male and CA (California).

Next, a subgraph $SG$ consisting of eight nodes from Levels B and C and the corresponding edges that connect these eight nodes is created.

As there is only one strongly connected components in the SG, we directly check for the intra-component pairwise feasibility. For each pair of domains, there is always a feasible committee that exists. Hence, the algorithm continues to execute. Moreover, none of the domains get reduced as the constraints are set to one. 

Before reaching the final step of the algorithm, the algorithm would have terminated if no feasible committee existed that satisfied all the pairwise constraints. 

In the final step, the $\mathtt{select\_unsatisfied\_variable}$ function selects a candidate at random as all the variables have the same ratio of 2 for $\sfrac{|D_i|}{(S_i - X_i\text{.inFlow})}$ as $|D_i|$=2 and $S_i$=1 for all $D_i \in D$ and $S_i \in S$. Next, for each variable that remains, we check whether adding that variable violates the global constraint (committee size on Level A of DiReGraph) or not. We keep on backtracking till we either find a committee or exhaustively navigate through our pruned search space. To get more than one committee, rerun the  $\mathtt{heuristic\_backtrack}$ function by applying an additional ``left-shift'' operation on the result of the $\mathtt{sort\_candidates}$ function each time the $\mathtt{heuristic\_backtrack}$ function is implemented. We note that this increases the time complexity of the algorithm linearly in the size of $\mathcal{G}$ and $\mathcal{P}$.
\end{example}

\section{Empirical Analysis}
\label{sec:empres}

We now empirically assess the efficiency of our heuristic-based algorithm using real and synthetic datasets. We also assess the effect of enforcing diversity and representation constraints on the feasibility and utility of the winning committee selected using different scoring rules.

\subsection{Datasets}
\label{sec:empres/data}

\subsubsection{Real Datasets}

\paragraph{RealData 1:} The \emph{\underline{Eurovision}} dataset \cite{Kaggle2019Eurovision} consists of 26 countries ranking the songs performed by each of the 10 finalist countries. 
We aim to select a 5-sized DiRe committee. Each candidate, a song performed by a country, 
has two attributes, the European region and the language of the song performed. Each voter has one attribute, the voter's European region. Specifically for the European region attribute, Australia and Israel were labeled as ``Others'' as they are not a part of Europe.

\paragraph{RealData 2:} The \emph{\underline{United Nations Resolutions}} dataset \cite{Voeten2014UN} consists of 193 UN member countries 
voting for 81 resolutions presented in the UN General Assembly in 2014. We aim to select a 12-sized DiRe committee. Each candidate has two attributes, the topic of the resolution and whether a resolution was a significant vote or not. Each voter has one attribute, the continent.

\subsubsection{Synthetic Datasets} 

\paragraph{SynData 1:} We set committee size ($k$) to 6 for 100 voters and 50 candidates. We generate complete preferences using RSM by setting selection probability $\Pi_{i, j}$ to replicate Mallows' \cite{Mallows1957} model ($\phi=0.5$, randomly chosen reference ranking $\sigma$ of size $m$) (Theorem 3, \cite{chakraborty2020algorithmic}) and preference probability $p(i)=1$, $\forall i \in [m]$.


\underline{Dividing Candidates into Groups and Voters into Populations:}
To assess the impact of enforcing constraints, 
we generate datasets with varying number of candidate and voter attributes by iteratively choose a combination of $(\mu$, $\pi)$ such that $\mu$ and $\pi \in \{0,1,2,3,4\}$. For each candidate attribute, we choose a number of non-empty partitions $q \in [2$, $k]$, uniformly at random. Then to partition $C$, we randomly sort the candidates $C$ and select $q-1$ positions from $[2$, $m]$, uniformly at random without replacement, with each position corresponding to the start of a new partition. The partition a candidate is in is the attribute group it belongs to. For each voter attribute, we repeat the above procedure, replacing $C$ with $V$, and choosing $q-1$ positions from the set $[2$, $n]$. For each combination of $(\mu,\pi)$, we generate five datasets.
We limit the number of candidate groups and number of voter populations per attribute to $k$ to simulate a real-world division of candidates and voters. 

\paragraph{SynData 2:} We use the same setting as SynData 1, except we fix $\mu$ and $\pi$ each to 2 and vary the cohesiveness of voters by setting selection probability $\Pi_{i, j}$ to replicate Mallows' \cite{Mallows1957} model's $\phi$ $\in$ $[0.1$, $1]$, with increments of 0.1.  We divide the candidates into groups and voters into populations in line with \textbf{SynData 1}.

\subsection{Setup}
\label{sec:empres/setup}

\paragraph{System.}
We note that all our experiments were run on a personal machine without requiring the use of any commercial or paid tools. 
More specifically, we used a controlled virtual environment using Docker(R) on a 2.2 GHz 6-Core Intel(R) Core i7 Macbook Pro(R) @ 2.2GHz with 16 GB of RAM running MacOS Big Sur (v11.1). 
We used Python 3.7.

\paragraph{Constraints.}
For each $G \in \mathcal{G}$, 
we choose $l_G^D$ $\in$ $[1,\min(k$, $|G|)]$ uniformly at random. For each $P \in \mathcal{P}$, 
we choose $l_P^R$ $\in$ $[1$, $k]$ uniformly at random.

\paragraph{Voting Rules.} We use previously defined $k$-Borda, $\beta$-CC, and Monroe rules. More specifically, we have two rules from from submodular, monotone class of \csr due to inherent difference in their method of computing committees and one rule from separable, monotone class of functions. We deem these to be sufficient due to our focus on the study of \DiReCF  as discussed later. 

\begin{figure}[t!]
\centering
\begin{subfigure}{.3\textwidth}
  \centering
  \includegraphics[trim=80 15 0 35,clip,width=0.975\linewidth]{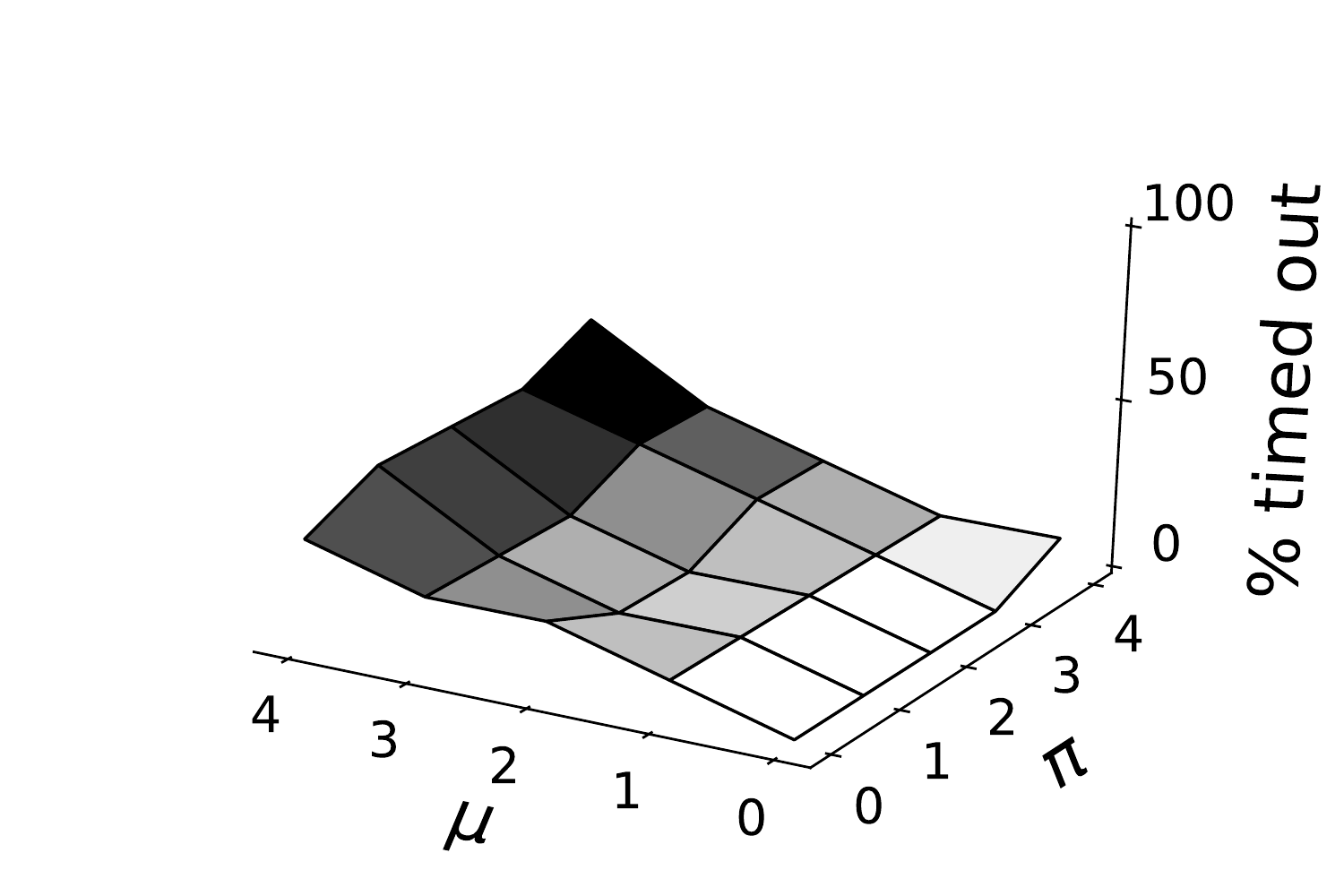}
  \caption{\textbf{timed out instances}}
  \label{fig:runningtimes/timedout}
\end{subfigure}%
\begin{subfigure}{.3\textwidth}
  \centering
  \includegraphics[trim=80 15 0 35,clip,width=0.975\linewidth]{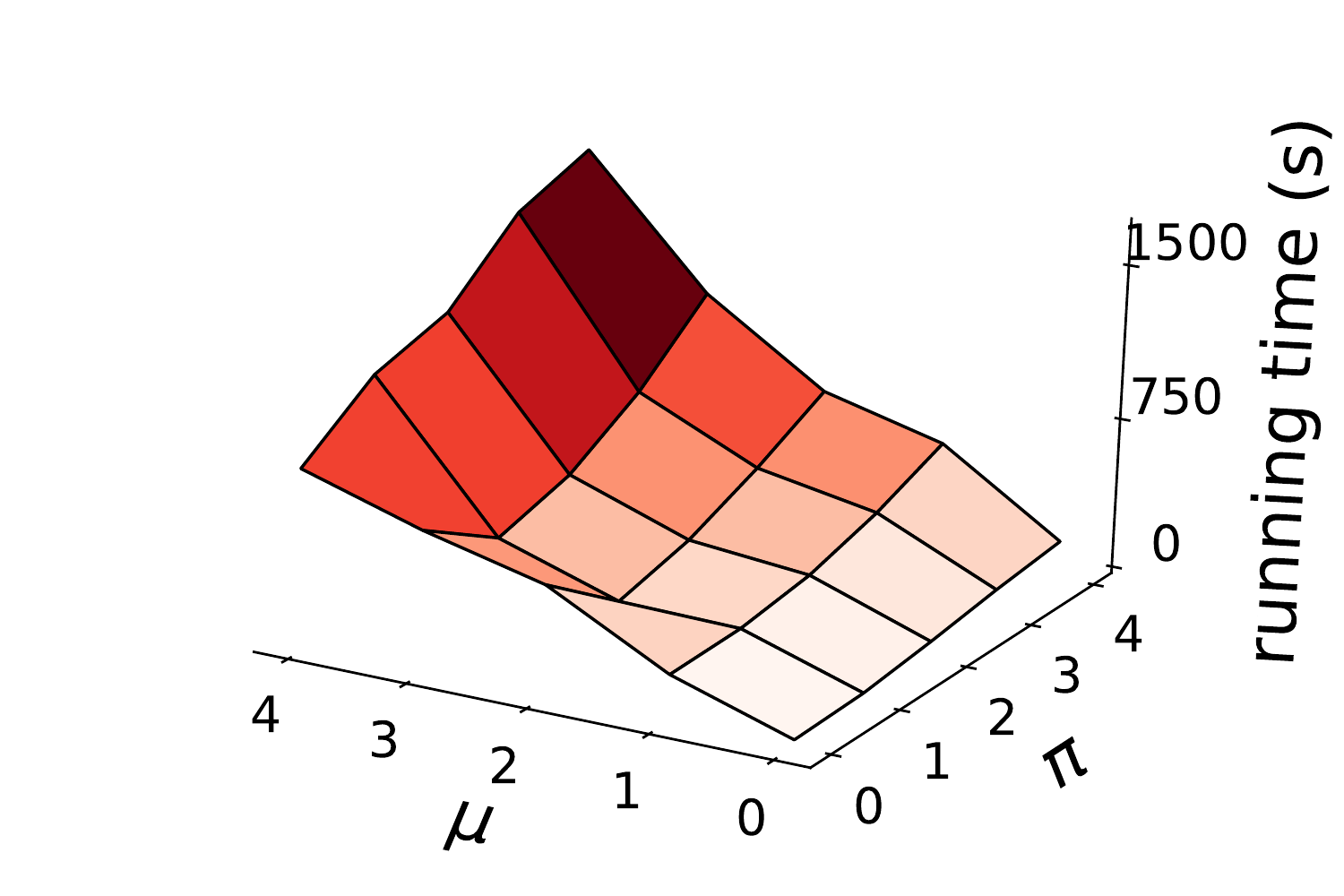}
  \caption{\textbf{mean running time}}
  \label{fig:runningtimes/avgtime}
\end{subfigure}
\caption{Using \textbf{SynData 1}, (a) Proportion (in \%) of instances that timed out at 2000 seconds and (b) mean running time of non-timed out instances. Each combination of $\mu$ and $\pi$ has 10 instances, 5 each for $k$-Borda and $\beta$-CC rule. 
}
\label{fig:runningtimes}
\end{figure}
\begin{figure}[t]
\centering
\begin{subfigure}{.3\textwidth}
  \centering
  \includegraphics[trim=80 15 0 35,clip,width=0.975\linewidth]{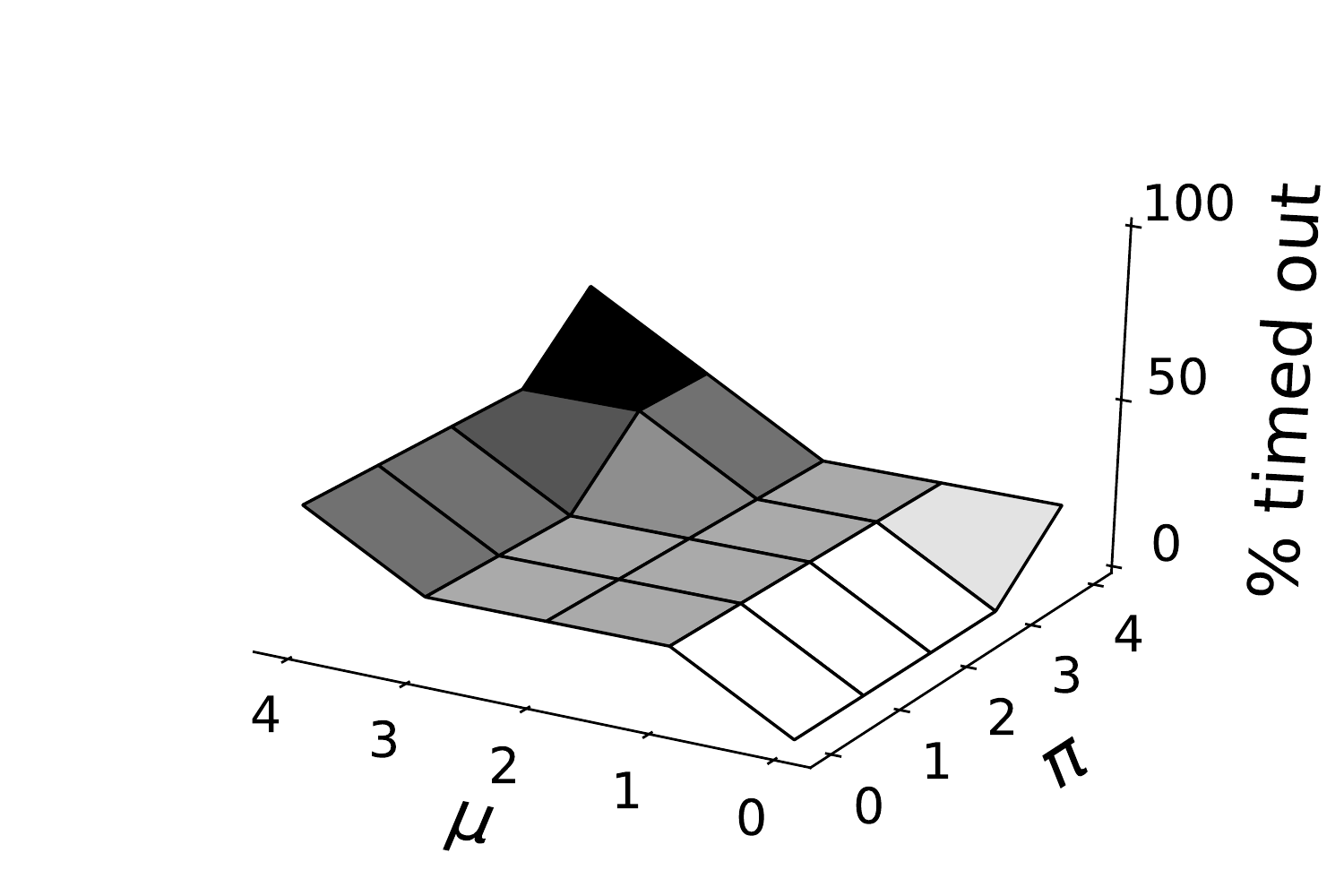}
  \caption{\textbf{timed out instances}}
  \label{fig:runningtimesMon/timedout}
\end{subfigure}%
\begin{subfigure}{.3\textwidth}
  \centering
  \includegraphics[trim=80 15 0 35,clip,width=0.975\linewidth]{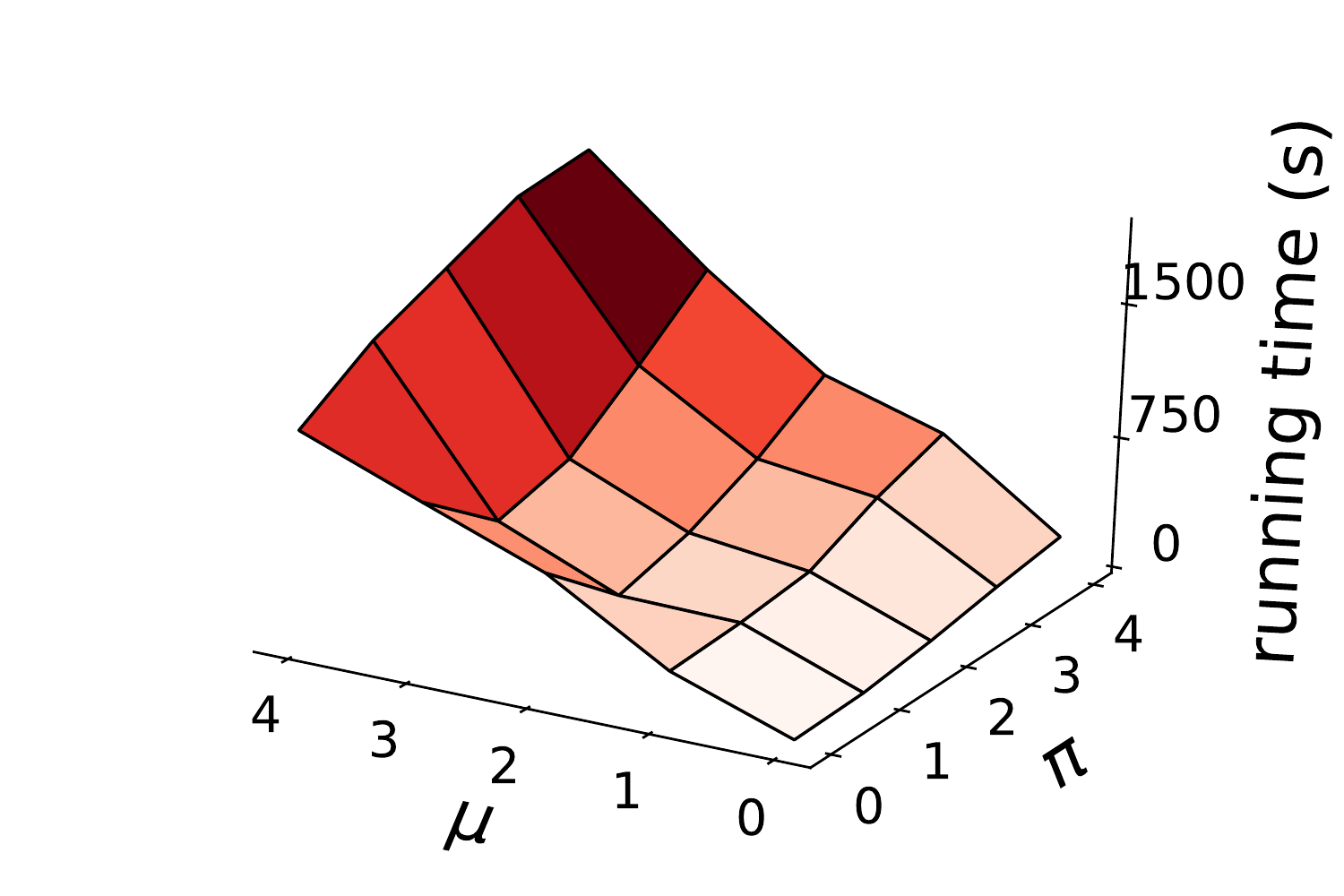}
  \caption{\textbf{mean running time}}
  \label{fig:runningtimesMon/avgtime}
\end{subfigure}
\caption{Using \textbf{SynData 1}, (a) Proportion (in \%) of instances that timed out at 2000 seconds and (b) mean running time of non-timed out instances. Each combination of $\mu$ and $\pi$ has 5 instances for Monroe Rule. 
}
\label{fig:runningtimesMon}
\end{figure}

\subsection{Results}
We now present results of empirical analyses of the efficiency of the heuristic algorithm and of the feasibility of DiRe committees and the cost of fairness.

\subsubsection{Efficiency of Heuristic Algorithm}
\label{subsec:effOfHA}
All experiments in this section combine instances of $k$-Borda and $\beta$-CC as there was no pairwise significant difference in the running time between the sets of instances of these two scoring rules (Student's t-test, $p>$ 0.05). We present the result for Monroe rule separately.

\paragraph{{Algorithm is efficient:}}
Our heuristic-based algorithm is efficient on tested data sets (Figure~\ref{fig:runningtimes}, Figure~\ref{fig:runningtimesMon}, and Figure~\ref{fig:phifigures2}). Only 18.90\% of 525 instances timed out at 2000 sec. Among the instances that did not time out, the average running time was 566.48 sec (standard deviation (sd) = 466.66) for $k$-Borda and $\beta$-CC, and 724.39 sec (sd = 575.31) for Monroe. In contrast, 93.71\% of 525 instances timed out at 2000 sec when using brute-force algorithm.

Promisingly, using DiReGraph made the algorithm more efficient on instances that were sparsely connected as the average running time for all $\mu$ when $\pi$ $\leq$ $2$ was 281.47 sec (sd = 208.65) for $k$-Borda and $\beta$-CC, and 358.87 sec (sd = 265.82) for Monroe. Higher $\pi$ led to denser DiReGraphs.

\paragraph{Performance when compared to ILP:}
The real-world application of an ILP-based algorithm is very
limited when using Chamberlin-Courant and Monroe rules
\cite{skowron2015achieving}. More specifically, some instances of the ILP-algorithm that implemented the Monroe rule for $k$=9, $m$=30, and $n$=100 timed out after one hour. The running time increased exponentially with increase in the number of voters as \emph{all} instances of the ILP-algorithm that implemented the Monroe rule for $k$=9, $m$=30, and $n$=200 did not terminate even after one day \cite{skowron2015achieving}. Hence, our algorithm, which (i) handles constraints and any committee selection rule and (ii) terminated in (avg) 724 sec, has a clear edge. Promisingly, the first committee returned by the algorithm (in $<$ 120 sec) was the winning DiRe committee among 63\% of all instances. Moreover, our algorithm scales linearly with an increase in the number of voters.

\paragraph{{Efficiency and cohesiveness:}} Our algorithm was the most efficient when the voters were either less cohesive ($\phi \leq$  0.3) or more cohesive ($\phi\geq$  0.8) (Figure~\ref{fig:phifigures2}). 
Among these two efficient sets of instances, the time taken (mean = 105.40 sec (sd = 4.16) for $k$-Borda and $\beta$-CC, and mean = 141.06 sec (sd = 8.08) for Monroe) by the preprocessing stage to return infeasibility for low $\phi$ was less and the time taken (mean = 156.80 sec (sd = 2.86) for $k$-Borda and $\beta$-CC, and mean = 203.98 sec (sd = 12.66) for Monroe) by the heuristic-based search stage to return a DiRe committee for higher $\phi$ was less. This shows the efficiency of our algorithm in opposing scenarios: the preprocessing step was efficient when $\phi$ was low as it was easy to find a pair of constraints that are pairwise infeasible, and the heuristic-based backtracking was efficient when $\phi$ was high as it was easy to find a DiRe committee.

\begin{figure}
\centering
\begin{subfigure}{.33\textwidth}
  \centering
  \includegraphics[trim=0 15 0 12,clip,width=0.995\linewidth]{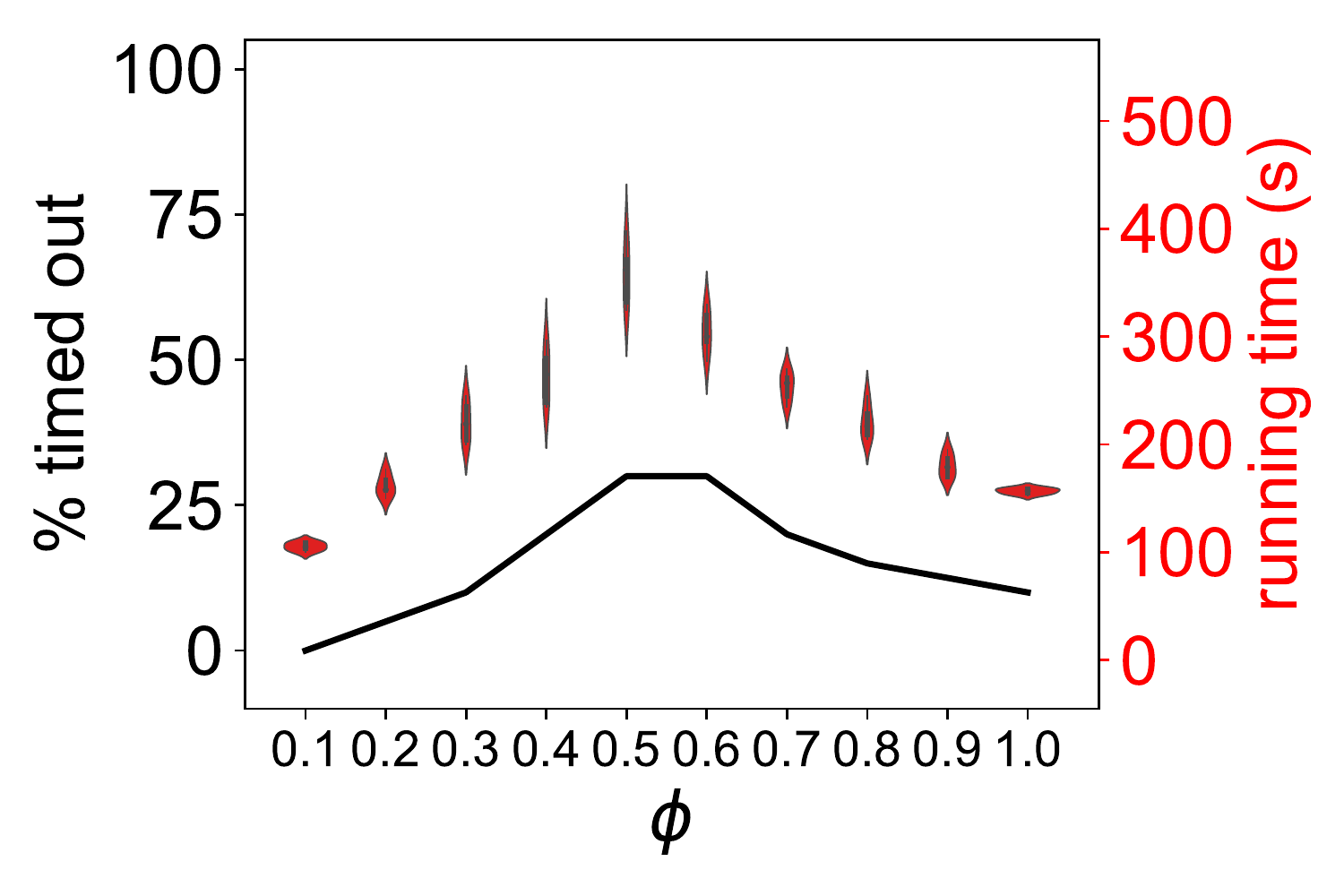}
  \caption{\textbf{$k$-Borda and $\beta$-CC}}
  \label{fig:phifigures/timebordaBcc}
\end{subfigure}
\begin{subfigure}{.33\textwidth}
  \centering
  \includegraphics[trim=0 15 0 12,clip,width=0.995\linewidth]{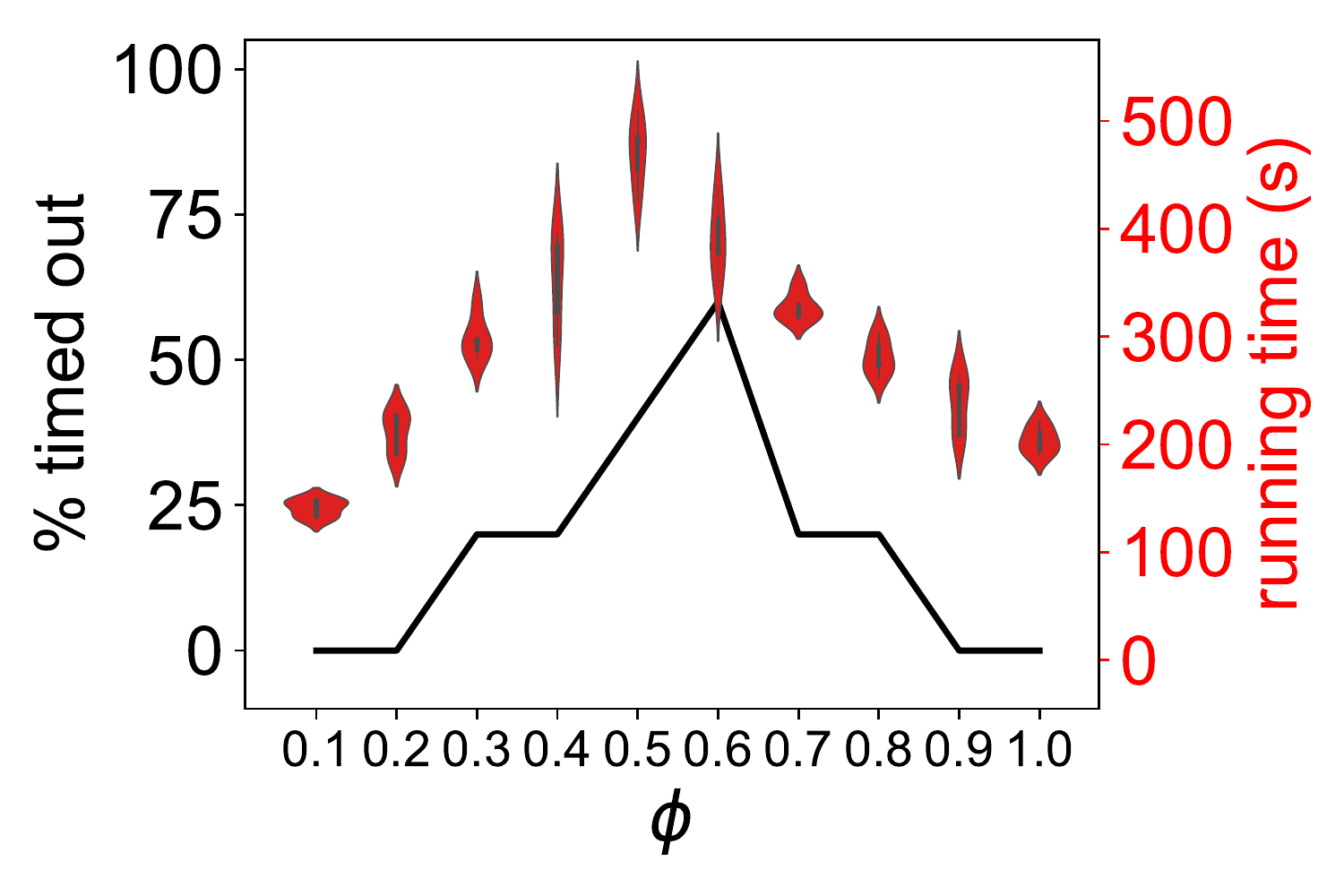}
  \caption{\textbf{Monroe}}
  \label{fig:phifigures/timemonroe}
\end{subfigure}
\caption{Using \textbf{SynData 2}, proportion (in \%) of instances that timed out at 2000 seconds and mean running time of non-timed out instances. (a) Each $\phi$ has 10 instances (5 for $k$-Borda and 5 for $\beta$-CC). (b) Each $\phi$ has 5 instances for Monroe rule.}
\label{fig:phifigures2}
\end{figure}

\begin{figure}[t]
\centering
\begin{subfigure}{.3\textwidth}
  \centering
  \includegraphics[trim=0 15 0 12,clip,width=0.995\linewidth]{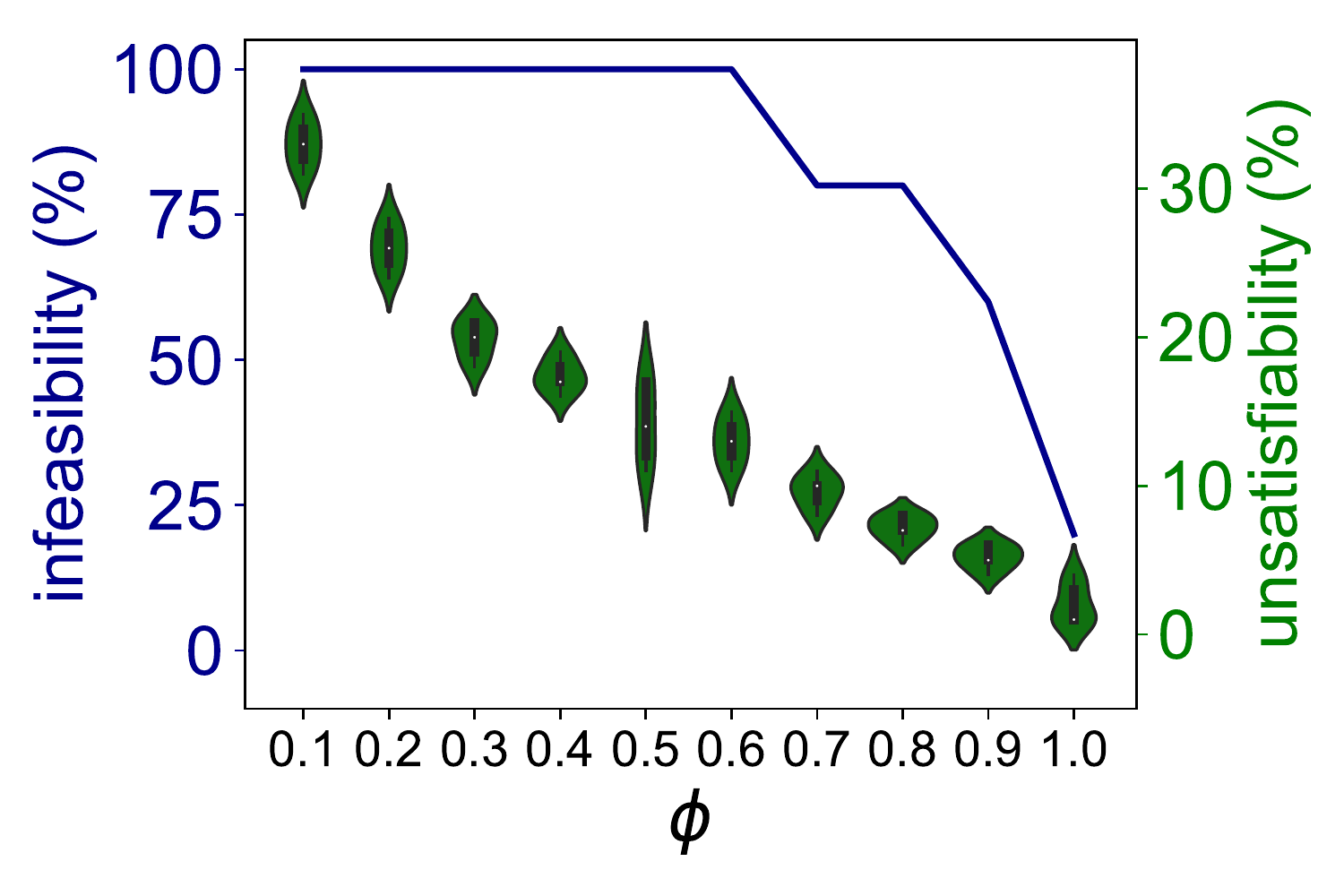}
  \caption{\textbf{$k$-Borda}}
  \label{fig:phifigures/kborda}
\end{subfigure}%
\begin{subfigure}{.3\textwidth}
  \centering
  \includegraphics[trim=0 15 0 12,clip,width=0.995\linewidth]{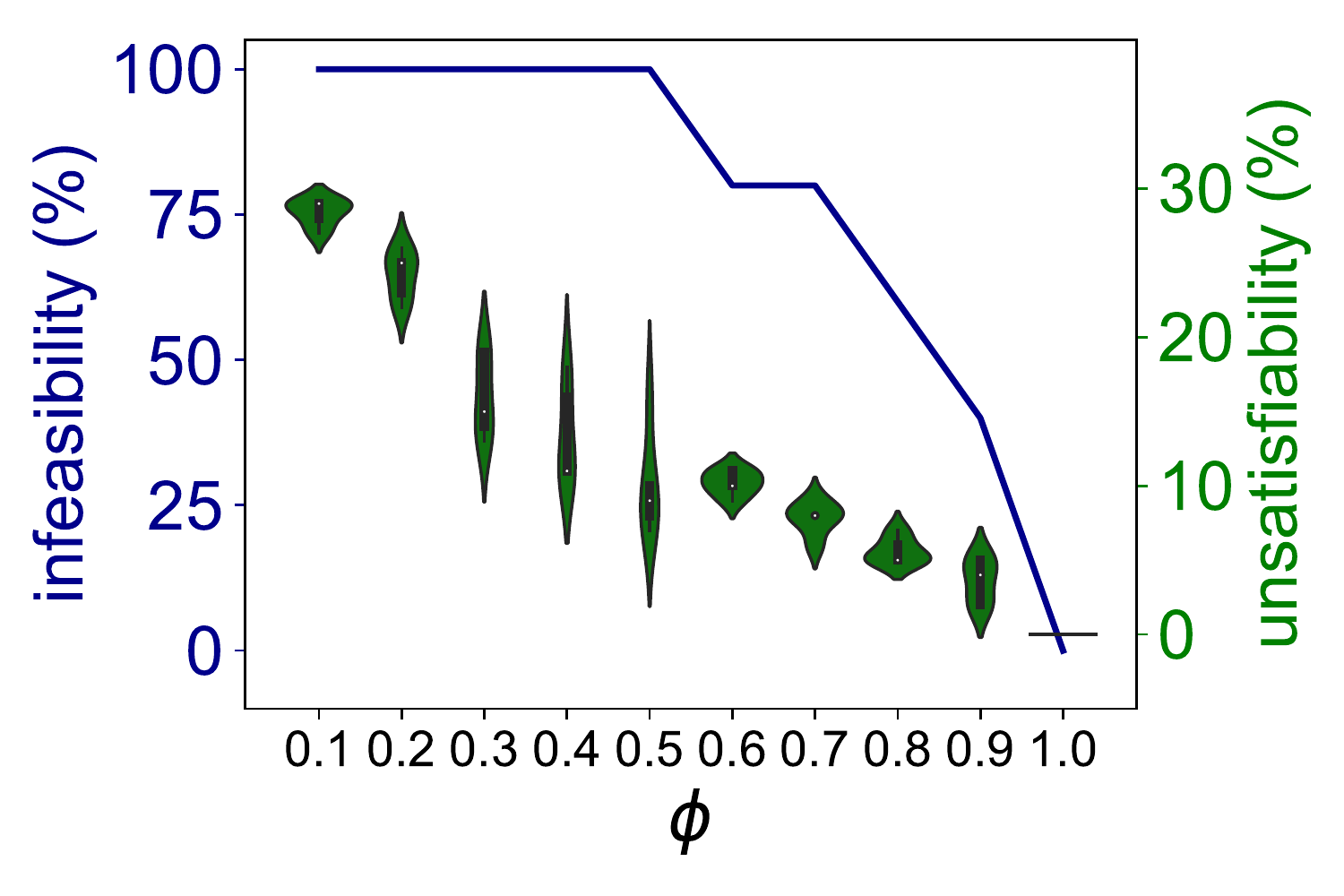}
  \caption{\textbf{$\beta$-CC}}
  \label{fig:phifigures/bcc}
\end{subfigure}
\begin{subfigure}{.3\textwidth}
  \centering
  \includegraphics[trim=0 15 0 12,clip,width=0.995\linewidth]{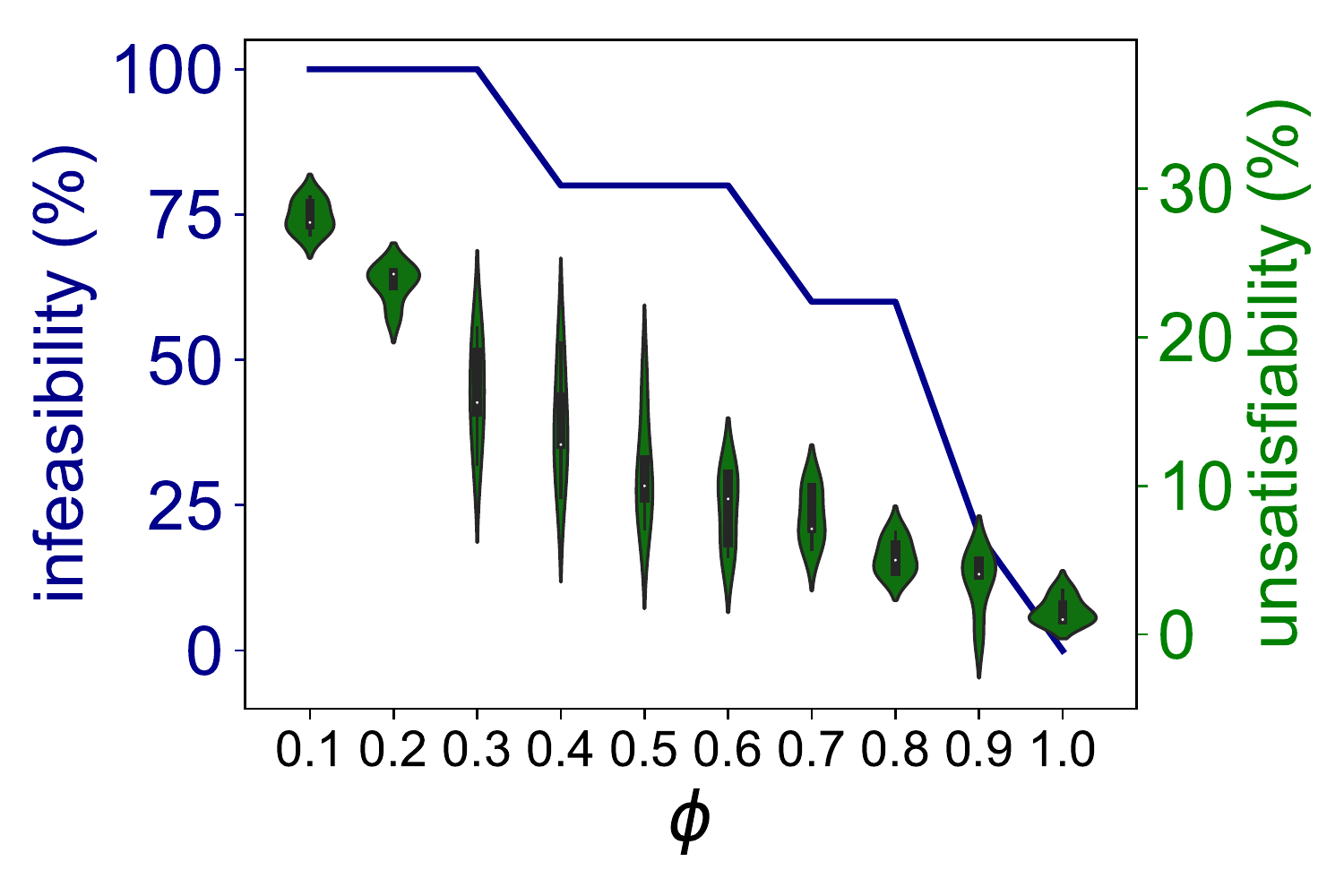}
  \caption{\textbf{Monroe}}
  \label{fig:phifigures/monroe}
\end{subfigure}
\caption{Using \textbf{SynData 2}, proportion (in \%) of instances that have an infeasible committee and the maximum proportion (in \%) of constraints that are unsatisfiable per instance; each $\phi$ has 5 instances.
}
\label{fig:phifigures}
\end{figure}

\subsubsection{Feasibility and Cost of Fairness}
All experiments in this section consider instances of $k$-Borda and $\beta$-CC separately as there was a difference in the unsatisfiability between the two (Student's t-test, $p<$ 0.05). We continue to an analyze Monroe separately.

\paragraph{{Higher number of attributes result in infeasible committee:}}
Figure~\ref{fig:infeasibility} shows the proportion of feasible instances for each combination of $\mu$ and $\pi$. 
As the number of attributes increases, the proportion of feasible instances decreases. However, Figure~\ref{fig:unsatisfiable} shows that 
the mean proportion of constraints satisfied for each instance is $\geq 90\%$ (sd $\in$ [0, 5]). Hence, from the computational perspective, these results show the real-world utility of breaking down \DiReCWD problem into two-steps: (i) \DiReCF problem solved using our algorithm followed by (ii) utility maximization problem. 
As we expect a constant number of committees to be feasible in real-world, we can overcome the intractability of using submodular scoring function, notwithstanding the worst case when all committees are feasible.

On the other hand, more promisingly: (i) When the sum of all the constraints was less than ($\mu \cdot k$), then, indeed, a feasible committee did exist on 85\% of instances. (ii) More specifically,  when the sum of the constraints was less than $k$ for all groups under each candidate attribute individually, then, indeed, a feasible committee did exist on all but one instance.



\paragraph{{Infeasibility and unsatisfiability is dependent on cohesiveness:} }
There was a negative correlation between the maximum proportion of unsatisfied constraints and $\phi$, for all the three scoring rules (mean Pearson's $\rho$ = -0.95, $p<$0.05)
It was to easier to  satisfy the constraints when the cohesiveness ($\phi$) was high, which led to lower infeasibility for higher $\phi$ (Figure~\ref{fig:phifigures}). 

Note that the correlation is stated keeping the candidate groups and voter populations constant. Only the preferences vary and hence, do the winning committee $W_P$ for each population $P \in \mathcal{P}$.  This is to say that higher cohesiveness of voters leads to higher cohesiveness among all $W_P$s and in turn easier to satisfy the constraints and in turn higher proportion of feasible committees.

\paragraph{$\beta$-CC and Monroe satisfies higher proportion of representation constraints:} $\beta$-CC and Monroe
rules are better at satisfying representation constraints as compared to $k$-Borda (Figure~\ref{fig:infeasibility}) as they are designed to maximize the voter representation, and in turn, the population satisfaction. However, we note that even when we use a \csr that guarantees proportional representation, our analysis found that it was indeed the smaller population whose the representation constraints were violated disproportionately more than that of the larger population. Hence, the price of diversity was paid more by smaller population as compared to larger population, which quantitatively reaffirms the need for DiRe committees.

\begin{figure*}[t!]
\centering
\begin{subfigure}{.3\textwidth}
  \centering
  \includegraphics[trim=80 15 0 25,clip,width=0.975\linewidth]{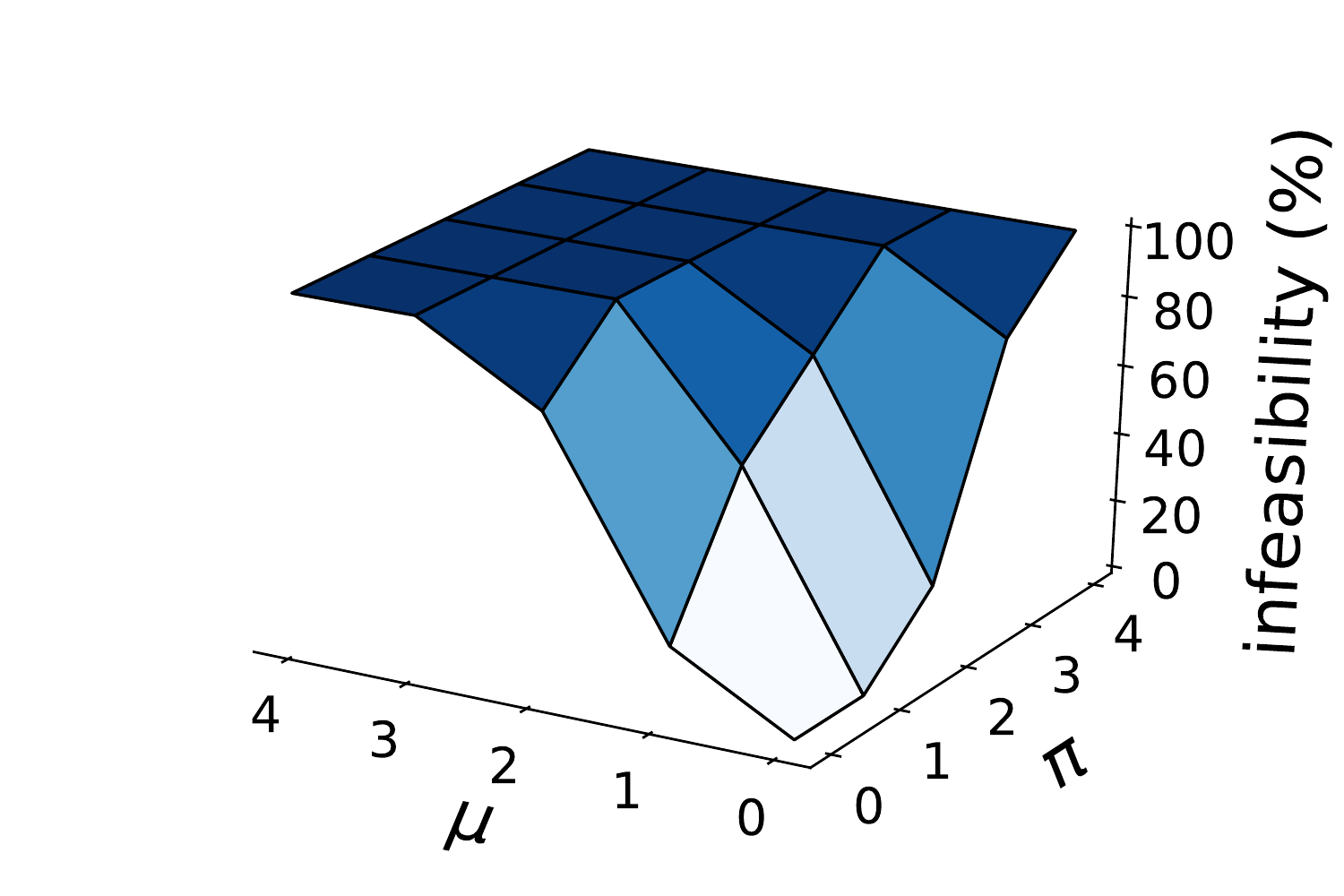}
  \caption{\textbf{$k$-Borda}}
  \label{fig:infeasibility/instances/kborda}
\end{subfigure}%
\begin{subfigure}{.3\textwidth}
  \centering
  \includegraphics[trim=80 15 0 25,clip,width=0.975\linewidth]{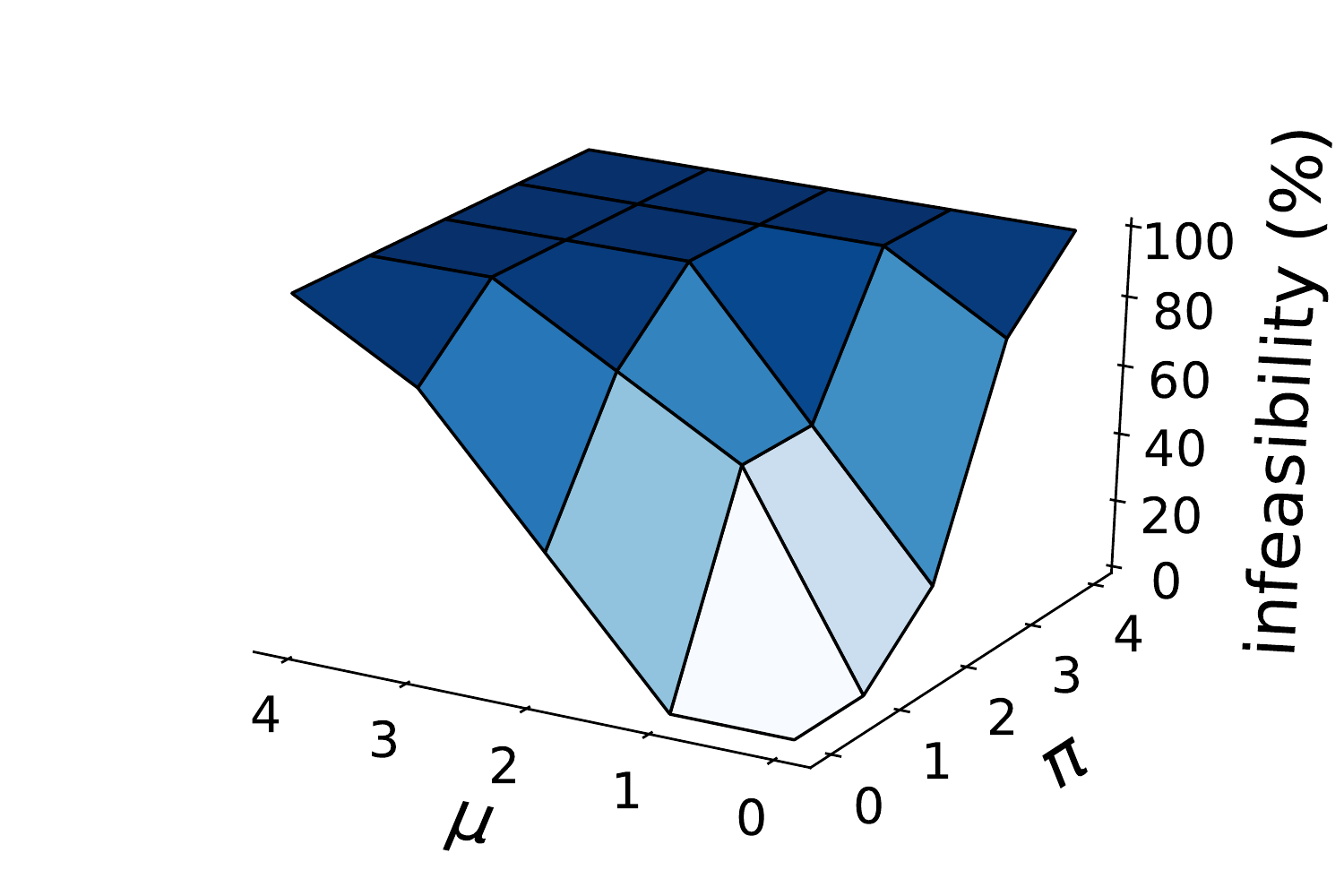}
  \caption{\textbf{$\beta$-CC}}
  \label{fig:infeasibility/instances/bcc}
\end{subfigure}
\begin{subfigure}{.3\textwidth}
  \centering
  \includegraphics[trim=80 15 0 25,clip,width=0.975\linewidth]{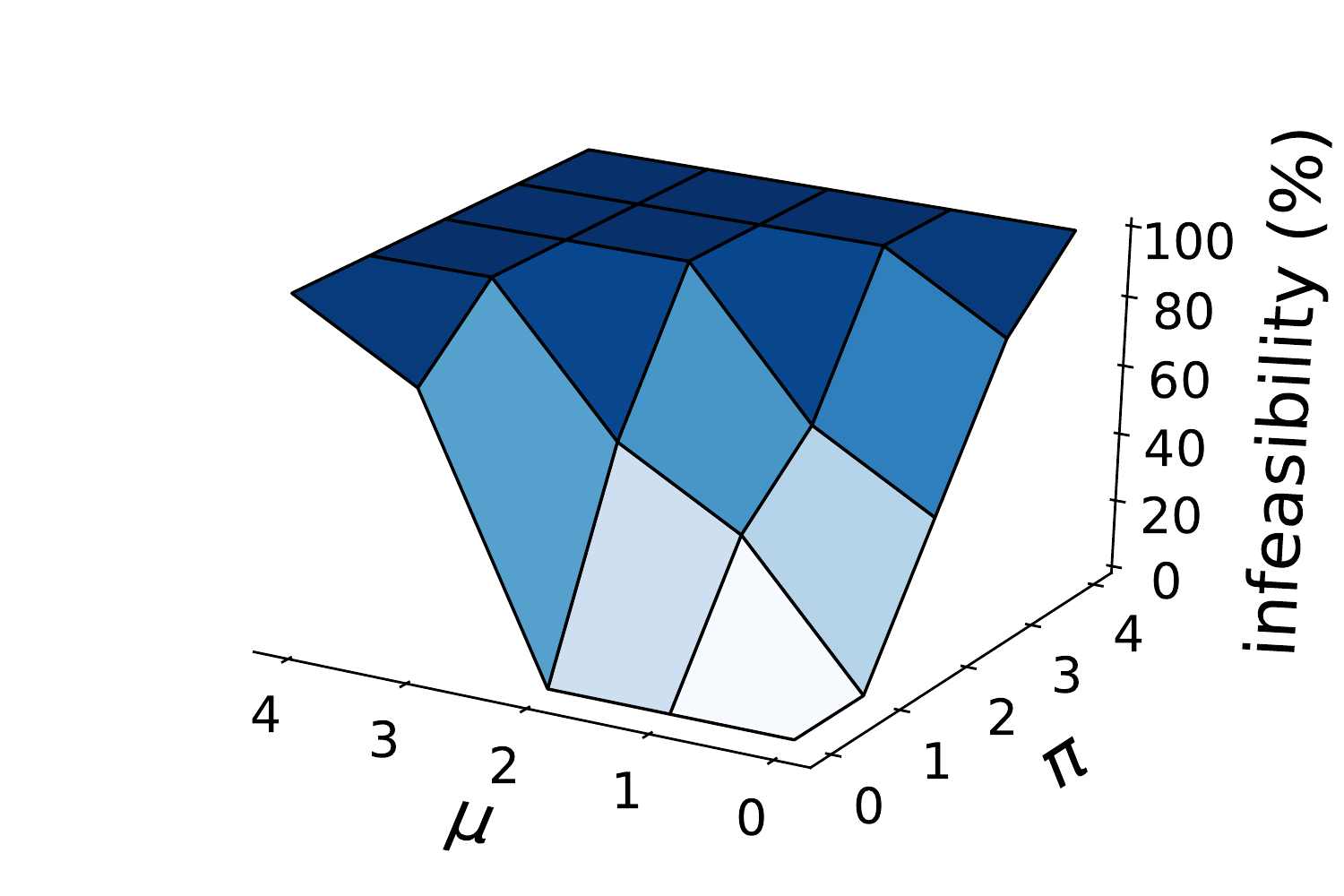}
  \caption{\textbf{Monroe}}
  \label{fig:infeasibility/instances/monroe}
\end{subfigure}
\caption{Using \textbf{SynData 1}, Proportion (in \%) of instances that have an infeasible committee. Each combination of $\mu$ (\# candidate attributes) and $\pi$ (\# voter attributes) has 5 instances.}
\label{fig:infeasibility}
\end{figure*}

\begin{figure*}[t]
\centering
\begin{subfigure}{.3\textwidth}
  \centering
  \includegraphics[trim=80 15 0 25,clip,width=0.975\linewidth]{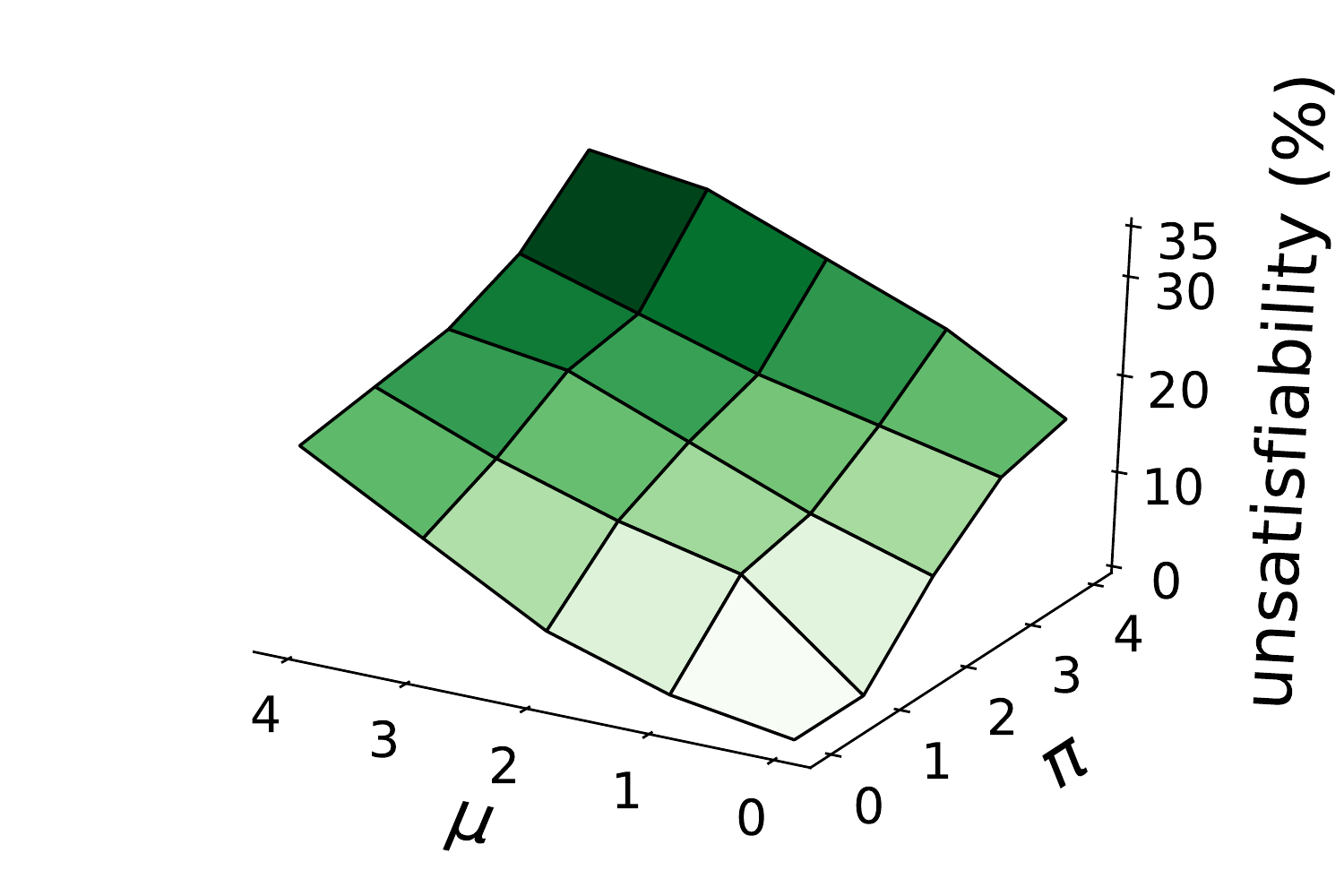}
  \caption{\textbf{$k$-Borda}}
  \label{fig:unsatisfiable/constraints/kborda}
\end{subfigure}
\begin{subfigure}{.3\textwidth}
  \centering
  \includegraphics[trim=80 15 0 25,clip,width=0.975\linewidth]{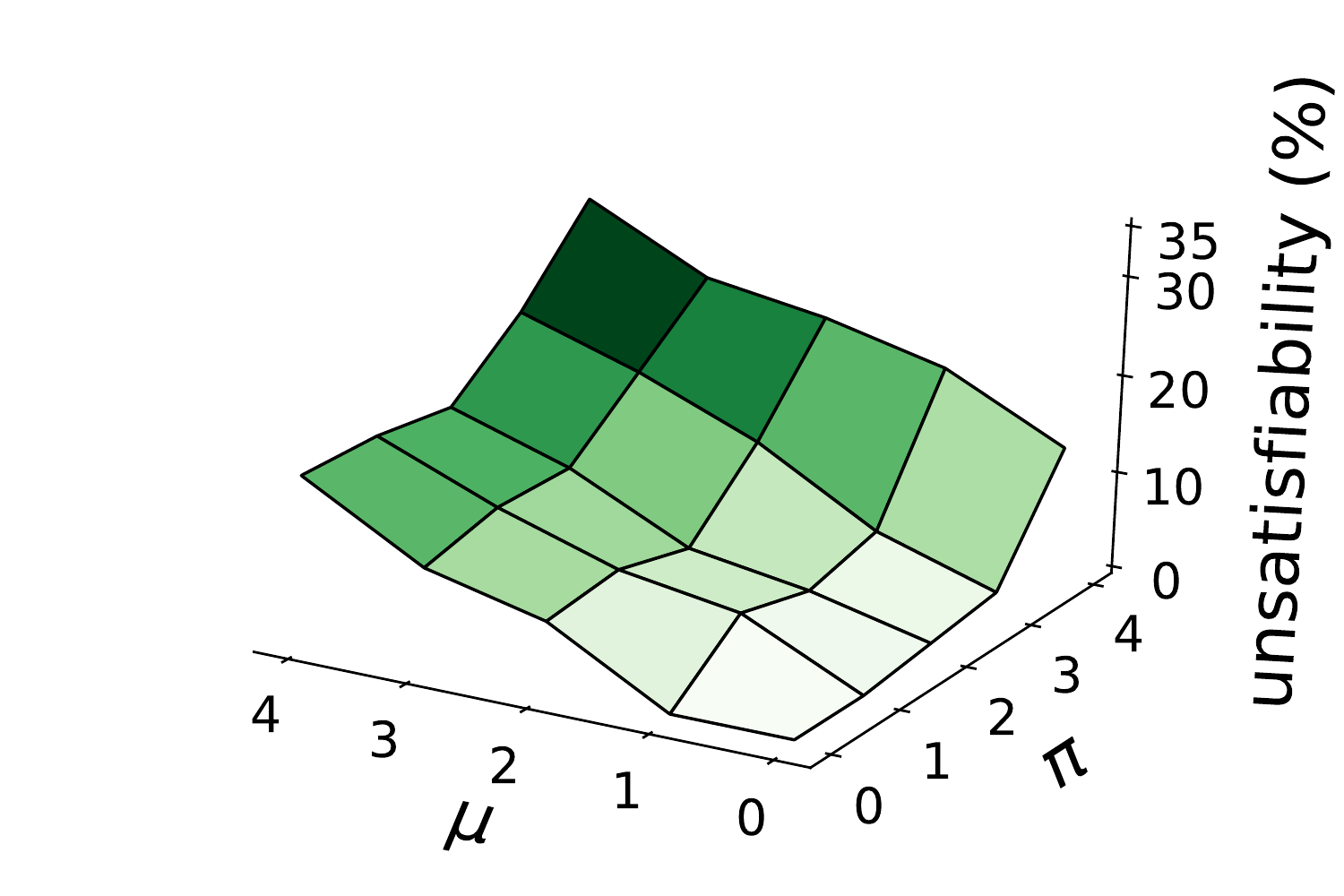}
  \caption{\textbf{$\beta$-CC}}
  \label{fig:unsatisfiable/constraints/bcc}
\end{subfigure}
\begin{subfigure}{.3\textwidth}
  \centering
  \includegraphics[trim=80 15 0 25,clip,width=0.975\linewidth]{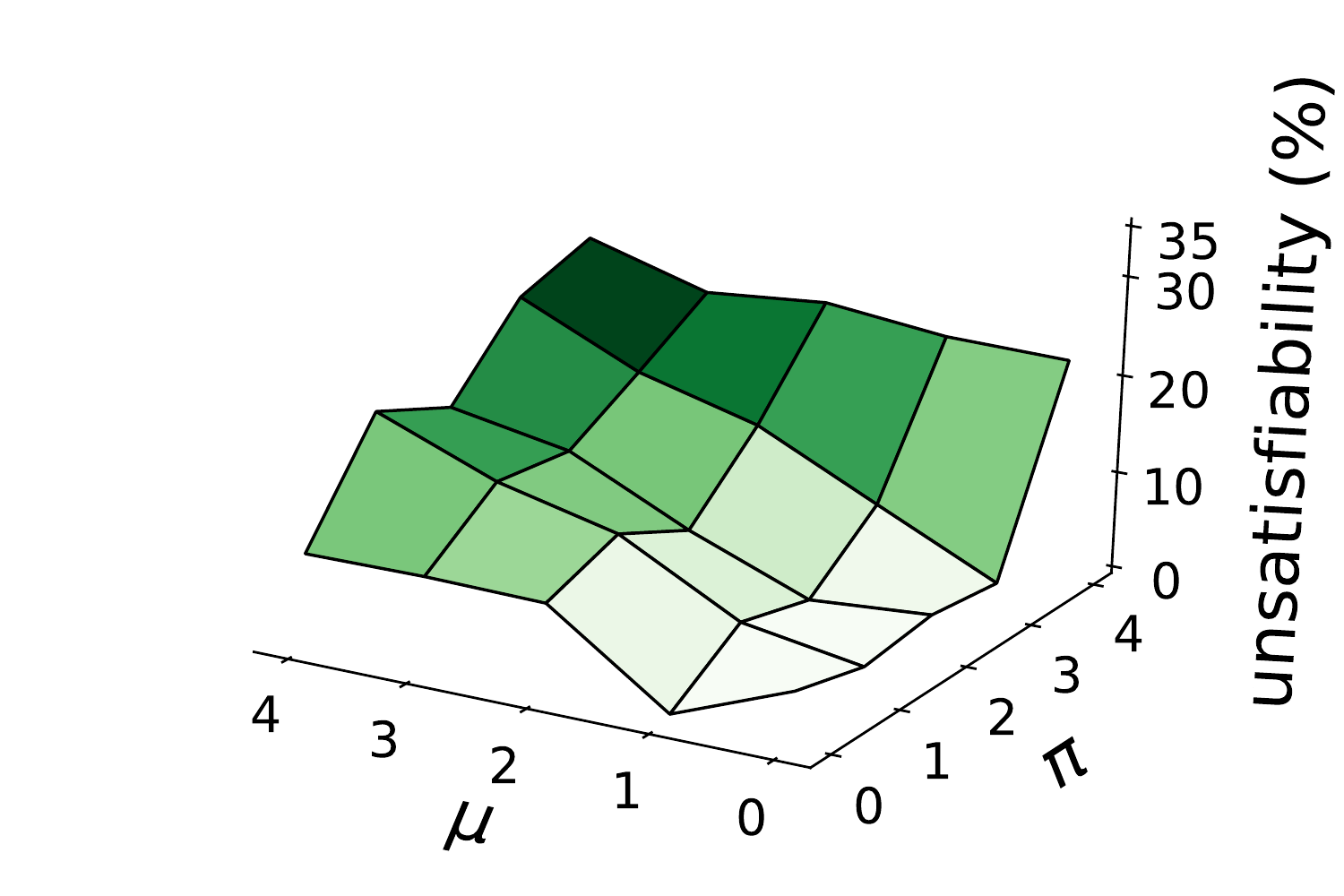}
  \caption{\textbf{Monroe}}
  \label{fig:unsatisfiable/constraints/monroe}
\end{subfigure}
\caption{Using \textbf{SynData 1}, the mean of maximum proportion (in \%) of constraints that are unsatisfiable per instance; the maximum proportion being 0\% if all constraints are satisfied, 100\% if no constraint is satisfied and so on. Each combination of $\mu$ (\# candidate attributes) and $\pi$ (\# voter attributes) has 5 instances.}
\label{fig:unsatisfiable}
\end{figure*}

\paragraph{Easier to satisfy constraints when $\sfrac{k}{(|\mathcal{G}|+|\mathcal{P}|)}$ ratio is higher:} The proportion of constraints that are satisfied per instance changed from 100\% to 49\%  (mean=82\%, sd=12\%) as the ratio changed from 1.00 to 0.25. 
This analysis basically captures the committee size to number of constraints ratio. 
Overall, it is easier to have a feasible instance with a larger the committee size ($k$) or a smaller the number of constraints ($|\mathcal{G}|+|\mathcal{P}|$). 

\paragraph{Loss in utility is proportional to the number of attributes:}
Among all feasible instances, the mean ratio of utilities of constrained to unconstrained committee ranged from 0.99 (sd=0.01; for $\mu$=1, $\pi$=0) to 0.43 (sd=0.22; for $\mu$=2, $\pi$=1, the highest number of attributes with feasible committee) for $k$-Borda, from 1.00 (sd=0.01; for $\mu$=0, $\pi$=1) to 0.49 (sd=0.18; for $\mu$=1, $\pi$=2)  for $\beta$-CC, and from 1.00 (sd=0.01; for $\mu$=0, $\pi$=1) to 0.48 (sd=0.15; for $\mu$=2, $\pi$=2) for the Monroe rule. 

\paragraph{Higher group size to lower constraint ratios are easier to satisfy:} 
An important step of our heuristic algorithm was the use of the ``minimum-remaining-value'' heuristic, which helped in the selection of unsatisfied variable. We quantified the need for this heuristic by systematically varying the ratio of group (and population) size to lower constraint (equivalent to $\sfrac{|D_i|}{S_i}$), and found that the utility ratio is the highest when the average of the said ratio across all groups and population is the highest. Also, feasibility of an instance increases with an increase in this ratio. Hence, our heuristic, which prioritizes the \emph{lower} ratio is efficient as it makes sense to first satisfy the groups or populations that are the hardest to satisfy.

\subsubsection{Real Datasets} 

For each dataset, we implemented our model using 3 sets of constraints: constraint 1 only, constraint 2 only, and constraints 1 \& 2. For Eurovision, these were at least one from each ``region'', at least one from each ``language'', and both combined. The ratios of utilities of constrained to unconstrained committees were 0.97, 0.88, and 0.82, respectively. For the UN resolutions, the constraints were at least two from each ``topic'', at least six from ``significant vote'', and both combined.  The ratio of utilities was 0.99 for each of the individual constraints. No feasible committee was found when the constraints were combined. Importantly, our algorithm always terminated in under 102 sec across all instances.


\section{Conclusion and Future Work}
\label{sec:conc}

\paragraph{Conclusion:} 
There is an understanding in social sciences that organizations that answer the call for diversity to avoid legal troubles or to avoid being labeled as ``racists'' may actually create animosity towards racial minorities due to their imposing nature \cite{dobbin2016diversity,bonilla2006racism, ray2019theory}. Similarly, when voters feel that diversity is mandatory and if it comes at the cost of their representation, it can do more harm than good. Hence, it is important to consider \textbf{all actors} of an election, namely candidates and voters, when designing fair algorithms. Doing so in this paper, we first motivated the need for diversity \emph{and} representation constraints in multiwinner elections, and developed a model, \DiReCWD. \DiReCWD, which gives DiRe committees, is also needed because the call for diversity is becoming ubiquitous. However, in the context of elections, \textbf{only diversity can do more harm than good} as the price of diversity may disproportionately be paid more by historically disadvantaged population. Finally, we show the importance to \textbf{delineate the candidate and voter attributes} as we observed that \textbf{diversity does not imply representation and vice versa}, which contrasts the common understanding, and hence, requires further investigation. This is to say that having a female candidate on the committee is different from having a candidate on the committee who is preferred by the female voters, and who themselves may or may not be female. These two are separate but equally important aims that need to be achieved simultaneously.

We 
note that \DiReCWD can satisfy many properties of multiwinner voting rules (e.g., monotonicity) \cite{elkind2017properties} and it can be used as a common framework to solve other problems. As our model was  computationally hard (Tables~\ref{tab:compResults} and \ref{tab:apxTractResults}), we developed a heuristic-based algorithm, which was efficient on tested datasets. Finally, we did an empirical analyses of feasibility, utility traded-off, and efficiency. 

\vspace{-0.075cm}
\paragraph{Future Work:} It remains open to determine \emph{how} the diversity and representation constraints are set to have a ``fair'' outcome. 
The way these constraints are set can lead to unfairness and hence, newer approaches are needed to ensure fairer outcomes. For instance, existing methods that guarantee representation fail when voters are divided into predefined population over one or more attributes. The apportionment method is one way to set the representation constraints, however, it does not account for the cohesiveness of the preferences within a population. Furthermore, this work can also give mathematical guarantees about the existence of DiRe committees. Additionally, just like correlation does not imply causation, diversity does not imply representation and vice versa. Hence, a mathematical framework is needed that can answer the following question: when does diversity imply representation, or when does representation imply diversity? The implications of such a formal framework range from hiring to clinical trials. 


Another open question is determining \emph{what} candidates are used to satisfy the constraints. Assuming that the given constraints are acceptable by everyone, the candidates chosen to satisfy the constraints can lead to unfair outcomes, especially for historically disadvantaged groups. For example, consider a $k$-sized committee election ($k = 4$). The accepted constraints are that the committee should have two male and two female candidates. Next, consider two cases: (Case 1) top two scoring male candidates are selected in the committee and the bottom two scoring female candidates are selected and (Case 2) the top and the bottom scoring male candidates are selected in the committee and the top and bottom scoring female candidates are selected. While both the cases satisfy the given constraints, Case 1 is unfair for female candidates as their top-scoring candidate does not get a seat on the committee while male candidates do get their top two scoring candidates in the committee. This inequality is distributed in Case 2 and it seems naturally ``fairer''.

Next open question pertains to the classification of the complexity of \DiReCWD w.r.t. the \csr $\mathtt{f}$. In this paper, $\mathtt{f}$ could take only two values: it can either be a monotone, submodular but not separable function or a monotone, separable function. We established that determining the winning committee using the former is NP-hard, which was done via establishing hardness of using Chamberlin-Courant rule that uses a positional scoring rule whose first two values of the scoring vector are same (Theorem~\ref{lemma:DiReCWDsubmod}). However, the classification of complexity of determining the winning committee using Chamberlin-Courant rule, Monroe rule, and other submodular but not separable scoring functions w.r.t. different families of positional scoring rules remains open. 

Continuing on the mathematical front, another future direction pertains to the relaxations made to group the candidates and the voters for Corollaries~\ref{lemma:DiReCWDrep} and \ref{cor:DiReCWDsubmod}. More specifically, we showed the hardness persists even when a candidate attribute groups all candidates into one group and a voter attribute groups all voters into one population. These are unrealistic instances as real-world stipulation will require that each candidate attribute partitions the candidates into two or more groups and each voter attribute partitions the voters into two or more populations. Mathematically, Corollaries~\ref{lemma:DiReCWDrep} and \ref{cor:DiReCWDsubmod} will not hold under this stipulation and new proofs are required such that they conform to the stated stipulations.


Finally, the restrictions on voter preferences is another direction for future research. The reduction used to prove Theorem~\ref{lemma:DiReCWDrep01} shows that finding a committee using our model is NP-hard even when each population has one voter. This can be generalized to say that the hardness persists even when each population is completely cohesive \emph{within itself}, which is a very natural assumption to make. For example, all male voters may have the same preferences and all female voters may have the same preferences and one population's preferences are different from the other's. However, given a set of constraints, one can explore as to how does the cohesiveness of voter preferences \emph{across} populations affect the complexity of our model? In addition, it remains open whether the structure of voter preferences across population affects the complexity? Finally, the winning committees of populations can also be cohesive or structured, independent of the structure and cohesiveness of the voter preferences. Hence, even when voters' preferences are not cohesive or structured, the cohesiveness or structure among the population's winning committees may make our model tractable. An immediate consequence of a positive result on this front may be the narrowing of the gap between the intractability of finding a proportionally representative committee (e.g., using Chamberlin-Courant rule) and its tractable instance due to structured preferences. If finding a winning committee is tractable under the assumption of cohesiveness and/or structure of winning committees of the voter populations, then there is hope that finding a proportionally representative committee may also be tractable even under a weaker assumption on the structure of the voter preferences. On the other hand, if we know that the preferences of the voters is cohesive by a factor of $\phi$ : $0 \leq \phi \leq 1$, then we conjecture that there exists a polynomial time approximation algorithm for the \DiReCF problem ($\mu=0$ and $\pi \geq 1$) with approximation ratio at most 
$k-(1-o(1))\frac{k(k-1)\ln\ln g(\phi)}{\ln g(\phi)}$ where $g(\phi)$ is a function that maps the cohesiveness of the preferences $\phi$ to the maximum number of winning committees $W_P$ of each population that a candidate can belong to. 
This approximation ratio improves on the general inapproximability ratio of $k$
(Theorem~\ref{thm:apxtract/01apx}) for \DiReCF when $\phi$ is not known. 

\section*{Acknowledgement}

I am grateful to Julia Stoyanovich for her insights that made me think about DiRe committees. I am thankful to Phokion G. Kolaitis for many helpful discussions that led to the classification of complexity of \DiReCWD. I acknowledge the efforts of high-school students Raisa Bhuiyan and Rachel Rose in collecting the real-world datasets, and of Théo Delemazure for comments on empirical analysis. Finally, I thank anonymous reviewers for their comments on an earlier version of this paper.

\bibliographystyle{unsrt}  
\bibliography{references_original}  






\end{document}